\documentclass[10pt]{article}
\usepackage[margin=1in]{geometry}
\usepackage{amsthm,amsfonts,amssymb,amsmath,comment,slashed,mathtools}
\usepackage[T1]{fontenc} 
\usepackage[utf8]{inputenc}

\usepackage[affil-it]{authblk} 
\usepackage[hidelinks]{hyperref}    
\usepackage{graphicx} 
\usepackage{epstopdf} 
\usepackage{thm-restate} 
\usepackage{xargs}          
\usepackage{mathrsfs} 
\usepackage{xcolor}  
\usepackage{here} 
\usepackage[capitalize]{cleveref} 
\usepackage{tikz-cd} 
\usepackage[backend=biber,bibencoding=ascii,giveninits=true,doi=false,isbn=false,url=false,sorting=nyt]{biblatex} 
\renewbibmacro{in:}{} 
\bibliography{scattering_interior.bib}
\usepackage[colorinlistoftodos,prependcaption,textsize=tiny]{todonotes}
\newcommandx{\unsure}[2][1=]{\todo[linecolor=red,backgroundcolor=red!25,bordercolor=red,#1]{#2}}
\newcommandx{\change}[2][1=]{\todo[linecolor=blue,backgroundcolor=blue!25,bordercolor=blue,#1]{#2}}
\newcommandx{\info}[2][1=]{\todo[linecolor=OliveGreen,backgroundcolor=OliveGreen!25,bordercolor=OliveGreen,#1]{#2}}
\newcommandx{\improvement}[2][1=]{\todo[linecolor=Plum,backgroundcolor=Plum!25,bordercolor=Plum,#1]{#2}}
\newcommandx{\thiswillnotshow}[2][1=]{\todo[disable,#1]{#2}}
\usetikzlibrary{decorations.pathmorphing}

\theoremstyle{plain}
\newtheorem{definition}{Definition} 
\newtheorem{prop}{Proposition}
\newtheorem{lemma}{Lemma}
\newtheorem{cor}{Corollary}
\newtheorem{rmk}{Remark}
\declaretheorem[name=Theorem]{theorem}
\renewcommand{\d}{\mathrm{d}}
\newcommand{\Hp}{\mathcal{H}}
\newcommand{\Ho}{\mathcal{H}}
\newcommand{\Ch}{\mathcal{CH}}

\makeatletter
\renewcommand{\paragraph}[1]{%
	\par 
	\addvspace{\medskipamount}
	\textit{#1\@addpunct{.}}\enspace\ignorespaces
}
\makeatother
\numberwithin{equation}{section}
\numberwithin{prop}{section}
\numberwithin{rmk}{section}
\numberwithin{lemma}{section}
\numberwithin{definition}{section}
\numberwithin{cor}{section}
\title{A scattering theory for linear waves \\ \makebox[0pt]{ on~the~interior~of~Reissner--Nordström~black~holes}}
\author[1]{Christoph Kehle\thanks{c.kehle@maths.cam.ac.uk}}
\author[2]{Yakov Shlapentokh-Rothman\thanks{yshlapen@math.princeton.edu}}
\affil[1]{\small  Department of Pure Mathematics and Mathematical Statistics, University~of~Cambridge,~Wilberforce~Road,~Cambridge CB3 0WB, United Kingdom \vskip.1pc \  }
\affil[2]{\small  Department of Mathematics, Princeton University, Washington~Road,~Princeton,~NJ~08544,~United~States~of~America \vskip.1pc \ }

\date{December 14, 2018}

\begin{document}
\maketitle
\begin{abstract}
	We develop a scattering theory for the  linear wave equation $\Box_g \psi = 0 $ on the interior of Reissner--Nordstr\"om black holes, connecting the fixed frequency picture to the physical space picture. 
	Our main result gives the existence, uniqueness and asymptotic completeness of finite energy scattering states. The past and future scattering states are represented as suitable traces of the solution $\psi$ on the bifurcate event and Cauchy horizons. The heart of the proof is to show that after separation of variables one has \emph{uniform boundedness} of the reflection and transmission coefficients of the resulting radial o.d.e.\ over all frequencies $\omega$ and~$\ell$. This is non-trivial because the natural $T$ conservation law is sign-indefinite in the black hole interior. In the physical space picture, our results imply that the Cauchy evolution from the event horizon to the Cauchy horizon is a Hilbert space isomorphism, where the past (resp.\ future) Hilbert space is defined by the finiteness of the degenerate $T$ energy fluxes on both components of the event (resp.\ Cauchy) horizon.
	Finally, we prove that, in contrast to the above, for a generic set of cosmological constants $\Lambda$,
	there is no analogous finite $T$ energy scattering theory for either the linear wave equation or the Klein--Gordon equation with conformal mass on the (anti-) de~Sitter--Reissner--Nordstr\"om  interior.
\end{abstract}

\thispagestyle{empty}
\newpage
\tableofcontents
\thispagestyle{empty}
\newpage
\section{Introduction}
\setcounter{page}{1}
One of the most stunning predictions of general relativity is the formation of \textbf{\textit{black holes}}, defined by the property that information cannot propagate from their interior region to outside far-away observers. 
Fortunately, we can count ourselves among the latter; nevertheless, if a group of physicists were so courageous as to cross the \emph{event horizon} and enter a black hole, they could still very well perform experiments and compare the outcomes amongst themselves. Indeed, the problem of determining the fate  of these black hole explorers (and their laboratories) has led to some of the most central conceptual puzzles in gravitational physics.

In view of the above, there has been 
 a lot of recent activity analyzing the Cauchy problem on black hole interiors, e.g.\ \cite{franzen,franzen2016boundedness,sbierski2014initial, luk2016instability,interiorschwarzschild}. However, for certain physical processes it is more natural to consider the \emph{\textbf{scattering problem}}
 (see \cite{futterman1988scattering} for scattering on the exterior of black holes).
 With this paper, we initiate the mathematical study of the finite energy scattering problem on black hole interiors.
 Specifically, we will consider solutions of the wave equation on what can be viewed as the most elementary interior, that of Reissner--Nordstr\"om. The Reissner--Nordstr\"om metrics constitute a family of spacetimes, parametrized by mass $M$ and charge $Q$, which satisfy the Einstein--Maxwell system in spherical symmetry \cite{reissner1916eigengravitation,nordstrom1918energy} and admit an additional Killing vector field $T$. For vanishing charge $Q=0$, the family reduces to Schwarzschild. We shall moreover restrict in the following to the subextremal case where $0<|Q|<M$.
 In addition to the bifurcate event horizon, these black hole interiors then admit an additional bifurcate inner horizon, the so-called \textit{Cauchy horizon}. Our past and future scattering states will be defined as suitable traces of the solution on the bifurcate event horizon and bifurcate Cauchy horizon, respectively, restricted to have finite $T$ energy flux on each component of the horizons.

In the rest of the introduction we will state our main results for the scattering problem on the interior of Reissner--Nordstr\"om  ({Theorems~\ref{thm:forwardevolution} -- \ref{thm:c1instab}}), 
relate them to existing literature in fixed frequency scattering, and draw links to various recent results in the interior and exterior of black holes. 
Finally, we will see that the existence of a bounded scattering map for the wave equation on Reissner--Nordstr\"om turns out to be a very fragile property; we shall show that there does \underline{not} exist an analogous scattering theory in the presence of a cosmological constant ({\cref{thm:cosmological}}) or Klein--Gordon mass ({\cref{thm:kleingordon}}). 
\paragraph{\textbf{The scattering problem on Reissner--Nordstr\"om interior}}
 In this paper, we will establish a scattering theory for \underline{finite energy} solutions of the linear wave equation, 
 \begin{align}\label{eq:linearwave}
\Box_g \psi =0,
\end{align} on the interior of a Reissner--Nordstr\"om black hole, from the bifurcate event horizon $\mathcal H = \Ho_A \cup \Ho_B\cup \mathcal{B}_-$  to the bifurcate Cauchy horizon $\mathcal{CH} = \Ch_A \cup \Ch_B\cup \mathcal{B}_+$, as depicted in \cref{fig:penroseinterior}.
 \begin{figure}[ht]\centering
 	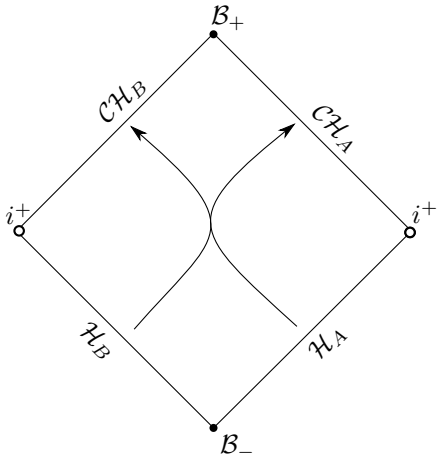
 	\caption{Penrose diagram of the interior of the Reissner--Nordstr\"om black hole and visualization of the scattering map.}
 	\label{fig:penroseinterior}
 \end{figure}  The first main result of our paper is \textbf{\cref{thm:forwardevolution}} (see \cref{sec:existencescatteringmap}) in which we will show \textit{existence}, \textit{uniqueness} and \textit{asymptotic completeness} of finite energy scattering states.  In this context, \textit{existence} and \textit{uniqueness} mean that for given finite energy data $\psi_0$ on the event horizon $\Hp$, there exist unique finite energy data on the Cauchy horizon $\Ch$ arising from $\psi_0$ as the evolution of \eqref{eq:linearwave}. With \textit{asymptotic completeness} we denote the property that all finite energy data on the Cauchy horizon $\Ch$ can indeed be achieved from finite energy data on the event horizon $\Hp$. This provides a way to construct solutions with desired asymptotic properties which is a necessary first step to properly understand quantum theories in the interior of a Reissner--Nordstr\"om black hole (cf.\ \cite{wald1994quantum,hafnerhawk,drouothawk}). The energy spaces on the event and Cauchy horizon are associated to the Killing field and generator of the time translation $T$. Indeed, $T$ is null on the horizons and, in particular, is the generator of the event and Cauchy horizon $\mathcal{H}$ and $\Ch$. Because of the sign-indefiniteness of the energy flux of the vector field $T$ on the bifurcate event (resp.\ Cauchy) horizon (see already \eqref{eq:t-r=1}), we define our energy space by requiring the finiteness of the $T$ energy on both components separately of the event (resp.\ Cauchy) horizon. These define Hilbert spaces with respect to which the scattering map is proven to be bounded.

Finally, it is instructive to draw a comparison between the interior of Reissner--Nordstr\"om and the interior of Schwarzschild ($Q=0$). As opposed to Reissner--Nordstr\"om discussed above, the Schwarzschild interior terminates at a singular boundary at which solutions to \eqref{eq:linearwave} generically blow-up (see\ \cite{interiorschwarzschild}). In contrast, the non-singular and, moreover, Killing, Cauchy horizons (see\ \cref{fig:penroseinterior}) of Reissner--Nordstr\"om immediately yield natural Hilbert spaces of finite energy data to consider. In view of this, Reissner--Nordstr\"om with $Q\neq 0$ can be considered the most elementary interior on which to study the scattering problem. Furthermore, in view of the recent work \cite{dafermos2017interior}, we have that  the causal structure of Reissner--Nordstr\"om is stable in a weak sense (see the discussion below about related works in the interior).

\paragraph{\textbf{Fixed frequency scattering}}
It is well known that the wave equation \eqref{eq:linearwave} on Reissner--Nordstr\"om spacetime allows separation of variables which reduces it to the radial o.d.e. 
\begin{align}\label{eq:radialode1}
	u^{\prime \prime} - V_\ell u + \omega^2 u =0,
\end{align} with potential $V_\ell$ (see already \eqref{eq:potential}), where $\omega\in \mathbb R $ is the time frequency and $\ell \in \mathbb{N}_0$ is the angular parameter.  Indeed, most of the existing literature concerning scattering of waves in the interior of Reissner--Nordstr\"om mainly considers fixed frequency solutions, e.g.\  \cite{mcnamara1978behaviour,mcnamara1978instability,hartle1982crossing,gursel1979evolution,matzner1979instability,gursel1979final,MR686719}. For a purely \emph{incoming} (i.e.\ supported only on $\mathcal{H}_A$) fixed frequency solution with parameters $(\omega,\ell)$, we can associate transmission and reflection coefficients $\mathfrak T(\omega,\ell)$ and $\mathfrak R(\omega,\ell)$. The transmission coefficient $\mathfrak T (\omega,\ell)$ measures what proportion of the incoming wave is transmitted to $\Ch_B$, whereas the reflection coefficient specifies the proportion reflected to  $\Ch_A$. An essential step to go from fixed frequency scattering to physical space scattering is to prove \textbf{\emph{uniform boundedness}} of $\mathfrak T(\omega,\ell)$ and $\mathfrak R(\omega,\ell)$. This is non-trivial in view of the discussion of the energy identity \eqref{eq:t-r=1} below.  
In this paper, we indeed obtain this uniform bound in \textbf{\cref{thm:boundednesstrans}}  (see \cref{subsec:scatteringcoefficients}).
 In particular, the regime $\omega \to 0, \ell \to \infty$ is the most involved frequency range to prove uniform boundedness. As we shall see, the proof relies on an explicit computation at $\omega =0$ (see \cite{gursel1979evolution}) together with a careful analysis of special functions and perturbations thereof.

The uniform boundedness of the scattering coefficients is the main ingredient to prove the boundedness of the scattering map in \cref{thm:forwardevolution}. Moreover, it allows us to connect the separated picture to the physical space picture by means of a Fourier representation formula. This is stated as \textbf{\cref{thm:fouriertophysical}} (see \cref{subsec:connfourierphysical}). 
A somewhat surprising, direct consequence of the Fourier representation of the scattered data on the Cauchy horizon is that purely incoming compactly supported data lead to a solution which vanishes at the future bifurcation sphere $\mathcal{B}_+$. This is moreover shown to be a necessary condition for the existence of a bounded scattering map (\textbf{\cref{cor:vanishingbifurcation}}).

\paragraph{\textbf{Comparison to scattering on the exterior of black holes}}On the exterior of black holes, the scattering problem has been studied more extensively; see the pioneering works  \cite{dimock1985scattering,dimock1987classical,dimock1986classical2,bachelot1991gravitational,bachelot1994asymptotic}, the book \cite{futterman1988scattering} and related results on conformal scattering in \cite{conf1,conf2,conf3,conf4}. Note that for the exterior of a Schwarzschild or Reissner--Nordstr\"om black hole, the uniform boundedness of the scattering coefficients or equivalently, the boundedness of the scattering map, can be viewed a posteriori\footnote{Note that proving \eqref{eq:Tid1} requires first establishing some form of qualitative decay towards $i^+$ and $i^-$.} as a consequence of the global $T$ energy identity
\begin{align}\label{eq:Tid1}
	\int_{\mathcal{H}^-} |T\psi|^2 + \int_{\mathcal{I}^-} |T\psi|^2 = \int_{\mathcal{H}^+} |T\psi|^2 + \int_{\mathcal{I}^+} |T\psi|^2.
\end{align} Considering only incoming radiation from $\mathcal{I}^-$, this statement translates into $|\mathfrak R|^2 + |\mathfrak T|^2 = 1$ for the reflection coefficient $ \mathfrak R$ and transmission coefficients $\mathfrak T$.  In the interior, however, due to the different orientations of the $T$ vector field on the horizons (cf. \cref{fig:compare}), boundedness of the scattering map is not at all manifest.  
 \begin{figure}[ht]
	\centering
	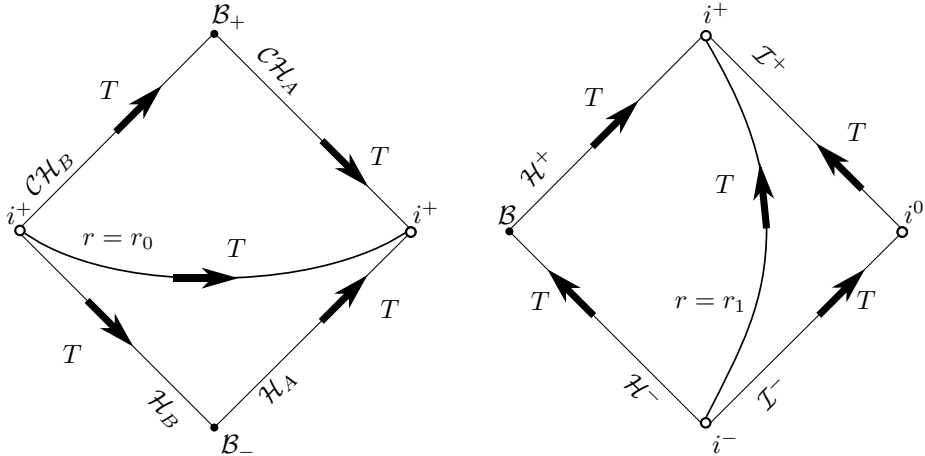
	\caption{Interior of Reissner--Nordstr\"om (left) and exterior of Schwarzschild or Reissner--Nordstr\"om (right).
		\newline In both diagrams the arrows denote the direction of the $T$ Killing vector field. Note that we have the identifications $\mathcal{H}_A = \mathcal{H}^+ $ and $ \mathcal{B}_- = \mathcal{B}$.}
	\label{fig:compare}
\end{figure}
In particular, the global $T$ energy identity on the interior of a Reissner--Nordstr\"om black hole reads
\begin{align}\label{eq:t-r=1}
	\int_{\Hp_A} |T\psi|^2 - \int_{\Hp_B} |T\psi|^2 = \int_{\Ch_B} |T\psi|^2 -\int_{\Ch_A} |T\psi|^2 
\end{align}
from which we cannot deduce boundedness of the scattering map even a posteriori. (Indeed, identity \eqref{eq:t-r=1} corresponds only to the ``pseudo-unitarity'' statement of \cref{thm:forwardevolution}.) Again, considering only ingoing radiation from $\Hp_A$, this translates to \begin{align}\label{eq:t-r=12}|\mathfrak T(\omega,\ell)|^2 - |\mathfrak R(\omega,\ell)|^2 = 1\end{align} for the reflection coefficient $\mathfrak R$ and the transmission coefficient $\mathfrak T$. Hence, while for fixed $|\omega| > 0$ and $\ell$, it is straightforward to show that $\mathfrak{T}$ and $\mathfrak{R}$ are finite, there is no a priori obvious obstruction from \eqref{eq:t-r=12} for these scattering coefficients to blow up in the limits $\omega \to 0,\pm\infty$ and $\ell \to \infty$. 

Moreover, note that in the exterior, the Killing field $T$ is timelike, so the radial o.d.e.\ \eqref{eq:radialode1} should be considered as an equation for a fixed time frequency wave on a constant time slice. In the interior, however, the Killing field $T$ is spacelike so the radial o.d.e.\ \eqref{eq:radialode1} is rather an evolution equation for a constant spatial frequency. 

The Schwarzschild family can be viewed as a special case ($a=0$) of the two parameter Kerr family, describing rotating black holes with mass parameter $M$ and rotation parameter $a$ \cite{MR0156674}.\footnote{Both Kerr and Reissner--Nordström can be viewed as special cases of the Kerr--Newman spacetime. For decay results on Kerr--Newman see \cite{civin2015stability}.} On the exterior of Kerr many other difficulties arise:~superradiance, intricate trapping, presence of ergoregion, etc.\ \cite{MR3488738}. Nevertheless, using the decay results in \cite{MR3488738}, a definitive physical space scattering theory for Kerr black holes has been established in \cite{dafermos2014scattering} (see also the earlier \cite{deSitterScatter}). The proof on the exterior of Kerr involved first establishing a scattering map from past null infinity $\mathcal{I}^-$ to a constant time slice $\Sigma$ and then concatenating that map with a second scattering map from $\Sigma$ to the future event horizon $\Hp^+$ and future null infinity $\mathcal{I}^+$. In the interior, however, we will directly show the existence of a ``global'' scattering map from the event horizon $\mathcal H$ to the Cauchy horizon $\Ch$. Indeed, due to blue-shift instabilities (see \cite{dafermos2017time}), we do not expect that the analogous concatenation of scattering maps (event horizon $\Hp$ to spacelike hypersurface $\Sigma$ and then from $\Sigma$ to the Cauchy horizon $\Ch$) as in the Kerr exterior, to be bounded in the interior of Reissner--Nordstr\"om.
 
 \paragraph{\textbf{Injectivity of the reflection map and blue-shift instabilities}} We can also conclude from our work that there is always non-vanishing reflection to the Cauchy horizon $\Ch_A$ arising from  non-vanishing purely ingoing radiation at $\Hp_A$. This follows from the fact that in the separated picture and for fixed $\ell$, the reflection coefficient $\mathfrak R (\omega,\ell)$ can be analytically continued to the strip $|\operatorname{Im}(\omega)| < \kappa_+$ and hence, only vanishes on a discrete set of points on the real axis. This is shown in \textbf{\cref{thm:nonvanishingreflection}} (see \cref{subsec:reflection}). 
   
 We will also deduce from the Fourier representation of the scattered data on the Cauchy horizon $\Ch$, and a suitable meromorphic continuation of the transmission coefficient, that there exist purely incoming compactly supported data on the event horizon $\mathcal{H}$ leading to solutions which fail to be $C^1$ on the Cauchy horizon $\mathcal{CH}$. This $C^1$-blow-up of linear waves puts on a completely rigorous footing the pioneering work of Chandrasekhar and Hartle \cite{hartle1982crossing}. We state this as \textbf{\cref{thm:c1instab}} (see \cref{sec:c1instab}).
 
For generic solutions arising from localized data on an asymptotically flat hypersurface, one expects a stronger instability, namely, non-degenerate energy blow-up at the Cauchy horizon $\Ch$. Such non-degenerate energy blow-up was proven in \cite{lukohchblowup} for generic compactly supported data on an asymptotically flat Cauchy hypersurface. Later, for the more difficult Kerr interior, non-degenerate energy blow-up was proven in \cite{luk2016instability} assuming certain polynomial lower bounds on the tail of incoming data on the event horizon $\Hp$ and in~\cite{dafermos2017time} for solutions arising from generic initial data along past null infinity $\mathcal{I}^-$ with polynomial tails.

  Finally, we mention the forthcoming work \cite{yakov} which studies the problem of non-degenerate energy blow-up from a scattering theory perspective and also uses the non-triviality of reflection to establish results related to mass inflation for the spherically symmetric Einstein--Maxwell--scalar field system (cf.\ \cite{luk2017strong,luk2017strong2}). 

  \paragraph{\textbf{Related results on the interior}}There has been a lot of recent progress studying the interior of black holes. In particular, new insights were gained concerning the stability of the Cauchy horizon and the strong cosmic censorship conjecture. 
 
 For the Cauchy problem for \eqref{eq:linearwave} on the interior of both a fixed Kerr and a Reissner--Nordstr\"om black hole, the works \cite{franzen2016boundedness,franzen,hintzinterior} establish uniform boundedness (in $L^\infty$) and continuity up to and including the Cauchy horizon for solutions arising from smooth and compactly supported data on an asymptotically flat spacelike hypersurface. Such data in particular give rise to solutions with polynomial decay along the event horizon.

   In contrast, for the scattering problem considered in the present paper, we are required to work with spaces which are symmetric with respect to the event and Cauchy horizons. This naturally leads to the rougher class of finite $T$ energy data in the statement of \cref{thm:forwardevolution}. Note that for such data on the Cauchy horizon, continuity or boundedness (in $L^\infty$) does \underline{not} necessarily hold true. 
 
 Turning finally to the full nonlinear dynamics of the Einstein equations, it is shown in \cite{dafermos2017interior} that the Kerr Cauchy horizon is $C^0$-stable. Thus, the existence of a Cauchy horizon, a very natural setting parameterizing scattering data in the interior, is not a pure artifact of symmetry but rather a stable property at least in a weak sense. On the other hand, in \cite{luk2017strong,luk2017strong2,VandeMoortel2018} it is proven that for a suitable Einstein--matter system under spherical symmetry, the Cauchy horizon, while $C^0$-stable, is generically $C^2$-unstable.
  Finally, we mention that for the Schwarzschild black hole ($Q=0$), which does not admit a Cauchy horizon, it is shown in \cite{interiorschwarzschild} that solutions to \eqref{eq:linearwave} generically blow up at the spacelike singularity $\{r=0\}$. 
 
 \paragraph{\textbf{Breakdown of $T$ energy scattering for $\Lambda \neq 0$ or $\mu\neq 0$}} If a cosmological constant $\Lambda \in \mathbb R$ is added to the Einstein--Maxwell system, we can consider the analogous (anti-) de Sitter--Reissner--Nordstr\"om family of solutions whose interiors have the same Penrose diagram as depicted in \cref{fig:penroseinterior}. For very slowly rotating Kerr--de Sitter and Reissner--Nordstr\"om--de Sitter spacetimes, boundedness, continuity, and regularity up to and including the Cauchy horizon has been shown for solutions to  \eqref{eq:linearwave} arising from smooth and compactly supported data on a Cauchy hypersurface \cite{hintzvasyinterior}. However, remarkably, there is no analogous $T$ energy scattering theory for either the linear wave equation \eqref{eq:linearwave} or the Klein--Gordon equation with conformal mass. This is the statement of~\textbf{\cref{thm:cosmological}} (see \cref{subsec:nonex}). The reason for this failure is the unboundedness of the transmission coefficient $\mathfrak T$ and reflection coefficients $\mathfrak R$ in the vanishing frequency limit $\omega \to 0$. To be more precise, we will prove that there does not exist a $T$ energy scattering theory of the wave or Klein--Gordon equation in the interior of a (anti-) de Sitter--Reissner--Nordstr\"om black hole for generic subextremal black hole parameters $(M,Q,\Lambda)$. In particular, for fixed $0<|Q|<M$, there is an $\epsilon >0$ such that there does not exist a $T$ energy scattering theory for all $0\neq |\Lambda| < \epsilon$.  
 
 Similarly, we prove in \textbf{\cref{thm:kleingordon}} (see \cref{subsec:notkleingordon}) that there does not exist a $T$ energy scattering theory for the Klein--Gordon equation $\Box_g \psi - \mu \psi =0$ on the Reissner--Nordstr\"om interior for a generic set of masses $\mu$. This is in contrast to the exterior, where $T$ energy scattering theories were established for massive fields in \cite{bachelot1994asymptotic,MR2016993}.
 
 It remains an open problem to formulate an appropriate scattering theory in the cosmological setting and for the Klein--Gordon equation as well as for the interior of Kerr.
 
\paragraph{\textbf{Outline}}
This paper is organized as follows. In \cref{prelims}, we shall set up the spacetime, introduce the relevant energy spaces, and define the scattering coefficients in the separated picture. In \cref{sec:mainthms} we state the main results of this paper, Theorems \ref{thm:forwardevolution} -- \ref{thm:kleingordon}.
Section~\ref{sec:radial} is devoted to the proof of \cref{thm:boundednesstrans}.
Then, the statement of \cref{thm:boundednesstrans} allows us to prove \cref{thm:forwardevolution} in \cref{sec:mainthm}.  
Finally, in the last two sections are show our non-existence results: In \cref{sec:cosmo}, we prove \cref{thm:cosmological} and in \cref{sec:kleingordonequation}, we give the proof of \cref{thm:kleingordon}.

\paragraph{\textbf{Acknowledgement.}} The authors would like express their gratitude to Mihalis Dafermos for many valuable discussions and  helpful remarks. The authors also thank Igor Rodnianski, Jonathan Luk, and Sung-Jin Oh for useful conversations. CK acknowledges support from the EPSRC and thanks Princeton University for hosting him as a VSRC. 
YS acknowledges support from the NSF Postdoctoral Research Fellowship under award no.\ 1502569.
\section{Preliminaries}
In this section we will define the background differentiable structure and metric for the Reissner--Nordström spacetime and introduce some convenient coordinate systems.
\label{prelims}
\subsection{Interior of the subextremal Reissner--Nordstr\"om black hole} 
The global geometry of Reissner--Nordström was first described in \cite{MR0128960}.  For completeness, we will explicitly construct in this section the coordinates for the region considered.
We start, in \cref{sec:interiorwoboundary}, by defining a Lorentzian manifold  corresponding to the interior of the Reissner--Nordström black hole without the horizons. Then, in \cref{sec:interiorwboundary}, we will attach the boundaries which will correspond to the event horizon and Cauchy horizon. 
\subsubsection{The interior without boundary}
\label{sec:interiorwoboundary}
We will now give an explicit description of the differential structure and metric. The Reissner--Nordström solutions \cite{reissner1916eigengravitation,nordstrom1918energy} are a  two-parameter family of spherically symmetric spacetimes with mass parameter $M\in \mathbb{R}$ and electric charge parameter $Q\in \mathbb R$ solving the Einstein--Maxwell system
\begin{align}
 	\label{eq:einsteinmaxwell}
	&{Ric}_{\mu\nu} - \frac{1}{2}g_{\mu\nu}{R} = 8 \pi T_{\mu\nu} := 8 \pi \left( \frac{1}{4\pi} \left( F_{\mu}^{~\lambda} F_{\lambda\nu}  - \frac{1}{4} g_{\mu\nu} F_{\lambda \rho} F^{\lambda \rho} \right)\right), 
	\\&	\nabla^\mu F_{\mu\nu} = 0 , \nabla_{[\mu} F_{\nu\lambda]} = 0.
\nonumber
\end{align}
 In this paper, we consider the subextremal family of black holes with parameter range $0<|Q|<M$. Define the manifold $\mathcal M$ by
\begin{align}
\mathcal M = \mathbb{R} \times (r_-,r_+) \times \mathbb{S}^2,
\end{align}
where $r_\pm = M \pm \sqrt{M^2 - Q^2} >0$. The manifold is parametrized by the standard coordinates $t \in \mathbb{R}$, $r\in (r_-, r_+)$,  and a choice of spherical coordinates $(\theta,\phi)$ on the sphere $\mathbb{S}^2$. We denote the global coordinate vector field $\partial_t$ by $T$:\begin{align}T := \frac{\partial}{\partial t}.\label{def:T}\end{align}
Using the above coordinates, we equip $\mathcal{M}$ with the Lorentzian metric
\begin{align}
g_{Q,M} = - \left(1-\frac{2M}{r} + \frac{Q^2}{r^2}\right) \d t \otimes \d t + \left(1-\frac{2M}{r} + \frac{Q^2}{r^2}\right)^{-1} \d r \otimes \d r + r^2 \slashed g_{\mathbb{S}^2},
\end{align}
where $\slashed g_{\mathbb{S}^2}$ is the round metric on the 2-sphere. We also define the quantities
\begin{align}
\Delta := r^2 - 2Mr + Q^2 = (r - r_+) (r-r_-)
\text{ and }h:= \frac{\Delta}{r^2}. \label{eq:h}
\end{align} 
Furthermore, define $r_\ast$ by
\begin{align}
\d r_\ast := \frac{r^2}{\Delta} \d r,
\end{align}
where we choose $r_\ast (\frac{r_+ + r_-}{2}) = 0$ for definiteness. Thus,
\begin{align}
r_\ast (r) = r + \frac{1}{2\kappa_+} \log|r - r_+| + \frac{1}{2\kappa_-} \log| r-r_-| + C \label{defn:rstar}
\end{align}
for a constant $C$ only depending on the black hole parameters and 
\begin{align}\label{eq:surface}
\kappa_\pm = \frac{r_\pm - r_\mp}{2r_\pm^2}.
\end{align}

We also introduce null coordinates
\begin{align}\label{eq:nullu}
v=  r_\ast + t \text{ and } u = r_\ast - t
\end{align}
on $\mathcal{M}$. The Penrose diagram of $\mathcal M$ is depicted in \cref{fig:interior_wo_boundary}.
\begin{figure}[ht]
	\centering
	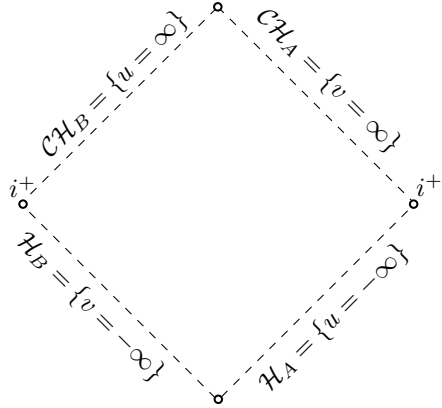
	\caption{Penrose diagram of $\mathcal{M}$; formally we have denoted the boundary (not part of the manifold) by $\mathcal{H} = \mathcal{H}_A \cup \mathcal{H}_B$ and $\mathcal{CH}=\Ch_A \cup \Ch_B$. }
	\label{fig:interior_wo_boundary}
\end{figure}
\subsubsection{Attaching the event and Cauchy horizon}
\label{sec:interiorwboundary}
Now, we will attach the Cauchy and event horizon to the manifold.  The Cauchy horizon characterizes the future boundary up to which the spacetime is uniquely determined as a solution to the Einstein--Maxwell system arising from data on the event horizon. Although the metric is smoothly extendible beyond the Cauchy horizon, any such extension fails to be uniquely determined from initial data, leading to a severe failure of determinism. 

Attaching the event and Cauchy horizon gives rise to a manifold with corners.
To do so, we first define the following two pairs of null coordinates. 

Let $U_{\Hp}\colon \mathbb{R} \to (0,\infty)$ and $V_{\Hp}\colon \mathbb{R} \to (0,\infty)$ be smooth and strictly increasing functions such that
\begin{itemize}
	\item $	U_{\Hp}(u) =u $ for $u\geq 1$, $V_\Hp (v) = v $ for $v\geq 1$,
	\item $ U_\Hp(u) \to 0$ as $u \to -\infty$ , $V_\Hp (v) \to 0 $ as $v \to -\infty$,
	\item there exists a $u_+ \leq 1$ such that $\frac{\d U_\Hp }{\d u} = \exp(\kappa_+u)	$ for $u\leq u_+$,
	\item there exists a $v_+ \leq 1$ such that $\frac{\d V_\Hp }{\d v} = \exp(\kappa_+v)	$ for $v\leq v_+$.
\end{itemize}
This defines -- in mild abuse of notation -- the null coordinates $U_\Hp \colon \mathcal{M} \to (0,\infty)$ via $U_\Hp (u)$ and $V_\Hp \colon \mathcal{M} \to (0,\infty)$ via $V_\Hp(v)$, where $u,v$ are the null coordinates defined in \eqref{eq:nullu}.

Similarly, let $U_{\Ch}\colon \mathbb{R}\to (-\infty,0)$ and $V_{\Ch}\colon \mathbb R \to (-\infty,0)$ be smooth and strictly increasing functions such that
\begin{itemize}
	\item $	U_{\Ch}(u) =u $ for $u\leq -1$, $V_\Ch (v) = v $ for $v\leq -1$,
	\item $ U_\Ch(u) \to 0$ as $u \to \infty$ , $V_\Ch (v) \to 0 $ as $v \to \infty$,
	\item there exists a $u_+ \geq -1$ such that $\frac{\d U_\Ch }{\d u} = \exp(\kappa_-u)	$ for $u\geq u_+$,
	\item there exists a $v_+ \geq -1$ such that $\frac{\d V_\Ch }{\d v} = \exp(\kappa_-v)	$ for $v\geq v_+$.
\end{itemize}
As above, this defines null coordinates $U_\Ch : \mathcal{M} \to (0,\infty)$ via $U_\Ch (u)$ and $V_\Ch\colon \mathcal{M} \to (0,\infty)$ via $V_\Ch(v)$, where $u,v$ are the null coordinates defined in \eqref{eq:nullu}.

To define the event horizon, we consider the coordinate chart $(U_\Hp, V_\Hp,\theta,\phi)$. We now define the event horizon without the bifurcation sphere as the union 
\begin{align}
\mathcal{H}_0 := \mathcal{H}_A \cup \mathcal{H}_B,
\end{align} 
where 
\begin{align}\mathcal{H}_A :=\{ U_\Hp =0 \}\times (0,\infty)\times \mathbb{S}^2 \text{ and }\mathcal{H}_B :=(0,\infty)\times \{V_\Hp = 0 \}\times \mathbb{S}^2.
\end{align}
Analogously, we also define the Cauchy horizon without the bifurcation sphere in the coordinate chart $(U_\Ch, V_\Ch, \theta, \phi)$ as the union
\begin{align}
\Ch_0 := \Ch_A \cup \Ch_B,
\end{align}
where 
\begin{align}\Ch_A :=(0,\infty)\times \{V_\Ch = 0 \}\times \mathbb{S}^2  \text{ and }\Ch_B :=\{ U_\Ch =0 \}\times (0,\infty)\times \mathbb{S}^2.
\end{align}

Then, we define the interior of the Reissner--Nordstr\"om spacetime without the bifurcation sphere as the manifold with boundary \begin{align}
\tilde{\mathcal{M}}:=\mathcal{M} \cup \mathcal{H} \cup \Ch.
\end{align}
The Lorentzian metric on $\mathcal{M}$ extends smoothly to $\tilde{\mathcal{M}}$. In particular, the boundary of $\tilde{\mathcal{M}}$ consists of four disconnected null hypersurfaces. In \cref{fig:penrosemrn} we have depicted its Penrose diagram.
\begin{figure}[ht]
	\centering
	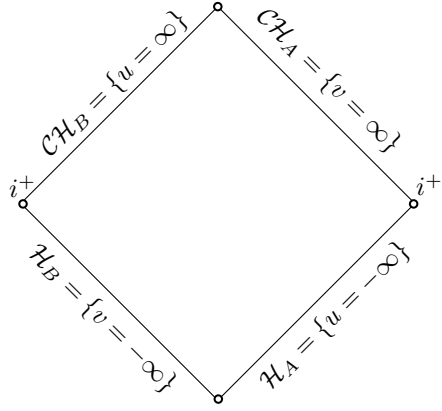
	\caption{Penrose diagram of $\tilde{\mathcal{M}}$.}
	\label{fig:penrosemrn}
\end{figure}
In mild abuse of notation we shall also use the coordinate systems
\begin{align}
&(U_\Hp, v, \theta,\phi) \text{ on }\mathcal M\cup \Hp_A\label{eq:ha},\\
&(u, V_\Hp, \theta,\phi) \text{ on }\mathcal M\cup \Hp_B\label{eq:hb},\\
&(u, V_\Ch, \theta,\phi) \text{ on }\mathcal M\cup \Ch_A\label{eq:ca},\\
&(U_\Ch, v, \theta,\phi) \text{ on }\mathcal M\cup \Ch_B.\label{eq:cb}
\end{align} In particular, we can write \begin{align}
&\Hp_A = \{ U_\Hp = 0 \}\times\{v \in \mathbb{R}\}\times \mathbb{S}^2,\\
&\Hp_B = \{ u \in \mathbb{R} \}\times\{V_\Hp = 0 \}\times \mathbb{S}^2,\\
&\Ch_A =  \{ u \in \mathbb{R} \}\times\{V_\Ch = 0 \}\times \mathbb{S}^2,\\
&\Ch_B =  \{ U_\Ch = 0 \}\times\{v \in \mathbb{R}\}\times \mathbb{S}^2.
\end{align}

Note also that the vector field $T$, initially defined on $\mathcal{M}$ in \eqref{def:T}, extends to a smooth vector field on  $\tilde{\mathcal{M}}$ with \begin{align}T\restriction_{\Hp_A} = \frac{\partial}{\partial v}\restriction_{\Hp_A},\end{align} where $\frac{\partial}{\partial v}$ is the coordinate derivative with respect to local chart defined in \eqref{eq:ha}.
Similarly, we have 
\begin{align}
&	T\restriction_{\Hp_B} = -\frac{\partial}{\partial u}\restriction_{\Hp_B} \text{ w.r.t. to the local chart \eqref{eq:hb}},\\
&T\restriction_{\Ch_A} = -\frac{\partial}{\partial u}\restriction_{\Ch_A} \text{ w.r.t. to the local chart \eqref{eq:ca}},\\
&T\restriction_{\Ch_B} = \frac{\partial}{\partial v}\restriction_{\Ch_B} \text{ w.r.t. to the local chart \eqref{eq:cb}}.
\end{align}
Note that $T$ is a Killing null generator of the Killing horizons $\Hp_A, \Hp_B, \Ch_A$, and $\Ch_B$. Recall also that $\nabla_T T \restriction_{\Ch} = \kappa_- T\restriction_{\Ch}$ and $\nabla_T T \restriction_{\Hp} = \kappa_+ T\restriction_{\Hp}$, where $\kappa_\pm$ is defined by~\eqref{eq:surface}.

At this point, we note that we can attach corners to $\mathcal{H}_0$ and $\mathcal{CH}_0$ to extend $\tilde{\mathcal{M}}$ smoothly to a Lorentzian manifold with corners. To be more precise, we attach the past bifurcation sphere $\mathcal{B}_-$ to $\mathcal{H}_0$ as the point $(U_\mathcal{H},V_\mathcal{H}) = (0,0)$. Then, define $\mathcal{H}:=\mathcal{H}_0 \cup \mathcal{B}_-$. Similarly, we can attach the future bifurcation sphere $\mathcal{B}_+$ to the Cauchy horizon which will be denoted by $\Ch$. We call the resulting manifold ${\mathcal{M}_{\mathrm{RN}}}$. 
Further details are not given since the precise construction is straightforward and the metric extends smoothly. Moreover, the $T$ vector field extends smoothly to $\mathcal{B}_+$ and $\mathcal{B}_-$ and vanishes there. Its Penrose diagram is depicted in \cref{fig:penrosembar}.
\begin{figure}[H]
	\centering
	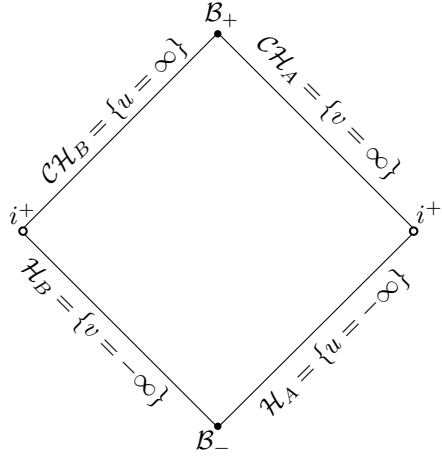
	\caption{Penrose diagram of ${\mathcal{M}_\text{RN}}$ which includes the bifurcate spheres $\mathcal{B}_+$ and $\mathcal{B}_-$.}
	\label{fig:penrosembar}
\end{figure}
Further details about the coordinate systems can be found in~\cite{o2014geometry}. From a dynamical point of view, we can also consider the spacetimes $(\mathcal M_\mathrm{RN} , g_{M,Q})$ as being contained in the Cauchy development of a spacelike hypersurface with two asymptotically flat ends solving the Einstein--Maxwell system in spherical symmetry. 

\subsection{The characteristic initial value problem for the wave equation}
In the context of scattering theory we will be interested in solutions to the wave equation \eqref{eq:linearwave} arising from suitable characteristic initial data. Recall the following well-posedness result for \eqref{eq:linearwave} with characteristic initial data.
\begin{prop}\label{thm:welldefined} Let  $\Psi \in  C_c^\infty(\Ho) $ be smooth compactly supported data on the event horizon $\Hp$. Then, there exists a unique smooth solution $\psi$ to \eqref{eq:linearwave} on $\mathcal{M}_{\mathrm{RN}}\setminus \Ch$ such that $\psi\restriction_{\mathcal{H}} = \Psi$. 
\end{prop}
The previous proposition is well known, see \cite{characteristic,rendall1990reduction}.
Analogously, we have the following for the backward evolution.
\begin{prop}\label{thm:time}
 Let  $\Psi \in  C_c^\infty(\Ch) $ be smooth compactly supported data on the Cauchy horizon $\Ch$. Then, there exists a unique smooth solution $\psi$ to \eqref{eq:linearwave} on $\mathcal{M}_{\mathrm{RN}}\setminus \Hp$ such that $\psi\restriction_{\Ch} = \Psi$.
\end{prop}

\subsection{Hilbert spaces of finite \texorpdfstring{$T$}{T} energy on both horizon components} 
Now, we are in the position to define the Hilbert spaces on the event $\mathcal{H}= \mathcal{H}_A \cup \mathcal{H}_B \cup \mathcal{B}_-$ and Cauchy horizon $\mathcal{CH}= \mathcal{CH}_A \cup \mathcal{CH}_B \cup \mathcal{B}_+$, respectively. 

We will start with constructing the Hilbert space on the event horizon. Roughly speaking, it will be defined by requiring finiteness of the $T$ energy flux on $\mathcal{H}_A$ \underline{minus} the $T$ energy flux on $\mathcal{H}_B$. 
More precisely, let $C^\infty_c(\mathcal{H})$ be the vector space of smooth compactly supported functions on $\mathcal{H}$. 
Recall that the Killing vector field $T$ is also a null generator of $\mathcal{H}$ and vanishes at the past bifurcation sphere $\mathcal{B}_-$. 
This allows us to define the norm $\| \cdot \|^2_{\mathcal{E}^T_{\mathcal H}}$ on the vector space $C^\infty_c(\mathcal{H})$ as 
\begin{align}\label{eq:firstdefnflux}
		\| \psi \|^2_{\mathcal{E}^T_{\mathcal H}}:= \int_{\mathcal{H}_A} J_\mu^T[\psi] n_{\mathcal{H}_A}^\mu \, \d\mathrm{vol}_{n_{\mathcal{H}_A}} -\int_{\mathcal{H}_B} J_\mu^T[\psi] n_{\mathcal{H}_B}^\mu \, \d\mathrm{vol}_{n_{\mathcal{H}_B}},
\end{align}
where  $\psi \in C_c^\infty(\mathcal{H})$, $\mathbf T[\psi]$ is the energy momentum tensor
\begin{align}
\mathbf{T}[\psi]_{\mu\nu} := \mathrm{Re}(\partial_\mu\psi\overline{ \partial_\nu \psi}) - \frac{1}{2}g_{\mu\nu} \partial_\alpha\psi \overline{\partial^\alpha\psi}, 
\end{align} and $J^T[\psi] := \mathbf{T}[\psi](T,\cdot)$. 
In \eqref{eq:firstdefnflux}, the fluxes are defined with respect to future directed normal vector fields $n_{\mathcal{H}_A}$ and $n_{\mathcal{H}_B}$ on $\mathcal{H}_A$ and $\mathcal{H}_B$, respectively.\footnote{A choice of such normal vectors fixes the volume form. Also note that this is the natural setup for energy estimates. }  Moreover, recall from \cref{fig:compare} that $T$ is future (resp.\ past) directed on $\mathcal{H}_A$ (resp.\ $\mathcal{H}_B$). Thus, the terms arising in \eqref{eq:firstdefnflux} satisfy $\int_{\mathcal{H}_A} J_\mu^T[\psi] n_{\mathcal{H}_A}^\mu \, \d\mathrm{vol}\geq 0$ and $-\int_{\mathcal{H}_B} J_\mu^T[\psi] n_{\mathcal{H}_B}^\mu \, \d\mathrm{vol}\geq 0 $.  
Again, in view of the fact that on the component $\Hp_B$ the normal vector field $T$ is past directed, we can choose $n_{\Hp_A}: = T\restriction_{\Hp_A}$ and $n_{\Hp_B}:= -T\restriction_{\Hp_B}$ as the future directed normal vector fields on $\mathcal{H}_A$ and $\mathcal{H}_B$, respectively, such that we can realize the norm \eqref{eq:firstdefnflux} as
(using the coordinate charts \eqref{eq:ha} and \eqref{eq:hb}) 
\begin{align}
	\| \psi \|^2_{\mathcal{E}^T_{\mathcal H}} = \int_{\mathbb R \times \mathbb{S}^2} |\partial_v\psi\restriction_{\Hp_A}|^2 \d v \sin\theta\d\theta \d \varphi +  \int_{\mathbb R \times \mathbb{S}^2} |\partial_u\psi\restriction_{\Hp_B}|^2 \d u \sin\theta\d\theta \d \varphi .
\end{align}
The norm \eqref{eq:firstdefnflux} defines an inner product, hence its completion is a Hilbert space.
\begin{definition}
	We define the Hilbert space of finite $T$ energy $\mathcal{E}^T_{\Hp}$ on both components of the event horizon as the completion of smooth and compactly supported functions $C_c^\infty(\mathcal{H})$ on the event horizon $\mathcal{H}= \mathcal{H}_A \cup \mathcal{H}_B \cup \mathcal{B}_-$ with respect to the norm \eqref{eq:firstdefnflux}.
\end{definition}
Analogously, we can consider the vector space $C_c^\infty(\Ch)$ and define the norm $\| \cdot \|^2_{\mathcal{E}^T_{\mathcal{CH}}}$ as the $T$ energy flux on the component $\mathcal{CH}_B$ \underline{minus} the $T$ energy flux on the component $\mathcal{CH}_A$:
\begin{align}\label{eq:firstdefnfluxch}
\| \psi \|^2_{\mathcal{E}^T_{\mathcal{CH}}}:= \int_{\mathcal{CH}_B} J_\mu^T[\psi] n_{\mathcal{CH}_B}^\mu \, \d\mathrm{vol}_{n_{\mathcal{CH}_B}} -\int_{\mathcal{CH}_A} J_\mu^T[\psi] n_{\mathcal{CH}_A}^\mu \, \d\mathrm{vol}_{n_{\mathcal{CH}_A}}.
\end{align} Again, in view of the orientation of the $T$ vector field (cf.\ \cref{fig:compare}), this norm can be represented as 
(using the coordinate charts \eqref{eq:ca} and \eqref{eq:cb}) 
\begin{align}
	\| \psi \|^2_{\mathcal{E}^T_{\mathcal{CH}}} = \int_{\mathbb R \times \mathbb{S}^2} |\partial_v\psi\restriction_{\Ch_B}|^2 \d v \sin\theta\d\theta \d \varphi +  \int_{\mathbb R \times \mathbb{S}^2} |\partial_u\psi\restriction_{\Ch_A}|^2 \d u \sin\theta\d\theta \d \varphi .
\end{align}
\begin{definition}
		We define the Hilbert space of finite $T$ energy $ \mathcal{E}^T_{\Ch}$ on both components of the Cauchy horizon as the completion of smooth and compactly supported functions $C_c^\infty(\mathcal{CH})$ the Cauchy horizon $\mathcal{CH}= \mathcal{CH}_A \cup \mathcal{CH}_B \cup \mathcal{B}_+$ with respect to the norm \eqref{eq:firstdefnfluxch}.
\end{definition}

\subsection{Separation of variables}
In this section we show how the radial o.d.e.\ \eqref{eq:radialode1} arises from decomposing a general solution in spherical harmonics and Fourier modes. For concreteness, let $\psi$ be a smooth solution to $\Box_g \psi = 0$ such that on each $\{r=const.\}$ slice, $\psi$ is compactly supported in the $t$ variable.\footnote{Note that we will prove later that such solutions arise from data which are dense in $\mathcal{E}_\Hp^T$.}
Then, we can define its Fourier transform in the $t$ variable and also decompose $\psi$ in spherical harmonics to end up with
\begin{align}
\hat \psi_{m\ell} (  r,\omega) := \int_{\mathbb{R}\times \mathbb{S}^2} e^{-i\omega t} Y_{m\ell}(\theta,\phi) \psi(t,r,\theta,\phi) \sin\theta \d\theta\d\phi \frac{\d t}{\sqrt{2\pi}}.
\end{align} Due to the compact support on constant $r$ slices, the wave equation $\Box_g \psi = 0$ implies that  \begin{align}\hat \psi_{m\ell}(r, \omega) =: R_{m\ell}^{(\omega)}(r) =: R(r) \end{align} satisfies the radial o.d.e. 
\begin{align}\label{eq:radialode}
\Delta \frac{\d}{\d r}\left(\Delta \frac{\d}{\d r} R\right) - \Delta \ell(\ell+1) R + r^4 \omega^2 R =0.
\end{align}
In \cref{sec:radial} we will analyze solutions to \eqref{eq:radialode} and denote a solution thereof with $R(r)$. Moreover, it is useful to introduce the function $u$ defined as
\begin{align}
u(r) := r R(r)
\end{align}
and consider $u = u(r(r_\ast))$ as a function of $r_\ast$, which is defined in \eqref{defn:rstar}.
Using the $r_\ast$ variable, the o.d.e.~\eqref{eq:radialode} finally reduces to 
\begin{align}
u^{\prime \prime} + (\omega^2 - V_\ell ) u= 0 \label{ODE1}
\end{align}
on the real line with potential 
\begin{align}\label{eq:potential}
V=V_\ell =  \Delta\left( \frac{r(r_+ + r_-) - 2 r_+ r_-}{r^3} + \frac{\ell(\ell+1)}{r^4}\right).
\end{align} 
In \cref{lem:asymptoticspotential} in the appendix it is proven that, as a function of $r_*$, the potential $V_\ell$ decays exponentially as $ r_\ast \to \pm \infty$. In particular, this indicates that we have asymptotic free waves (asymptotic states) near the event and Cauchy horizon of the form  $e^{\pm i \omega r_\ast}$ as $|r_\ast | \to \infty$. In order to construct these solutions we use the following proposition for Volterra integral equations (see Lemma 2.4 of \cite{schlag2010decay}).
\begin{prop}\label{lem:volterra}
	Let $x_0\in \mathbb{R}\cup \{+\infty\}$ and $g\in L^\infty(-\infty,x_0)$. Suppose the integral kernel $K$ satisfies 
	\begin{align}
	\alpha := \int_{-\infty}^{x_0} \sup_{ \{x: y<x<x_0 \} } |K(x,y)| \d y < \infty.
	\end{align}
	Then, the Volterra integral equation
	\begin{align}\label{avolterraeqn}
	f(x) = g(x) + \int_{-\infty}^{x} K(x,y) f(y) \d y
	\end{align}
	has a unique solution $f$ satisfying
	\begin{align}
	\| f\|_{L^\infty(-\infty,x_0)}\leq e^{\alpha} \| g\|_{L^\infty(-\infty,x_0)}.
	\end{align}
	If in addition $K$ is smooth in both variables and 
	\begin{align}
	\int_{-\infty}^{x_0} \sup_{\{x: y<x<x_0\}} |\partial_x^k K(x,y)| \d y < \infty
	\end{align}
	for all $k\in \mathbb{N}$, then the solution $f$ is smooth on $(-\infty, x_0)$ and the derivatives can be computed by formal differentiation  of~\eqref{avolterraeqn}.
\end{prop}
\begin{rmk}\label{rmk:volterra}
	Analogous results as in \cref{lem:volterra} also hold true for Volterra integral equations on intervals of the form $(x_0,x_1)$ or $(x_0,+\infty)$. 
\end{rmk}
This allows us to define the following fundamental pairs of solutions to the o.d.e.~\eqref{ODE1}. In view of the exponential decay of the potential, it is straightforward to check that the assumptions of \cref{lem:volterra} are satisfied.
\begin{definition}\label{defn:u1u2}
	Let $\omega \in \mathbb R$ and $\ell\in\mathbb N_0$ be fixed. Define asymptotic state solutions $u_1$ and $u_2$ of the radial o.d.e.\ \eqref{ODE1} as the unique solutions to the Volterra integral equations
	\begin{align}\label{eq:constructionu1}
		&u_1(\omega, r_\ast) = e^{i\omega r_\ast} + \int_{-\infty}^{r_\ast} \frac{\sin(\omega(r_\ast - y))}{\omega} V(y) u_1(\omega, y) \d y,
\\
&u_2(\omega, r_\ast) = e^{-i\omega r_\ast} + \int_{-\infty}^{r_\ast} \frac{\sin(\omega(r_\ast - y))}{\omega} V(y) u_2(\omega, y) \d y.	\end{align}
Analogously, define $v_1$ and $v_2$ as the unique solutions to the Volterra integral equations
\begin{align}
&v_1(\omega, r_\ast) = e^{i\omega r_\ast} - \int_{r_\ast}^{\infty} \frac{\sin(\omega(r_\ast - y))}{\omega} V(y) v_1(\omega, y) \d y,
\\
&v_2(\omega, r_\ast) = e^{-i\omega r_\ast} - \int_{r_\ast}^{\infty} \frac{\sin(\omega(r_\ast - y))}{\omega} V(y) v_2(\omega, y) \d y.	\end{align}

For $\omega =0$, we set $\frac{\sin(\omega(r_\ast - y))}{\omega}\restriction_{\omega =0} = r_\ast - y$ in the integral kernel in which case $u_1$ and $u_2$ coincide. We define 
\begin{align} \tilde u_1(r_\ast) := u_1(0,r_\ast) = u_2(0,r_\ast) \end{align} 
and similarly,
\begin{align}
	\tilde v_1 (r_\ast) := v_1(0,r_\ast) = v_2(0,r_\ast).
\end{align}
Since $u_1(0,r_\ast) = u_2(0,r_\ast)$ for $\omega=0$, there exists another linearly independent fundamental solution $\tilde u_2$ solving the Volterra integral equation
\begin{align}
	&\tilde u_2 (r_\ast) = r_\ast + \int_{-\infty}^{r_\ast} (r_\ast -y) V(y) \tilde u_2(y) \d y.
	\end{align}
	Similarly, we also have another fundamental solution, which is linearly independent from $\tilde v_1$, solving
	\begin{align}
	&\tilde v_2 (r_\ast) = r_\ast - \int_{r^\ast}^{\infty} (r_\ast -y) V(y) \tilde v_2(y) \d y.
\end{align} 
Since $r_*$ is not uniformly bounded, we cannot apply \cref{lem:volterra} to construct $\tilde u_2$ and $\tilde v_2$. Nevertheless, after switching to coordinates which are regular at $\mathcal{H}$ or $\mathcal{CH}$, the existence of the desired solutions follows immediately from the usual local theory of regular singularities (see~\cite{olver2014asymptotics}).
\end{definition}
\begin{rmk}
	\label{rmk:holo}
	Due to the exponential decay of the potential $V_\ell$ (see \cref{lem:asymptoticspotential} in the appendix), it follows from standard theory that the solutions $u_1(\omega, r_\ast), u_2(\omega, r_\ast)$, $v_1(\omega, r_\ast)$ and $v_2(\omega, r_\ast)$ can be continued to holomorphic functions of $\omega$ in the strip $|\operatorname{Im}(\omega)| <  \kappa_+$ for fixed $r_\ast \in \mathbb R$. Indeed, in \cite{hartle1982crossing} it is shown that $u_1(\omega,r_\ast)$ is analytic in $\mathbb C\setminus \{im\kappa_+\colon m\in \mathbb{N}\}$ with possible poles at $\{ i m \kappa_+ \colon m\in \mathbb N \}$ and similarly for $u_2, v_1$, and $v_2$. See also the proof of \cref{prop:rtboundedcomplex} in the appendix.
\end{rmk}
This allows us now to define the reflection and transmission coefficients $\mathfrak R$ and $\mathfrak T$.
\begin{definition} \label{defn:TandR} Let $\omega \neq 0$. Then we define the transmission coefficient $\mathfrak T(\omega, \ell)$ and reflection coefficient $\mathfrak R(\omega, \ell)$ as the unique coefficients such that
\begin{align}
u_1 = \mathfrak T v_1 + \mathfrak R v_2. \label{eq:defn1}
\end{align}

 Using the fact that the Wronskian 
 \begin{align}\mathfrak  W(f,g) := f g^\prime - f^\prime g\end{align} of two solutions $f$ and $g$ is independent of $r_\ast$, we can equivalently define the scattering coefficients as
\begin{align}
	\mathfrak T : = \frac{ \mathfrak W(u_1,v_2)}{ \mathfrak W(v_1,v_2)} = \frac{\mathfrak   W(u_1,v_2)}{-2i\omega}
\end{align} 
and 
\begin{align}
\mathfrak R : = \frac{\mathfrak W(u_1,v_1)}{\mathfrak W(v_2,v_1)} = \frac{ \mathfrak W(u_1,v_1)}{2i\omega}.
\end{align} 
\end{definition}
The transmission and reflection coefficients satisfy a pseudo-unitarity property proven in the following.
\begin{prop}[Pseudo-unitarity in the separated picture]\label{prop:pseudouni}
	The transmission and reflection coefficients satisfy\begin{align}\label{eq:pseudounitaryode}
	1 =	|\mathfrak T|^2  - |\mathfrak R|^2.
	\end{align}
	\begin{proof} First, note that any solution to the o.d.e.\ \eqref{ODE1} satisfies the identity \begin{align}
			\operatorname{Im}(\bar u u^\prime ) = const.\
		\end{align}
		Applying this to the solution $u_1 = \mathfrak T v_1 + \mathfrak R v_2$ shows the claim.
	\end{proof}
\end{prop}
In the following we shall see that the reflection and transmission coefficients are regular at $\omega =0$. 
\begin{prop}\label{cor:rdoesntvanish}Let $\ell \in \mathbb N_0 $ be fixed. Then the scattering coefficients $\mathfrak R(\omega,\ell)$ and $\mathfrak T (\omega,\ell)$ are analytic functions of $\omega$ in the strip $
	\{\omega \in \mathbb C \colon |\operatorname{Im}(\omega)| < \kappa_+
	\}$  with values for $\omega =0$ given by
	\begin{align}\label{eq:r0}
	&\mathfrak R(0,\ell) = \frac{(-1)^\ell}{2} \left( \frac{r_-}{r_+} - \frac{r_+}{r_-}\right),\\
	&\label{eq:t0}	\mathfrak T(0,\ell) = \frac{(-1)^\ell}{2} \left(\frac{r_-}{r_+}+ \frac{r_+}{r_-}  \right).
	\end{align} In particular, the reflection coefficient $\mathfrak R(\omega,\ell)$ only vanishes on a discrete set of points $\omega$. 
	
	Moreover, the reflection and transmission coefficients $\mathfrak R(\omega,\ell)$ and $\mathfrak T(\omega,\ell)$ are analytic functions on $\mathbb C \setminus \mathbb P$ with possible poles at $\mathbb P = \{ i m \kappa_+\colon m \in \mathbb N\} \cup \{i k \kappa_- \colon k \in \mathbb Z \setminus \{ 0\} \}$.
	\begin{proof}
		From the analyticity of $u_1,u_2,v_1$, and $v_2$ in the strip $|\operatorname{Im}(\omega) |<  \kappa_+$ (cf.\ \cref{rmk:holo}), we conclude that $\mathfrak T$ and $\mathfrak R$ are holomorphic in $\{\omega\neq 0 \in \mathbb C : |\operatorname{Im}(\omega) |< \kappa_+\}$ with a possible pole at $\omega =0$. In the following we shall show that $\{\omega =0\}$ is a removable singularity. Indeed, we will compute the explicit value of the reflection and transmission coefficient at $\omega=0$ and deduce that for fixed $\ell\in\mathbb{N}_0$, the transmission coefficient $\mathfrak T(\omega,\ell)$ and the reflection coefficient $\mathfrak R(\omega,\ell)$ are analytic functions on the strip $\{\omega \in \mathbb C \colon \operatorname{Im}(\omega) |< \kappa_+ \}$ (cf.\ unpublished work of McNamara cited in \cite{gursel1979evolution}). To do so, note that from \cref{prop:intermediate} in \cref{sec:boundedrefltransmi} we conclude the pointwise limits
		\begin{align}
		&u_1(\omega,r_\ast) \to \tilde u_1(r_\ast),\\
		&v_1(\omega,r_\ast) \to \tilde v_1 (r_\ast) = (-1)^\ell \frac{r_+}{r_-}\tilde u_1(r_\ast),\\
		&v_2(\omega,r_\ast) \to \tilde v_1 (r_\ast) =  (-1)^\ell \frac{r_+}{r_-}\tilde u_1(r_\ast)
		\end{align}
		as $|\omega|\to 0$.  Using the definition in \eqref{eq:defn1} of $\mathfrak T(\omega,\ell)$, $\mathfrak R(\omega,\ell)$, and the condition $1+|\mathfrak R|^2 = |\mathfrak T|^2$ (cf.~\cref{prop:pseudouni}), we deduce that the limits $\lim_{\omega\to 0} \mathfrak R(\omega, \ell)$ and  $\lim_{\omega\to 0} \mathfrak T(\omega, \ell)$ exist and moreover can be computed to be \eqref{eq:r0} and \eqref{eq:t0}.		
		Note that \eqref{eq:r0} and \eqref{eq:t0} have been established in \cite{gursel1979final}.
		Also note that in view of the analyticity properties of $u_1$, $v_1$, and $v_2$, the  $\mathfrak R(\omega,\ell)$ and $\mathfrak T(\omega,\ell)$ are analytic functions on $\mathbb C \setminus \mathbb P$ with possible poles at $\mathbb P = \{ i m \kappa_+\colon m \in \mathbb N\} \cup \{i k \kappa_- \colon k \in \mathbb Z \setminus \{ 0\} \}$.
	\end{proof}
\end{prop}
\subsection{Conventions} Let $X$ be a point set with a limit point $c$ (e.g.\ $X = \mathbb R, [a,b], \mathbb C$). Throughout this paper we will use the symbols $\lesssim$ and $\gtrsim$, where the implicit constants might depend on the black hole parameters $M$ and $Q$. In particular, for functions (or constants) $a(x),b(x) >0$ the notation $a\lesssim b$ means that there is a constant $C = C(M,Q) > 0 $ such that $a(x)\leq Cb(x) $ for all $x\in X$. We will also make use of the notation $\lesssim_\ell$ or $\gtrsim_\ell$ which means that the constant may additionally also depend on $\ell$. We also write $a\sim b$ if there are constants $C(M,Q),\tilde C(M,Q) >0$ such that $C a(x) \leq b(x) \leq \tilde Ca(x)$ for all $x\in X$. 

We shall also make use of the standard Landau notation $O$ and $o$ \cite{NIST:DLMF,olver2014asymptotics}.  To be more precise, as $x\to c$ in $X$
	\begin{align}
	f(x) = O(g(x)) &\text{ means } \left|\frac{f(x)}{g(x)}\right| \leq C(M,Q)\label{eq:O}
	\end{align}
	and
\begin{align}
	f(x) = o(g(x)) &\text{ means }  \frac{f(x)}{g(x)} \to 0\label{{eq:O1}}.
	\end{align}
We will also use the notation $O_\ell$ if the constant $C$ in \eqref{eq:O} may additionally depend on $\ell$. 
\section{Main theorems}
\label{sec:mainthms}
In this section we will  formulate our  main theorems.

\cref{thm:forwardevolution}, which we state in \cref{sec:existencescatteringmap}, establishes the existence of a scattering map $S^T$ of the form
\begin{align}
	&S^T: \mathcal{E}^T_{\Ho}\to\mathcal{E}^T_{\Ch},
\end{align}
which is a Hilbert space isomorphism, i.e.\ a bounded and invertible map with bounded inverse. \cref{thm:forwardevolution} will be proven in \cref{sec:mainthm}. In the separated picture, the boundedness of $S^T$ corresponds to the uniform boundedness of the transmission and reflection coefficients which is stated as \cref{thm:boundednesstrans} in \cref{subsec:scatteringcoefficients}.  \cref{thm:boundednesstrans} will be proven in \cref{sec:radial} (and later used in the proof of \cref{thm:forwardevolution}).

  \cref{subsec:connfourierphysical} is then devoted to  \cref{thm:fouriertophysical}, which
  connects our physical space scattering theory to the fixed frequency scattering theory. (We will infer \cref{thm:fouriertophysical} as a consequence of \cref{thm:forwardevolution}.) In \cref{subsec:reflection}, this connection allows us to prove that the reflection map is injective, which is the content of \cref{thm:nonvanishingreflection}. In \cref{thm:c1instab}, which is stated and proven in \cref{sec:c1instab}, we construct data which are incoming and compactly supported but nevertheless, lead to a solution which fails to be in $C^1$ on the Cauchy horizon. 

We end this section with the statement of our two non-existence results. In \cref{subsec:nonex} we formulate \cref{thm:cosmological}, the non-existence of the $T$ energy scattering theory for the Klein--Gordon equation with conformal mass on the interior of (anti-) de~Sitter--Reissner--Nordstr\"om black holes. The proof of \cref{thm:cosmological} is given in \cref{sec:cosmo}. Finally, in \cref{thm:kleingordon}, stated in \cref{subsec:notkleingordon}, we show the non-existence of the $T$ energy scattering map for the Klein--Gordon equation on the interior of Reissner--Nordstr\"om. The proof of \cref{thm:kleingordon} is given in \cref{sec:kleingordonequation}.

\subsection{Existence and boundedness of the \texorpdfstring{$T$}{T} energy scattering map}
\label{sec:existencescatteringmap}
 First, we define the forward (resp.\ backward) evolution on a dense domain. 
\begin{definition}\label{defn:domainforward}
	The domains of the forward and backward evolution are defined as \begin{align}\nonumber
	\mathcal{D}^T_{\mathcal{H}}:= \{ \psi \in C_c^\infty(\mathcal H) &\subset \mathcal{E}_{\Ho}^T \text{ s.t. the Cauchy evolution of $\psi$ has}\\ &\text{ compact support on  constant $r=const.$ hypersurfaces} \} \end{align}
	and
	\begin{align}\nonumber
	\mathcal{D}^T_{\mathcal{CH}}:= \{ \psi \in C_c^\infty(\mathcal{CH})& \subset \mathcal{E}_{\Ch}^T \text{ s.t. the backward evolution of $\psi$ has}\\ & \text{ compact support on  constant $r=const.$ hypersurfaces} \}, \end{align}
	respectively. Here, we consider $r_-< r < r_+$ and note that if $\psi$ is compactly supported on one $\{r=const.\}$ slice, then, as a direct consequence of the domain of dependence, its evolution will be compactly supported on all other $\{r=const.\}$ hypersurfaces for $r_-< r < r_+$. 
	
We will prove in \cref{lem:lemmadense} in \cref{sec:mainthm} that $\mathcal{D}^T_{\mathcal{H}} \subset  \mathcal{E}_{\Hp}^T$ and $ \mathcal{D}^T_{\mathcal{CH}} \subset \mathcal{E}_{\Ch}^T $ are dense domains.
\end{definition}
These definitions of the domains are motivated by the following observation.
\begin{rmk}\label{rmk:domain}
 Suppose we are given data in $\mathcal{D}_\Ho^T$ on the event horizon $\Hp$. Consider now the unique Cauchy development (cf.\ \cref{thm:welldefined}) and observe that its restriction to the Cauchy horizon $\Ch$ will lie in $\mathcal{D}_{\Ch}^T$. This holds true since we can first smoothly extend the metric beyond the Cauchy horizon $\Ch$ and then use the compact support on a constant $r_\ast$ hypersurface to solve an equivalent Cauchy problem in an appropriate region which extends the Cauchy horizon $\Ch$, includes the support of the solution, but does not include $i^+$. The smoothness of the solution up to and including the Cauchy horizon $\Ch$ follows now from Cauchy stability.
\end{rmk}

In view of \cref{rmk:domain} we can define the forward and backward map on the domains $\mathcal{D}^T_\Hp$ and $\mathcal{D}^T_\Ch$, respectively.
\begin{definition}
	Define the forward map $S_0^T\colon \mathcal{D}^T_{\mathcal{H}} \subset  \mathcal{E}_{\Hp}^T \to \mathcal{D}^T_{\mathcal{CH}} \subset \mathcal{E}_{\Ch}^T $ as the unique forward evolution from data on the event horizon to data on the Cauchy horizon. More precisely, let  $\psi$ be the solution to \eqref{eq:linearwave} arising from initial data $\Psi \in \mathcal{D}^T_{\mathcal{H}} \subset  \mathcal{E}_{\Hp}^T$. Then, define $S_0^T(\Psi)$ as the restriction of $\psi$ to the Cauchy horizon, i.e.\ $S_0^T(\Psi) := \psi\restriction_{\Ch} \in \mathcal{D}^T_{\Ch}$. 
	
	Similarly, let $\phi$ be the unique backward evolution of \eqref{eq:linearwave} arising from $\Phi \in \mathcal{D}^T_{\Ch}$. Then, define the backward map  by $B_0^T (\Phi):= \phi\restriction_{\Hp} \in \mathcal{D}^T_\Hp$.
\end{definition}

\begin{rmk}\label{rmk:injectiv}
		Note that by the uniqueness of the Cauchy evolution we have that  $S_0^T$ and $B_0^T$ are inverses of each other, i.e.\  $B_0^T\circ	S_0^T = \mathrm{Id}_{ \mathcal{D}^T_{\Ho} }, \; S_0^T \circ B_0^T = \mathrm{Id}_{ \mathcal{D}^T_{\Ch}}. $
\end{rmk}
Now, we are in the position to state our main theorem. 
\begin{theorem}\label{thm:forwardevolution}
The map $S_0^T\colon \mathcal{D}^T_\Hp \subset \mathcal{E}^T_\Hp \to \mathcal{D}^T_\Ch \subset \mathcal{E}^T_\Ch $ is bounded and uniquely extends to 
 \begin{align}
 &	S^T\colon \mathcal{E}^T_{\Ho} \to \mathcal{E}^T_{\Ch},
 \end{align}
called the ``scattering map''. The scattering map $S^T$ is a Hilbert space isomorphism, i.e.\ a \underline{bounded} and \underline{invertible} linear map with \underline{bounded inverse} $B^T\colon \mathcal{E}^T_{\Ch} \to  \mathcal{E}^T_{\Ho}$ satisfying
 \begin{align}
& \label{eq:inverses}B^T\circ	S^T = \mathrm{Id}_{ \mathcal{E}^T_{\Ho} }, \; S^T \circ B^T = \mathrm{Id}_{ \mathcal{E}^T_{\Ch}}.
 \end{align}
 Here,  $B^T\colon \mathcal{E}^T_{\Ch} \to \mathcal{E}^T_{\Ho}$ is the ``backward map'', which is the unique bounded extension of $B_0^T$. 
 
 In addition, the scattering map $S^T$ is pseudo-unitary, meaning that  for $\psi\in \mathcal{E}^T_{\mathcal H}$, we have
 \begin{align}\label{eq:pseudounitary}
	\int_{\Hp_A} |T\psi|^2 - \int_{\Hp_B} |T\psi|^2 = \int_{\Ch_B} |TS^T\psi|^2 -\int_{\Ch_A} |T S^T\psi|^2.
\end{align}
\end{theorem}
In more traditional language, \cref{thm:forwardevolution} yields existence, uniqueness, and asymptotic completeness of scattering states. 

The proof of \cref{thm:forwardevolution} is given in \cref{sec:mainthm}. Let us note already that \cref{thm:forwardevolution} is a posteriori the physical space equivalent of the uniform boundedness of the scattering coefficients proven in \cref{thm:boundednesstrans} (see \cref{subsec:scatteringcoefficients}). This equivalence is made precise in \cref{thm:fouriertophysical} (see \cref{subsec:connfourierphysical}). 

\begin{rmk}
	Note that in general, neither initial data nor scattered data have to be bounded in $L^\infty$ or continuous. Indeed, we have that $\Phi_A(u,\theta,\varphi) = \log(u)\chi_{u\geq 1} \in \mathcal{E}^T_{\Ch_A}$, where $\chi_{u\geq 1}$ is a smooth cutoff. Thus, there exist initial data $B^T(\Phi_A)\in \mathcal{E}^T_{\Hp}$ such that its image under the scattering map is not in $L^\infty$ and not continuous. We emphasize the contrast with the estimates from \cite{franzen2016boundedness} for which more regularity and decay along the event horizon $\mathcal{H}$ are necessary. 
\end{rmk}
\subsection{Uniform boundedness of the transmission and reflection coefficients}
\label{subsec:scatteringcoefficients}
On the level of the o.d.e.~\eqref{ODE1} in the separated picture, the problem of boundedness of the scattering map reduces to proving that the transmission coefficient $\mathfrak T$ and the reflection coefficient $\mathfrak R$ are uniformly bounded over all parameter ranges of $\omega\in\mathbb R$ and $\ell\in\mathbb{N}_0$. This is stated as \cref{thm:boundednesstrans} below.  
\begin{restatable}{theorem}{boundedness}
	\label{thm:boundednesstrans}The reflection and transmission coefficients $\mathfrak R(\omega,\ell)$ and $\mathfrak T(\omega, \ell)$ are uniformly bounded, i.e.\ they satisfy
	\begin{align}
	\sup_{\omega \in \mathbb R , \ell \in \mathbb N_0 }(	|\mathfrak R(\omega, \ell) | + |\mathfrak T (\omega, \ell)| )  \lesssim 1.
	\end{align}
\end{restatable}
\cref{thm:boundednesstrans} is proved in \cref{sec:radial}. As discussed in the introduction, the proof relies on an explicit calculation for $\omega =0$ together with a careful analysis of the radial o.d.e.~\eqref{ODE1}, involving properties of special functions and perturbations thereof.

Let us note that, \emph{given} \cref{thm:forwardevolution}, we could infer \cref{thm:boundednesstrans} as a corollary (using the theory to be described in \cref{subsec:connfourierphysical}). We emphasize, however, that in the present paper we \emph{use} \cref{thm:boundednesstrans} to prove \cref{thm:forwardevolution} in \cref{sec:mainthm}.
\subsection{Connection between the separated and the physical space picture}
In this section, we will make the connection of the separated and physical space picture precise. 

\label{subsec:connfourierphysical}
First, let us note that we have natural Hilbert space decompositions $\mathcal{E}^T_{\Ho} \cong \mathcal{E}^T_{\Ho_A} \oplus \mathcal{E}^T_{\Ho_B}$ and $ \mathcal{E}^T_{\Ch}\cong \mathcal{E}^T_{\Ch_B} \oplus \mathcal{E}^T_{\Ch_A}$. 
	\begin{prop}\label{cor:decomposition}
		The Hilbert spaces $\mathcal{E}^T_{\Hp}$ and $\mathcal{E}^T_{\Ch}$ of finite $T$ energy on the event horizon $\mathcal H$ and on the Cauchy horizon $\Ch$ admit the orthogonal decomposition
		\begin{align}
		\mathcal{E}^T_{\Hp} \cong \mathcal{E}^T_{\Ho_A} \oplus \mathcal{E}^T_{\Ho_B}\, \text{  and  }\; \mathcal{E}^T_{\Ch} \cong \mathcal{E}^T_{\Ch_A} \oplus \mathcal{E}^T_{\Ch_B}.
		\end{align}
		\begin{proof}
			Clearly, the embedding $i\colon \mathcal{E}^T_{\Ho_A} \oplus \mathcal{E}^T_{\Ho_B} \hookrightarrow 	\mathcal{E}^T_{\Hp}$ is well-defined and isometric. It remains to show that $i$ is surjective. Let $\psi \in C_c^\infty(\mathcal{H})$. First, we show that we can approximate (in $T$-energy) $\psi\restriction_{\Ho_A}$ on $\mathcal{H}_A$ with functions $\psi_\epsilon \in C_c^\infty(\mathcal{H}_A)$ which are supported away from the past bifurcation sphere. On $\mathcal{H}_A $ choose non-degenerate coordinates $(V,\theta,\varphi) := (V_\Ho, \theta,\varphi)$ as in \cref{sec:interiorwboundary} and recall that the past bifurcation sphere is $\{V=0\}$. Then, for small $\epsilon >0$, set 
			\begin{align}\label{eq:logcut1}
			\psi_\epsilon(V,\theta, \varphi) := \psi(U=0,V,\theta,\varphi) \chi(-\epsilon \log(V)),
			\end{align}
			where $\chi\colon \mathbb{R} \to [0,1]$ is smooth and such that $\operatorname{supp} (\chi) \subseteq (-\infty,2]$ and $\chi\restriction_{(-\infty,1]} = 1$. 
			Then, it is straightforward to check that $\psi_\epsilon \in C_c^\infty(\Ho_A)$ and 
			\begin{align}\label{eq:logcut2}
			\int_{\Ho_A} J^T[\psi-\psi_\epsilon]_\mu n^\mu \d\mathrm{vol} \lesssim \int_{\mathbb{S}^2} \int_{0}^{\infty}  V (\partial_V (\psi - \psi_\epsilon) )^2 \d{V} \sin\theta \d \theta \d\varphi \to 0 
			\end{align}
			as $\epsilon \to 0$. 
			Analogously, we can do this for $\mathcal{H}_B$ from which the claim follows. \end{proof}
	\end{prop}
	We will use this identification to represent the scattering map also in the Fourier picture and show how these pictures connect. To do so we define the following. 
\begin{definition}
	For $(\Psi_A, \Psi_B) \in \mathcal{E}^T_{\Hp_A} \oplus \mathcal{E}^T_{\Hp_B}$ note that $\partial_v\Psi_A (v,\theta,\phi) \in L^2(\mathbb R \times \mathbb{S}^2; \mathbb C)$ and analogously for $\Psi_B$. Hence, in mild abuse of notation, we can define the Fourier and spherical harmonics coefficients $ \mathcal{F}_{\Ho_A}(\Psi_A)$ and $\mathcal{F}_{\Ho_B}(\Psi_B)$ as 
	\begin{align}\label{eq:defnF1}
	i \omega \mathcal{F}_{\Ho_A}(\Psi_A) (\omega,m,\ell):= r_+ \int_{\mathbb R} \int_{\mathbb{S}^2} \partial_v\Psi_A(v,\theta,\varphi) e^{-i\omega v} Y_{\ell m}(\theta,\varphi) \sin\theta \d\theta \d\varphi \frac{\d v}{\sqrt{2\pi}}
	\end{align}
	and 
	\begin{align}
-	i \omega \mathcal{F}_{\Ho_B}(\Psi_B)(\omega, m ,\ell) := r_+ \int_{\mathbb R} \int_{\mathbb{S}^2} \partial_u\Psi_B(u,\theta,\varphi) e^{i\omega u} Y_{\ell m}(\theta,\varphi)\sin\theta \d\theta \d\varphi \frac{\d u}{\sqrt{2\pi}}.
	\end{align}
	Similarly, for $(\Phi_A, \Phi_B) \in  \mathcal{E}^T_{\Ch_A} \oplus \mathcal{E}^T_{\Ch_B}$ set
	\begin{align}
	-i \omega \mathcal{F}_{\Ch_A}(\Phi_A) (\omega, m, \ell ):= r_-\int_{\mathbb R} \int_{\mathbb{S}^2} \partial_u\Phi_A(u,\theta,\varphi) e^{i\omega u} Y_{\ell m}(\theta,\varphi) \sin\theta \d\theta \d\varphi \frac{\d u}{\sqrt{2\pi}}
	\end{align}
	and 
	\begin{align}\label{eq:defnF2}
	i \omega \mathcal{F}_{\Ch_B}(\Phi_B) (\omega, m , \ell ):= r_- \int_{\mathbb R} \int_{\mathbb{S}^2} \partial_v\Phi_B(v,\theta,\varphi) e^{-i\omega v} Y_{\ell m}(\theta,\varphi)\sin\theta \d\theta \d\varphi \frac{\d v}{\sqrt{2\pi}}.
	\end{align}
	 \end{definition}
	 Also, recall the Hilbert space decomposition $\mathcal{E}^T_{\Ho} \cong \mathcal{E}^T_{\Ho_A} \oplus \mathcal{E}^T_{\Ho_B}$ and $ \mathcal{E}^T_{\Ch}\cong \mathcal{E}^T_{\Ch_B} \oplus \mathcal{E}^T_{\Ch_A}$. Thus, the scattering matrix can be also decomposed as 
	 \begin{align}
	 S^T = \begin{pmatrix}
	 S^T_{BA}& S^T_{BB} \\
	 S^T_{AA} & S^T_{AB}
	 \end{pmatrix},
	 \end{align}   
	 where \begin{align}S^T_{ij}\colon \mathcal{E}^T_{\Ho_j} \to \mathcal{E}^T_{\Ch_i}\end{align} is a bounded linear map for $i,j \in \{ A,B\}$.\footnote{Note that $T$ does not denote the transpose but the fact that it is the scattering map associated with the $T$ vector field.}
	  \begin{definition}\label{def:hilbertspacesfourier}
Define the Hilbert spaces \begin{align*} &\hat{\mathcal{E}}^T_{\Hp_A} :=  \ell^2(Z;L^2(r_+^{-2} {\omega^2} \d \omega)),\, \; \hat{\mathcal{E}}^T_{\Hp_B} :=  \ell^2(Z;L^2(r_+^{-2} {\omega^2} \d \omega)),\\ &\hat{\mathcal{E}}^T_{\Ch_A} :=  \ell^2(Z;L^2(r_-^{-2} {\omega^2} \d \omega)),\,\; \hat{\mathcal{E}}^T_{\Ch_B}  :=  \ell^2(Z;L^2(r_-^{-2} {\omega^2} \d \omega)),\end{align*}
 where $Z = \{(m,\ell) \in \mathbb Z \times \mathbb N_0: |m|\leq \ell \}$.
\end{definition}
The Hilbert spaces defined in \cref{def:hilbertspacesfourier} are unitary isomorphic to their corresponding physical energy spaces. This is captured in
\begin{prop}
The linear maps defined in \eqref{eq:defnF1}--\eqref{eq:defnF2} \begin{align}&\mathcal{F}_{\Ho_A}\oplus \mathcal{F}_{\Ho_B}\colon  \mathcal{E}^T_{\Hp_A} \oplus \mathcal{E}^T_{\Hp_B} \to \hat{\mathcal{E}}^T_{\Hp_A} \oplus \hat{\mathcal{E}}^T_{\Hp_B} \\ &\mathcal{F}_{\Ch_B}\oplus \mathcal{F}_{\Ch_A} \colon \mathcal{E}^T_{\Ch_B} \oplus \mathcal{E}^T_{\Ch_A} \to \hat{\mathcal{E}}^T_{\Ch_B} \oplus \hat{\mathcal{E}}^T_{\Ch_A}  \end{align} are unitary.
\begin{proof}
	This follows from the fact that the Fourier transform and the decomposition into spherical harmonics are unitary maps. 
\end{proof}
\end{prop}
Now, we will define the scattering map in the separated picture and show that it is bounded.
\begin{prop}
	 The scattering map in the separated picture \begin{align}\hat{S^T} \colon  \hat{\mathcal{E}}^T_{\mathcal{H}_A} \oplus  \hat{\mathcal{E}}^T_{\mathcal{H}_B} \to  \hat{\mathcal{E}}^T_{\mathcal{CH}_B} \oplus  \hat{\mathcal{E}}^T_{\mathcal{CH}_A},\end{align}
	 defined as the multiplication operator
	 \begin{align}\label{eq:scatteringfourier}
	 \hat{S^T} = \begin{pmatrix}
	 \hat{S^T_{BA}}& \hat{S^T_{BB}} \\
	 \hat{S^T_{AA}} & \hat{S^T_{AB}}\end{pmatrix} := \begin{pmatrix}
	 \mathfrak T(\omega,  \ell )& \bar{\mathfrak R}  (\omega,  \ell ) \\
	 \mathfrak R(\omega,  \ell ) & \bar{\mathfrak T}  (\omega,  \ell ) 
	 \end{pmatrix},
	 \end{align}
	 is bounded.  Moreover, the map $\hat S^T$ is invertible with bounded inverse given by 
\begin{align}\label{eq:inverse}
	{{}{\hat{S^{T}}}}^{-1}  =  \begin{pmatrix}
\bar{\mathfrak T}(\omega,  \ell )& -\bar{\mathfrak R}  (\omega,  \ell ) \\
	-\mathfrak R(\omega,  \ell ) & {\mathfrak T}  (\omega, \ell ) 
	\end{pmatrix}.
\end{align}
\begin{proof}
	Indeed, $\hat{S}^T$ is bounded in view of the uniform boundedness of the transmission and reflection coefficients $\mathfrak T$ and $\mathfrak R$ (cf. \cref{thm:boundednesstrans}). Also note that
	 $	 |\mathfrak T|^2 = 1 + |\mathfrak R|^2 $
	 implies that \begin{align}\det \left( \hat S^T \right) = 1\end{align}
	 which shows \eqref{eq:inverse}. The boundedness of ${{}{\hat{S^{T}}}}^{-1}$ is again immediate since the scattering coefficients are uniformly bounded.
\end{proof}
\end{prop}
Using the previous definitions, we obtain the following connection for the scattering map between the physical space and the separated picture. 
\begin{theorem}\label{thm:fouriertophysical} The following diagram commutes and each arrow is a Hilbert space isomorphism:

\[
 \begin{tikzcd}
  \mathcal{E}^T_{\Hp_A} \oplus \mathcal{E}^T_{\Hp_B} \arrow{r}{S^T} \arrow[swap]{d}{\mathcal{F}_{\Ho_A}\oplus \mathcal{F}_{\Ho_B}} & \mathcal{E}^T_{\Ch_B} \oplus \mathcal{E}^T_{\Ch_A} \arrow{d}{\mathcal{F}_{\Ch_B}\oplus \mathcal{F}_{\Ch_A}} \\%
\hat{\mathcal{E}}^T_{\Hp_A} \oplus \hat{\mathcal{E}}^T_{\Hp_B}   \arrow{r}{\hat{S^T}} &   \hat{ \mathcal{E}}^T_{\Ch_B} \oplus \hat{ \mathcal{E}}^T_{\Ch_A}.
\end{tikzcd}
\]
Moreover, the maps $S^T$ and $\hat S^T$ are pseudo-unitary satisfying \eqref{eq:pseudounitary} and \eqref{eq:pseudounitaryode}, respectively. 
More concretely, for $(\Psi_A, \Psi_B) \in \mathcal{E}^T_{\Ho_A}\oplus \mathcal{E}^T_{\Ho_B}$, we can write
\begin{align}
	\begin{pmatrix}
	\Phi_B \\ \Phi_A
	\end{pmatrix} =  S^T 	\begin{pmatrix}
	\Psi_A \\\Psi_B
	\end{pmatrix}, \end{align}
	where $ \partial_u \Phi_A \in L^2(\mathcal{CH}_A)$ and $\partial_v \Phi_B\in L^2(\mathcal{CH}_B) $ can be represented by
	\begin{align}\nonumber 
	\partial_u	\Phi_A (u,\theta,\varphi )= & \frac{1}{\sqrt {2\pi}r_-} \int_{\mathbb R} \sum_{|m|\leq\ell}-i\omega  \mathfrak R(\omega,\ell) \,\mathcal{F}_{\mathcal{H}_A}(\Psi_A) (\omega,m,\ell) Y_{m\ell}(\theta,\varphi) e^{-i\omega u } \d\omega \\
		&+\frac{1}{\sqrt {2\pi} r_-} \int_{\mathbb R} \sum_{|m|\leq \ell} -i\omega \bar{\mathfrak T}(\omega,\ell) \,\mathcal{F}_{\mathcal{H}_B}(\Psi_B) (\omega,m,\ell) Y_{m\ell}(\theta,\varphi) e^{-i\omega u } \d\omega \label{eq:fourierrep1}
	\end{align}
	and 
	\begin{align}\nonumber\partial_v \Phi_B (v,\theta,\varphi )=& \frac{1}{\sqrt {2\pi} r_-}  \int_{\mathbb R} \sum_{|m|\leq \ell} i \omega  \mathfrak T(\omega,\ell) \,\mathcal{F}_{\mathcal{H}_A}(\Psi_A) (\omega,m,\ell) Y_{m\ell}(\theta,\varphi) e^{i\omega v} \d\omega \\
	&+\frac{1}{\sqrt {2\pi} r_-} \int_{\mathbb R} \sum_{|m|\leq \ell} i \omega \bar{\mathfrak R}(\omega,\ell) \,\mathcal{F}_{\mathcal{H}_B}(\Psi_B) (\omega,m,\ell) Y_{m\ell}(\theta,\varphi) e^{i\omega v } \d\omega\label{eq:fourierrep2}
	\end{align}
	as well as $ \Phi_A \in \mathcal{E}^T_{\mathcal{CH}_A} \cong \dot{H}^1(\mathbb R; L^2(\mathbb S^2)), \Phi_B \in \mathcal{E}^T_{\mathcal{CH}_B}\cong \dot{H}^1(\mathbb R; L^2(\mathbb S^2))$ can be represented by regular distributions as 
		\begin{align}\nonumber
			\Phi_A (u,\theta,\varphi )= & \frac{1}{\sqrt {2\pi}r_-} \operatorname{p.v.}\int_{\mathbb R} \sum_{|m|\leq \ell} \mathfrak R(\omega,\ell) \,\mathcal{F}_{\mathcal{H}_A}(\Psi_A) (\omega,m,\ell) Y_{m\ell}(\theta,\varphi) e^{-i\omega u } \d\omega \\
		&+\frac{1}{\sqrt {2\pi} r_-} \operatorname{p.v.}\int_{\mathbb R} \sum_{|m|\leq \ell}  \bar{\mathfrak T}(\omega,\ell) \,\mathcal{F}_{\mathcal{H}_B}(\Psi_B) (\omega,m,\ell) Y_{m\ell}(\theta,\varphi) e^{-i\omega u } \d\omega 
		\end{align}
		and 
		\begin{align}\nonumber \Phi_B (v,\theta,\varphi )=& \frac{1}{\sqrt {2\pi} r_-}  \operatorname{p.v.}\int_{\mathbb R} \sum_{|m|\leq \ell}  \mathfrak T(\omega,\ell) \,\mathcal{F}_{\mathcal{H}_A}(\Psi_A) (\omega,m,\ell) Y_{m\ell}(\theta,\varphi) e^{i\omega v} \d\omega \\
		&+\frac{1}{\sqrt {2\pi} r_-} \operatorname{p.v.}\int_{\mathbb R} \sum_{|m|\leq \ell}  \bar{\mathfrak R}(\omega,\ell) \,\mathcal{F}_{\mathcal{H}_B}(\Psi_B) (\omega,m,\ell) Y_{m\ell}(\theta,\varphi) e^{i\omega v } \d\omega.
\label{eq:scatteredcauchyb}
		\end{align}
		\begin{proof}
			This is a direct consequence of \cref{thm:forwardevolution}, \cref{thm:boundednesstrans} and \eqref{eq:scatteringcoef1}, \eqref{eq:scatteringcoef2} in the proof of \cref{thm:thmbounds}. 
		\end{proof}
\end{theorem}
From the previous representation of the scattered solution we can draw a link between the boundedness of the scattering map and the fact that compactly supported incoming data will lead to solutions which vanish on the future bifurcation sphere $\mathcal{B}_+$. This is the content of the following
\begin{cor}\label{cor:vanishingbifurcation}
	Let $\Psi= (\Psi_A,0)\in \mathcal{E}^T_{\Hp_A} \oplus \mathcal{E}^T_{\Hp_B}$ be purely incoming smooth data. Assume further that $\Psi_A$ is supported away from the past bifurcation sphere $\mathcal{B}_-$ and future timelike infinity $i^+$. 
	
	Then, the Cauchy evolution $\psi$ arising from $\Psi_A$ vanishes at the future bifurcation sphere $\mathcal{B}_+$. 
	
	On the other hand, if $\Psi$, as above, led to a solution which does not vanish at the future bifurcation sphere $\mathcal{B}_+$, then the scattering map $S^T: \mathcal{E}^T_{\Hp} \to \mathcal{E}^T_{\Ch}$ could not be bounded. 
	\begin{proof}
		The first claim is a direct consequence of \eqref{eq:scatteredcauchyb} in \cref{thm:fouriertophysical}.
	
	For the second statement let $\Psi_A$ be compactly supported data on the event horizon and assume that its Cauchy evolution $\psi$ does not vanish at the future bifurcation sphere $\mathcal{B}_+$. Now take data $\tilde \Psi_A$ which is supported away from the past bifurcation sphere $\mathcal{B}_-$ and satisfies $T \tilde \Psi_A = \Psi_A$. Then, $\tilde \Psi_A \in \mathcal{E}^T$ but its Cauchy evolution $\tilde \psi $ satisfies $\tilde \psi\restriction_{\Ch}\notin \mathcal{E}^T_{\Ch}$ since 
	\begin{align}
		\| \tilde \psi\restriction_{\Ch_B} \|^2_{\mathcal{E}^T_{\Ch_B}}  =  \int_{\mathbb R \times \mathbb S^2} |\psi\restriction_{\Ch_B} (v,\theta,\varphi)|^2 \d v \sin\theta \d \theta \d \varphi = \infty,
	\end{align}
	as $\psi\restriction_{\Ch_B} = T \tilde \psi \restriction_{\Ch_B}$ does not vanish at the future bifurcation sphere $\mathcal{B}_+$. By cutting off smoothly, one can construct normalized (in $\mathcal{E}^T_{\Hp}$-norm) smooth compactly supported initial data on $\mathcal{E}^T_{\Hp}$ such that its Cauchy evolution has arbitrary large norm $\mathcal{E}^T_{\Ch}$-norm at the Cauchy horizon. 
\end{proof}
\end{cor}
\begin{rmk} 
	For convenience we have stated the second statement of \cref{cor:vanishingbifurcation} only for the interior of Reissner--Nordström. However, note that it holds true for more general black hole interiors (e.g.\ subextremal (anti-) de~Sitter--Reissner--Nordström) with Penrose diagram as depicted in \cref{fig:penrosembar}. 
\end{rmk}
\subsection{Injectivity of the reflection map}
\label{subsec:reflection}
In this section, we define the reflection operator of purely incoming radiation (cf.\ \cref{fig:reflection}) and prove that it is has trivial kernel as an operator from $\mathcal{E}^T_{\Hp_A} \to \mathcal{E}^T_{\Ch_A}$.   \begin{figure}[ht]\centering
	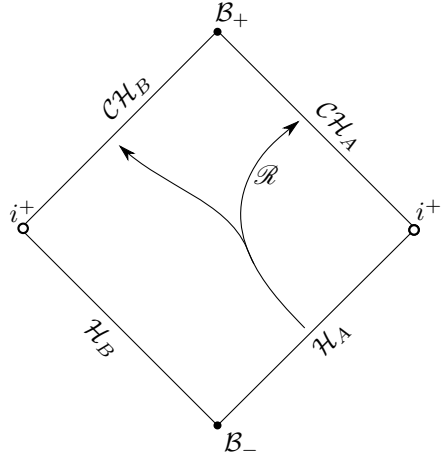
	\caption{Reflection $\mathscr{R}$ of purely incoming radiation.}
	\label{fig:reflection}
\end{figure} 
\begin{definition}[Reflection operator]\label{defn:reflectionoperator}
	For purely incoming radiation $(\Psi_A,0)\in \mathcal{E}^T_{\Hp_A} \oplus \mathcal{E}^T_{\Hp_B}$, define the reflection operator 
	\begin{align}
		\mathscr{R}: \mathcal{E}^T_{\Hp_A} \to \mathcal{E}^T_{\Ch_A}
	\end{align}
	as
	\begin{align}
		\mathscr{R}(\Psi_A) = \Phi_A := \operatorname{pr}_{A} \left( S^T  \begin{pmatrix}\Psi_A \\ 0
		\end{pmatrix}\right),
	\end{align}
	where $\operatorname{pr}_{A}\colon \mathcal{E}^T_{\Ch_B} \oplus \mathcal{E}^T_{\Ch_A} \to \mathcal{E}^T_{\Ch_A}  $ is the orthogonal projection.
\end{definition}
\begin{theorem}\label{thm:nonvanishingreflection}
The reflection operator $\mathscr{R}$ defined in \cref{defn:reflectionoperator} has trivial kernel. 
\begin{proof}
	Assume $\mathscr{R}(\Psi_A) =0$ for some $\Psi_A \in \mathcal{E}^T_{\Hp_A}$. Then, in view of \cref{thm:fouriertophysical}, \begin{align}\mathfrak R(\omega,\ell) \mathcal{F}_{\Hp_A} (\Psi_A) (\omega,m,\ell) =0\end{align} for all $m,\ell$, and a.e.\ $\omega \in \mathbb {R} $. Moreover, since $\mathfrak R(\omega,\ell)$ only vanishes on a discrete set (cf.\ \cref{cor:rdoesntvanish}), we obtain that $\mathcal{F}_{\Hp_A} (\Psi_A) (\omega,m,\ell) =0$ for all $m,\ell$, and a.e.\ $\omega \in \mathbb {R} $. Again, in view of \cref{thm:fouriertophysical}, we conclude $\Psi_A=0$ as an element of $\mathcal{E}^T_{\Hp_A}$.
\end{proof}
\end{theorem}

\subsection{\texorpdfstring{$C^1$}{C1}-blow-up on the Cauchy horizon}\label{sec:c1instab}
In this section, we shall revisit and discuss the seminal work  \cite{hartle1982crossing} of Chandrasekhar and Hartle. The Fourier representation of the scattered data on the Cauchy horizon in \cref{thm:fouriertophysical} serves as a good framework to provide a completely rigorous framework for the $C^1$-blow-up at the Cauchy horizon studied in~\cite{hartle1982crossing}.
\begin{theorem}[$C^1$-blow-up on the Cauchy horizon \cite{hartle1982crossing}]\label{thm:c1instab}
	There exist smooth, compactly supported and purely incoming data $\Psi_A$ on the event horizon $\mathcal{H}_A$ for which the Cauchy evolution of \eqref{eq:linearwave} fails to be $C^1$ at the Cauchy horizon $\mathcal{CH}$. More precisely, the solution $\psi$ arising from $\Psi_A$ fails to be $C^1$ at every point on the Cauchy horizon $\Ch_A\cup \mathcal{B}_+$. Moreover, the incoming radiation can be chosen to be only supported on any angular parameter $\ell_0$ which satisfies $\ell_0(\ell_0+1) \neq r_+^2 ( r_+ - 3 r_-)$. 
	\begin{proof}
		Let $\ell_0$ be fixed and satisfy $\ell_0(\ell_0+1) \neq r_+^2 ( r_+ - 3 r_-)$. Define purely incoming smooth data $\Psi_A(v,\theta,\varphi) = f(v) Y_{\ell_0 0}(\theta,\varphi)$ on $\mathcal{H}_A$, where $f(v)$ is smooth and compactly supported. Moreover, assume that the entire function $\hat f$ satisfies $\hat f(i\kappa_+) \neq 0$. In view of the representation formula \eqref{eq:scatteredcauchyb} from \cref{thm:fouriertophysical}, the degenerate derivative of its Cauchy evolution $\Phi_B$ on the Cauchy horizon $\Ch_B$  reads
		\begin{align}\label{eq:contour}
		\partial_v \Phi_B(v,\theta,\varphi) = \frac{r_+}{\sqrt{2\pi}r_-} \int_{\mathbb R} i \omega \mathfrak T(\omega,\ell_0) \hat f(\omega)e^{i\omega v} \d \omega Y_{\ell_0 0}(\theta,\varphi).
		\end{align}
		Since $\mathfrak T(\omega,\ell)$ has a simple pole at $\omega = i\kappa_+$ (cf.\ \cref{prop:rtboundedcomplex} in the appendix), we pick up the residue at $i\kappa_+$ when we deform the contour of integration in \eqref{eq:contour} from the real line to the line $\operatorname{Im}(\omega) = \kappa_+ + \delta$ for some $\kappa_+ >\delta >0$. Here we use that the compact support of $f(v)$ implies the bound $|\hat f(\omega)| \leq e^{|\operatorname{Im}(\omega)|\sup| \operatorname{supp}(f)|} \hat f(\operatorname{Re}(\omega))$ and that, in addition, by \cref{prop:rtboundedcomplex}, the transmission coefficient $\mathfrak T$ remains bounded as $|\operatorname{Re}(\omega)| \to \infty$. This ensures that the deformation of the integration contour is valid. Hence, 
		\begin{align}\nonumber
		\partial_v \Phi_B(v,\theta,\varphi) =& \frac{i r_+}{\sqrt{2\pi}r_-} 2 \pi i (i\kappa_+) \hat f(i\kappa_+) e^{-\kappa_+ v } Y_{\ell_0 0} (\theta,\varphi)  \operatorname{Res}(\mathfrak T(\omega,\ell_0), i\kappa_+)  \\\nonumber& + i
		\frac{r_+ e^{-(\kappa_+ +\delta )v}}{\sqrt{2\pi}r_-} \int_{\mathbb R}\Big[  (\omega_R + i(\kappa_+ + \delta)) \mathfrak T(\omega_R + i (\kappa_+ + \delta))
		\\\nonumber& \hat f(\omega_R + i (\kappa_+ + \delta)) e^{i \omega_R v} Y_{\ell_0 0} (\theta,\varphi) \Big] \d \omega_R\\ 
		&= C e^{-\kappa_+ v} Y_{\ell_0 0 }(\theta,\varphi) + o\left(e^{-(\kappa_+ + \delta) v}\right)
		\end{align}
		as $v\to\infty$,
		where 
		\begin{align}
		C = -i \kappa_+ \frac{r_+}{r_-}\sqrt{2\pi} \hat f(i\kappa_+) \operatorname{Res}(\mathfrak T(\omega,\ell_0), \omega = i\kappa_+)\neq 0
		\end{align} 
		by construction. 
		Thus, $\Phi_B$ is not in $C^1$ at the future bifurcation sphere as the non-degenerate derivative diverges as $v \to \infty$:
		\begin{align}
		\frac{\partial}{\partial V_{\Ch}} \Phi_B  = e^{-\kappa_- v} \partial_v \Psi_B(v,\theta,\varphi) = C e^{-(\kappa_+ + \kappa_-) v}  (1 + o(1))  ,
		\end{align}
		where we recall that $\kappa_- < - \kappa_+<0$. Finally, propagation of regularity gives that the solution is not in $C^1$ at each point on the Cauchy horizon $\Ch_A$. More precisely, expressing \eqref{eq:linearwave} is $(u,v)$ coordinates gives
		\begin{align}\label{eq:nullcords}
			\partial_u \partial_v \psi = \frac{-\Delta}{2r^3} (\partial_v \psi + \partial_u \psi) + \frac{\Delta}{4r^4}\ell_0(\ell_0 +1) \psi,
		\end{align}
		where $\Delta$ is as in \eqref{eq:h} and where we have used that $\Delta_{\mathbb S^2} \psi = -\ell_0(\ell_0 +1) \psi$. Now, note that $|\psi|, |\partial_u\psi|$ and $|\partial_v\psi|$ are uniformly bounded in the interior by a higher order norm of $\Psi_A$. This follows from \cite{franzen2016boundedness}, commuting with $T$ and angular momentum operators as well as elliptic estimates. Finally, integrating  \eqref{eq:nullcords} in $u$, using the estimate $|\Delta|\lesssim e^{\kappa_- (u+v)}$ for $r_\ast \geq 0$ (see \eqref{eq:decayindelta}) and using the non-degenerate coordinate $V_\Ch$ gives the $C^1$ blow-up also everywhere on $\Ch_A$.
			\end{proof}
\end{theorem}
\subsection{Breakdown of \texorpdfstring{$T$}{T} energy scattering for cosmological constants \texorpdfstring{$\Lambda\neq 0$}{Lambdaneq0}}
\label{subsec:nonex}
Interestingly, the analogous result to  \cref{thm:forwardevolution} on the interior of a subextremal (anti-) de Sitter--Reissner--Nordstr\"om black hole does not hold for almost all cosmological constants $\Lambda$. In the presence of a cosmological constant it is also natural to consider the Klein--Gordon equation with conformal mass $\mu = \frac{3}{2} \Lambda$. We will consider in fact a general 
 mass term of the form $\mu = \nu \Lambda $, where $\nu \in \mathbb{R}$. Note that $\nu = \frac{3}{2}$  corresponds to the conformal invariant Klein--Gordon equation. To be more precise, we prove that for generic subextremal black hole parameters $(M,Q,\Lambda)$, there exists a normalized (in  $\mathcal{E}^T_{\Ho}$-norm) sequence of Schwartz initial data on the event horizon for which the $\mathcal{E}^T_{\Ch}$-norm of the evolution restricted to the Cauchy horizon blows up.
 
We define a black hole parameter triple $(M,Q,\Lambda)$ to be \emph{subextremal} if \begin{align}\label{defn:subextremal}(M,Q,\Lambda) \in \mathcal{P}_{\mathrm{se}} := \mathcal{P}_{\mathrm{se}}^{\Lambda =0} \cup \mathcal{P}_{\mathrm{se}}^{\Lambda >0} \cup \mathcal{P}_{\mathrm{se}}^{\Lambda < 0},\end{align} where
 \begin{align}
 	\nonumber \mathcal{P}_{\mathrm{se}}^{\Lambda =0} := &\{ (M,Q,\Lambda ) \in \mathbb{R}_+ \times \mathbb{R} \times \{ 0 \} \colon\\ & \Delta(r):= r^2 - 2Mr  + Q^2 \text{ has two positive simple roots satisfying } 0 < r_- < r_+. \} ,
\\ \nonumber 
 \mathcal{P}_{\mathrm{se}}^{\Lambda >0} := &\{ (M,Q,\Lambda) \in \mathbb{R}_+ \times \mathbb{R} \times \mathbb{R}_+ \colon\\ & \Delta(r):= r^2 - 2Mr - \frac 1 3 \Lambda r^4 + Q^2 \text{ has three positive simple roots satisfying } 0 < r_- < r_+ < r_c\} \label{defn:subextremal>0},\\\nonumber
 \mathcal{P}_{\mathrm{se}}^{\Lambda <0} := &\{ (M,Q,\Lambda) \in \mathbb{R}_+ \times \mathbb{R} \times \mathbb{R}_-  \colon\\ & \Delta(r):= r^2 - 2Mr - \frac 13 \Lambda r^4 + Q^2 \text{ has two positive roots satisfying } 0 < r_- < r_+  \}. \label{defn:subextremal<0}
 \end{align}
\begin{restatable}{theorem}{cosmological}
	\label{thm:cosmological}Let $\nu\in\mathbb{R}$ be a fixed Klein--Gordon mass parameter. (In particular, we may choose $\nu = \frac{3}{2}$ to cover the conformal invariant case or $\nu =0$ for the wave equation \eqref{eq:linearwave}.) Consider the interior of a subextremal (anti-) de Sitter--Reissner--Nordstr\"om  black hole with generic parameters $(M,Q,\Lambda) \in \mathcal{P}_{\mathrm{se}}\setminus D(\nu)$.  (Here, $ D(\nu)\subset \mathcal{P}_{\mathrm{se}} $ is a set with measure zero defined in \cref{prop:Bneq0} (see \cref{sec:cosmo}). Moreover $D(\nu)$ satisfies $ \mathcal{P}_{\mathrm{se}}^{\Lambda =0} \subset D(\nu) $ and $U \cap D(\nu) = \mathcal{P}_{\mathrm{se}}^{\Lambda =0}$ for some open set $U\subset \mathcal{P}_{\mathrm{se}}$.)
	
	  Then, there exists a sequence $(\Psi_n)_{n\in \mathbb{N}}$ of purely ingoing and compactly supported  data on $\mathcal{H}_A$ with 
	\begin{align}
	\| \Psi_n \|_{ \mathcal{E}^T_{\Ho}} = 1 \text{ for all } n
	\end{align}such that the solution $\psi_n$ to the Klein--Gordon equation with mass $\mu = \nu \Lambda$
	\begin{align}
		\Box_{g_{M,Q,\Lambda}}\psi  - \mu \psi = 0
	\end{align}
 arising from 	$\Psi_n$ has unbounded $T$ energy at the Cauchy horizon 
	\begin{align}
			\| \psi_n \restriction_{\mathcal{CH}} \|_{ \mathcal{E}^T_{\Ch}} \to \infty \text{ as } n\to\infty.
	\end{align}
\end{restatable}
\begin{proof}
	See \cref{sec:cosmo}.
\end{proof}
\begin{rmk}
	Note that from \cref{thm:cosmological} it also follows that for fixed $0<|Q|<M$, the $T$ energy scattering breaks down (in sense of \cref{thm:cosmological}) for all cosmological constants $0<|\Lambda|<\epsilon$, where $\epsilon = \epsilon(M,Q)>0$ is small enough.
\end{rmk}
\subsection{Breakdown of \texorpdfstring{$T$}{T} energy scattering for the Klein--Gordon equation}
\label{subsec:notkleingordon}
Finally, we will also prove that the $T$ energy scattering theory does not hold for the Klein--Gordon equation for a generic set of masses $\mu$, even in the case of vanishing cosmological constant $\Lambda =0$. 
\begin{restatable}{theorem}{kleingordon}
\label{thm:kleingordon}Consider the interior of a subextremal Reissner--Nordstr\"om black hole.
	There exists a discrete set $\tilde D(M,Q) \subset \mathbb R$ with $0\in \tilde D$ such that the following holds true.
	For any $\mu \in \mathbb R \setminus \tilde D$  there exists a sequence $(\Psi_n)_{n\in \mathbb{N}}$ of purely ingoing and compactly supported  data on $\mathcal{H}_A$ with 
	\begin{align}
	\| \Psi_n \|_{ \mathcal{E}^T_{\Ho}} = 1 \text{ for all } n
	\end{align}such that the solution $\psi_n$ to the Klein--Gordon equation with mass $\mu$
	\begin{align}
	\Box_{g_{M,Q,\Lambda}}\psi  - \mu \psi = 0
	\end{align}
	arising from 	$\Psi_n$ has unbounded $T$ energy at the Cauchy horizon 
	\begin{align}
	\| \psi_n \restriction_{\mathcal{CH}} \|_{ \mathcal{E}^T_{\Ch}} \to \infty \text{ as } n\to\infty.
	\end{align}
\end{restatable}	
\begin{proof}
See \cref{sec:kleingordonequation}.
\end{proof}
The above \cref{thm:cosmological} and \cref{thm:kleingordon} show that the existence of a $T$ energy scattering theory for the wave equation \eqref{eq:linearwave} on the interior of Reissner--Nordstr\"om is in retrospect a surprising property. Implications of the non-existence of a $T$ energy scattering map and in particular, the unboundedness of the scattering map in the cosmological setting $\Lambda \neq 0$, are yet to be understood.
\section{Proof of \texorpdfstring{\cref{thm:boundednesstrans}}{Theorem 2}: Uniform boundedness of the transmission and reflection coefficients}
\label{sec:radial}
This section is doteevoted to the proof of \cref{thm:boundednesstrans}. We will analyze solutions to the o.d.e.\ (recall from \eqref{eq:radialode})
\begin{align*}
\Delta \frac{\d}{\d r}\left(\Delta \frac{\d}{\d r} R\right) - \Delta \ell(\ell +1) R + r^4 \omega^2 R =0.
\end{align*} 
This o.d.e.\ can be written equivalently (recall from \eqref{ODE1}) as
\begin{align*}
u^{\prime \prime} +(\omega^2 - V_\ell) u =0,
\end{align*} in the $r_\ast$ variable, where $u=r R$.

For the convenience of the reader we recall the statement of \cref{thm:boundednesstrans}. 
 \boundedness*
The proof of \cref{thm:boundednesstrans} will involve different arguments for different regimes of parameters. Also, note that in view of \eqref{eq:r0} and \eqref{eq:t0} it is enough to assume $\omega \neq 0$. 

The first regime for bounded frequencies ($|\omega|\leq \omega_0$, $\ell$ arbitrary) requires the most work. One of its main difficulties is to obtain estimates which are uniform in the limit $\ell \to \infty$. We shall use that the o.d.e.\ \eqref{ODE1} with $\omega =0$, which reads
\begin{align}
u^{\prime\prime} - V_\ell u = 0,
\end{align}
can be solved explicitly in terms of Legendre polynomials and Legendre functions of second kind. The specific algebraic structure of the Legendre o.d.e.\ leads to the feature that solutions which are bounded at $r_\ast=-\infty$ are also bounded at $r_\ast=+\infty$. For generic perturbations of the potential this property fails to hold. Nevertheless, for perturbations of the form as in \eqref{ODE1} for $\omega \neq 0$ and $|\omega|\leq |\omega_0|$, this behavior survives and most importantly, can be quantified. To prove this we will essentially divide the real line $\mathbb{R}\ni r_\ast$ into three regions.

 The first region will be near the event horizon ($r_\ast = - \infty$), where we will consider the potential $V_\ell$ as a perturbation. The second region will be the intermediate region, where we will consider the term involving $\omega$ as a perturbation. Finally, in the third region near the Cauchy horizon ($r_\ast = +\infty$), we consider the potential $V_\ell$ as a perturbation again. This eventually allows us to prove the uniform boundedness of the reflection and transmission coefficients $\mathfrak R$ and $\mathfrak T$ in the bounded frequency regime $|\omega| < \omega_0$.

The second regime will be bounded angular momenta and $\omega$-frequencies bounded from below $(|\omega|\geq \omega_0, \ell\leq \ell_0)$. For this parameter range we will consider $V_\ell$ as a perturbation of the o.d.e.\ since $V_\ell$ might only grow with $\ell$, which is, however, bounded in that range. Again, this allows us to show uniform boundedness for the transmission and reflection coefficients $\mathfrak T$ and $\mathfrak R$.

The third regime will be angular momenta and frequencies both bounded from below ($|\omega| \geq \omega_0$, $\ell \geq \ell_0$). To prove boundedness of reflection and transmission coefficients $\mathfrak R$ and $\mathfrak T$, we will consider $\frac{1}{\ell}$ as a small parameter to perform a WKB-approximation.

\subsection{Low frequencies \texorpdfstring{($|\omega|\leq \omega_0$)}{wleqw0}}
\label{subsec:smallfreq}
We first analyze the o.d.e.\ for the special case of vanishing frequency. Then, we will summarize properties of special functions, which we will need to finally prove the boundedness of reflection and transmission coefficients in the low frequency regime. Let \begin{align}0<\omega_0 \leq \frac 12\end{align} be a fixed constant.
\subsubsection{Explicit solution for vanishing frequency (\texorpdfstring{$\omega = 0$}{w=0})}
For $\omega=0$ we can explicitly solve the o.d.e.\ with special functions.
In that case the o.d.e.\ reads
\begin{align}\label{eq:ODEw0}
\frac{\d}{\d r} \left( \Delta \frac{\d R}{\d r}\right) -  \ell ( \ell +1)  R =0.
\end{align} We define the coordinate $x(r)$ as
\begin{align}
\label{eq:x(r)}
x(r):= - \frac{2r}{r_+ - r_-} + \frac{r_+ + r_-}{r_+ - r_-}
\end{align}
or equivalently,
\begin{align}
r(x) =- \frac{r_+ - r_-}{2} x+ \frac{r_+ + r_-}{2}.
\end{align}
Then, we can write
\begin{align}
\Delta(x ) = \left(\frac{r_+ - r_-}{2}\right)^2(x+1)(x-1) = \left(\frac{r_+ - r_-}{2} \right)^2(x^2 -1).
\end{align}
Hence, \cref{eq:ODEw0} reduces to the Legendre o.d.e.\
\begin{align}\label{eq:legendreode}
\frac{\d}{\d x} \left((1-x^2) \frac{\d R}{\d x} \right) + \ell(\ell+1) R = 0.
\end{align}
We will denote by $P_\ell(x)$ and $Q_\ell(x)$ the two independent solutions, the Legendre polynomials and the Legendre functions of second kind, respectively \cite{olver2014asymptotics,NIST:DLMF}.
We will prove later in \cref{prop:intermediate} that $\tilde u_1$ and $\tilde u_2$ from \cref{defn:u1u2} satisfy
\begin{align}\label{eq:uP}
&\tilde u_1(r_\ast) = w_1(r_\ast) := (-1)^\ell \frac{r(r_\ast)}{r_+} P_\ell(x(r_\ast)),\\
&\tilde u_2(r_\ast) = w_2(r_\ast) := (-1)^\ell \frac{r(r_\ast)}{k_+ r_+} Q_\ell(x(r_\ast)).\label{eq:uQ}
\end{align}
These are a fundamental pair of solutions for the o.d.e.\ in the case $\omega =0$. We will perturb these explicit solutions for the regime of low frequencies ($|\omega|\leq \omega_0$). To do so, we will need properties about special functions which will be considered first.

In view of the fact that $\omega_0$ is fixed, constants appearing in $\lesssim$ and $\gtrsim$ may also depend on $\omega_0$. 
Before we begin, we shall summarize the special functions we will use and list their relevant properties in the case $|\omega| \leq \omega_0$. 
\subsubsection{Special functions}
Good references for the following discussion are \cite{abramowitz1964handbook,olver2014asymptotics,NIST:DLMF}.
First, we shall recall the definition of the Gamma and Digamma function.
\begin{definition}For $z\in \mathbb{C}$ with $\operatorname{Re}(z)>0$ we denote the Gamma function with $\Gamma(z)$ and will also make use of the Digamma function $	\digamma(z )$ defined as
	\begin{align}
		\digamma(z ) := \int_0^\infty \left( \frac{e^{-x}}{x}  - \frac{e^{-zx}}{1-e^{-x}} \right) \d x. 
	\end{align}
	Note that
	\begin{align}
		\digamma(z+1) - \digamma(z) =  \frac 1 z
	\end{align}
	and \begin{align}
	\digamma(n)  = \sum_{k=1}^{n-1} \frac{1}{k} - \gamma = \log(n) + O(n^{-1}),
	\end{align}
	where $\gamma$ is the Euler--Mascheroni constant.  
\end{definition}
As we mentioned above, we shall use the Legendre polynomials and the Legendre functions of second kind.
We will express them in terms of the hypergeometric function $\mathbf{F}(a,b;c;x)$ for $x\in (-1,1)$, $a,b,c \in \mathbb R$ as defined in \cite[Equation~(9.3)]{olver2014asymptotics}.
\begin{definition}[Legendre functions of first and second kind] We use the standard conventions which are used in  \cite{olver2014asymptotics,NIST:DLMF}.
	
	For $x\in(-1,1)$, we define the associated Legendre polynomials by
	\begin{align}
		P_\ell^m(x) =\left( \frac{1+x}{1-x}\right)^{\frac m2} \mathbf{F}\left(\ell + 1 ,  - \ell ; 1-m; \frac{1-x}{2}\right)
	\end{align}
	and the associated Legendre functions of second kind by
	\begin{align}
		Q^m_\ell(x) = - \frac{1}{2}\pi \sin\left(\frac 12 \pi (\ell +m) \right) w_1(\ell,x)  + \frac 12 \pi \cos\left(\frac 12 (\ell +m) \pi \right) w_2(\ell,x).
	\end{align}
	Here,
	\begin{align}
	&	w_1(\ell,x) = \frac{2^m \Gamma(\frac{\ell + m+ 1}{2} )}{\Gamma({1+\frac{\ell}{2}})} (1-x^2)^{-\frac m2} \mathbf{F}\left( -\frac{\ell + m}{2}, \frac{1+\ell-m}{2};\frac 12 ; x^2\right),\\
	&w_2(\ell,x) = \frac{2^m\Gamma({1+\frac{\ell+m}{2}})}{\Gamma(\frac{\ell-m + 1}{2} )}x (1-x^2)^{-\frac m2} \mathbf{F}\left(\frac{1-\ell-m}{2}, 1 + \frac{\ell-m}{2}; \frac 32 ; x^2\right).
	\end{align}
	
\end{definition}
We shall also use the convention $P_\ell = P_\ell^0$ and $Q_\ell^m = Q_\ell^0$. Also, recall the symmetry
\begin{align}
&	P_\ell( x) = (-1)^\ell P_\ell(-x),\\
& Q_\ell(x) = (-1)^{\ell+1} Q_\ell(-x).
\end{align}In the asymptotic expansion in the parameter $\ell$ for the Legendre polynomials and functions  we will make use of Bessel functions which we define in the following.

\begin{definition}[Bessel functions of first and second kind] Recall the Bessel functions of first kind
	\begin{align}
		&J_0(x) := \sum_{k=0}^\infty \frac{x^{2k} }{(-4)^k k!^2},\\
	&J_1(x) := \frac{x}{2} \sum_{k=0}^\infty \frac{x^{2k} }{(-4)^k k! (k+1)!},
	\end{align}
	and the Bessel functions of second kind
	\begin{align}
	Y_0(x) := 	&\frac{2}{\pi} J_0(x)   \left(  \log\left( \frac{x}{2}\right)  + \gamma \right) - \frac{2}{\pi} \sum_{k=1}^\infty H_k \frac{x^{2k} }{(-4)^k (k!)^2},
	\\ \nonumber Y_1(x) :=&  -\frac{1}{2\pi x} + \frac{2}{\pi} \log\left(\frac{x}{2}\right) J_1(x)\\&\;\;\; - \frac{x}{2\pi} \sum_{k=0}^\infty( \digamma(k+1)  + \digamma(k+2)) \frac{x^{2k}}{(-4)^k k! (k+1)!},
	\end{align}
	where $H_k = \sum_{n=1}^k n^{-1}$ is the $k$-the harmonic number.
	We have the asymptotic expansions
	\begin{align}\label{eq:estimatea}
		&J_0(x) = 1 + O(x^2),\\
		&J_1(x) = \frac{x}{2} + O(x^3),\label{eq:estimateJ1} \\
		&Y_0(x) = \frac{2}{\pi} \log\left(\frac{x}{2}\right) + O(1),\\
		&Y_1(x) = -\frac{1}{2\pi x} +  o(1) \text{ as } x\to 0.\label{eq:estimateY1}
	\end{align}
	Note that bounds deduced from \eqref{eq:estimatea} -- \eqref{eq:estimateY1} hold uniformly on any interval $(0,a]$ of finite length. 
	We shall also use the bounds 
	\begin{align}\label{eq:boundsonJY1}
		&|J_0(x)|\leq 1 ,\,
		|Y_0(x)| \lesssim  1+ |\log(x)| 
		\end{align}
		for  $0< x\leq 1$ and 
		\begin{align}\label{eq:boundsonJY2}
		&|J_0(x)| \lesssim \frac{1}{\sqrt x}, \,
		|Y_0(x)| \lesssim \frac{1}{\sqrt{x}} \; 
	\end{align}
	for  $x\geq 1$  \cite[p.\ 360, p.\ 364]{abramowitz1964handbook}.
\end{definition}
In the proof we will also use the following asymptotic formulae for $P_\ell$ and $Q_\ell$ for large $\ell$ in terms of Bessel functions. 
\begin{lemma}\cite[\S 14.15(iii)]{NIST:DLMF} \label{lem:plandql}We have  
 \begin{align}
 	&P_\ell(\cos\theta ) = \left(\frac{\theta}{\sin\theta}\right)^{\frac 12} \left(J_0\left(\frac{\theta(2\ell + 1)}{2} \right)  + e_{1,\ell} ( \theta) \right),\\
 	 	&Q_\ell(\cos\theta ) =-\frac{\pi}{2} \left(\frac{\theta}{\sin\theta}\right)^{\frac 12} \left(Y_0\left(\frac{\theta(2\ell + 1)}{2} \right)  + e_{2,\ell} ( \theta) \right),\label{eq:estimateql}\\
 	 	& Q_\ell^1(\cos\theta) = - \frac{\pi}{2\ell}\left(\frac{\theta}{\sin\theta}\right)^{\frac 12} \left( Y_1\left(\frac{\theta(2\ell + 1)}{2}\right) + e_{3,\ell}(\theta) \right), \label{eq:estimateql1}
 \end{align}
 where the error terms can be estimated by
 \begin{align}
& |e_{1,\ell}(\theta)| ,\, 	|e_{2,\ell}(\theta) | \lesssim \frac{1}{1+\ell}\left[ \left|J_0 \left( \frac{\theta ( 2\ell +1)}{2}\right)\right| + \left|Y_0\left( \frac{\theta ( 2\ell +1)}{2}\right)\right|\right] ,\label{eq:estimateerror}
\\
 & |e_{3,\ell}(\theta) | \lesssim \frac{1}{1+\ell} \left[ \left| J_1 \left( \frac{\theta ( 2\ell +1)}{2}\right)\right| + \left|Y_1\left( \frac{\theta ( 2\ell +1)}{2}\right)\right| \right]\label{eq:errorestimatee3}
 \end{align}
 for $\theta\in ( 0, \pi - \delta)$ and for any fixed $\delta>0$. In particular, this holds uniformly as $\theta \to 0$. 
\end{lemma}
We shall use the following asymptotic formulae for the Legendre functions at the singular endpoints. 
\begin{lemma}\cite[\S 14.8]{NIST:DLMF} For $0<x<1$ we have
	\begin{align}
	&	P_\ell(x) = 1 + f_1(x), \\
	&	Q_\ell(x) = \frac 12 (\log(2) - \log(1-x) ) - \gamma - \digamma(\ell +1) + f_1(x),\label{eq:largeellQ}
	\end{align}
	where $|f_1(x)|\lesssim_\ell  (1-x)$. Moreover, analogous results hold true for $-1<x<0$ due to symmetry.
\end{lemma}
Now, we will estimate the derivatives of the Legendre polynomials and Legendre functions of second kind.
\begin{lemma} For  $x\in (-1,1)$ we have
\begin{align}\label{eq:derivativelegendre1}
	&\left|\frac{\d P_\ell }{\d x}\right| \leq \ell^2.\end{align}
For $x_{\alpha, \ell}:=1-\frac{\alpha}{1+\ell^2}$ with $0<\alpha<1$ and $\ell \in \mathbb N$ we have 
	\begin{align}
    & ( 1- (\pm x_{\alpha,\ell})^2) \left| \frac{\d Q_\ell}{\d x}(\pm x_{\alpha, \ell})\right| \lesssim 1.\label{eq:derivativelegendre2}
\end{align}
\begin{proof}Inequality \eqref{eq:derivativelegendre1} is known as Markov's inequality and is proven in \cite[Theorem 5.1.8]{borwein2012polynomials}. We only have to prove \eqref{eq:derivativelegendre2} for $x=+x_{\alpha,\ell} $ due to symmetry. 
	From the recursion relation \cite[\S 14.10]{NIST:DLMF} we have 
	\begin{align}\nonumber
	(\ell + 1)^{-1} & (1-x_{\alpha,\ell}^2)	\frac{\d Q_\ell }{\d x} (x_{\alpha,\ell}) = x_{\alpha,\ell} Q_\ell(x_{\alpha,\ell}) -  Q_{\ell + 1}(x_{\alpha,\ell})  \\ &= (x_{\alpha,\ell}-1)Q_\ell(x_{\alpha,\ell}) + (Q_\ell(x_{\alpha,\ell}) - Q_{\ell +1}(x_{\alpha,\ell})).\label{eq:derivativeQ}
	\end{align}
We will consider both summands separately. 

\textbf{Part 1}: Summand $(x_{\alpha,\ell}-1)Q_\ell(x_{\alpha,\ell})$

\noindent
First, consider $1-x_{\alpha,\ell}=\frac{\alpha}{1+\ell^2}$, where we implicitly define $\cos(\theta_{\alpha,\ell}) = x_{\alpha,\ell}$. Note that we have \begin{align}\nonumber
	\theta_{\alpha,\ell}(x) = \sqrt{2(1-x_{\alpha,\ell})} + O((1-x_{\alpha,\ell})^\frac{3}{2}) &= \sqrt{\frac{2\alpha}{1+\ell^2}} + O\left(\left( \frac{\alpha}{1+\ell^2} \right)^{\frac 32 } \right)\\ & =
	 \sqrt{\frac{2\alpha}{1+\ell^2}} \left( 1+ O\left(\frac{\alpha}{1+\ell^2}\right)\right).\label{eq:theta}
\end{align}
In particular, we have $\theta_{\alpha,\ell} \ell \lesssim 1$. 
 This gives
\begin{align}
-Q_\ell(x_{\alpha,\ell}) =- Q_\ell(\cos\theta_{\alpha,\ell})  = \frac{\pi}{2} \left(\frac{\theta_{\alpha,\ell}}{\sin\theta_{\alpha,\ell}}\right)^{\frac 12 }\left( Y_0\left(\frac{\theta_{\alpha,\ell}(2\ell +1)}{2}\right) + e_{2,\ell}(\theta_{\alpha,\ell}) \right).
\end{align}
Again, we will look at the two terms independently. First, note that
\begin{align}\nonumber
	\frac{\pi}{2} &\left( \frac{\theta_{\alpha,\ell}}{\sin \theta_{\alpha,\ell}}\right)^{\frac 12} \left( Y_0\left(\theta_{\alpha,\ell}\left(\ell + \frac 12 \right)\right) \right)\\
	\nonumber & = \frac{\pi}{2} \left( \frac{\theta_{\alpha,\ell}}{\sin\theta_{\alpha,\ell}}\right)^\frac{1}{2} \left( \frac{2}{\pi} \log\left( \frac{ \theta_{\alpha,\ell}(2\ell + 1)}{4}\right)  + O(1) \right)\\\nonumber& =
	\left( 1+ O(\theta_{\alpha,\ell}^2)\right) \left(  \log(\theta_{\alpha,\ell}) + \log\left(\ell + \frac 12 \right) + O(1) \right)\\\nonumber & =
	\left( 1+ O\left(\frac{\alpha}{1+\ell^2}\right)\right) \left( \frac 12 \log\left(\frac{\alpha}{1+\ell^2}\right)  + \log\left(\ell+ \frac 12 \right) + O(1) \right)\\ \nonumber & =
\left( 1+ O\left(\frac{\alpha}{1+\ell^2}\right)\right) \left( \frac{1}{2} \log(\alpha) + \frac{1}{2} \log\left(1 + \frac{\ell- \frac{3}{4}}{\ell^2 + 1 }\right)  + O(1) \right)\\ & =
\frac{1}{2} \log(\alpha) + O(1).
\end{align}
In order to estimate $e_{2,\ell}(\theta_{\alpha,\ell})$ we shall recall inequality~\eqref{eq:estimateerror}. It works analogously to the previous estimate up to a good term of $\frac{1}{1+\ell}$. In particular, this shows
\begin{align}
	|Q_\ell(x_{\alpha,\ell})| \lesssim | \log(\alpha) | + 1
\end{align}
and 
\begin{align}
	|(x_{\alpha,\ell}-1)Q_\ell(x_{\alpha,\ell})| \lesssim \frac{\alpha}{1+\ell^2} (|\log(\alpha) | + 1)\lesssim \frac{1}{1+\ell^2}.
\end{align}

\textbf{Part 2}: Summand $(Q_\ell(x_{\alpha,\ell}) - Q_{\ell +1}(x_{\alpha,\ell}))$

\noindent
Using the recursion relation for the difference of two Legendre function \cite[\S14.10]{NIST:DLMF}, we have 
\begin{align}
	(\ell +1) (Q_\ell(x_{\alpha,\ell}) - Q_{\ell +1}(x_{\alpha,\ell}) = - ( 1- x_{\alpha,\ell}^2)^\frac{1}{2} Q_\ell^1(x_{\alpha,\ell}) + (1-x_{\alpha,\ell}) Q_\ell(x_{\alpha,\ell}) .
\end{align}
We estimate the term $(1-x_{\alpha,\ell})Q_\ell(x_{\alpha,\ell})$ by what we have done above as
\begin{align}
|	(1-x_{\alpha,\ell}) Q_\ell(x_{\alpha,\ell}) |\lesssim \frac{\alpha}{1+\ell^2}( |\log(\alpha) |+ 1) \lesssim 1.
\end{align}
For the term $-(1-x_{\alpha,\ell}^2)^{\frac 12} Q_\ell^1(x_{\alpha,\ell})$ we use \eqref{eq:estimateql1} to get
\begin{align}\nonumber
&\left|-(1-x_{\alpha,\ell}^2)^{\frac 12} Q_\ell^1(x_{\alpha,\ell})\right|\\ &\lesssim\sqrt{ 	\frac{\alpha}{\ell^2 + 1}} \frac{1}{1+\ell} \left(1+O\left(\frac{\alpha}{1+\ell^2}\right)\right) \left( Y_1\left( \left(\ell + \frac 12\right) \theta_{\alpha,\ell} \right) + e_{2,\ell} (\theta_{\alpha,\ell})\right).
\end{align}
As before, we shall start estimating the first term using \eqref{eq:estimateY1} and \eqref{eq:theta} to obtain
\begin{align}\nonumber
	&\sqrt{ 	\frac{\alpha}{\ell^2 + 1}} \frac{1}{1+\ell}  \left(1+O\left(\frac{\alpha}{1+\ell^2}\right)\right)  Y_1\left( \left(\ell + \frac 12\right) \theta_{\alpha,\ell} \right) \\\nonumber &=  	\sqrt{ 	\frac{\alpha}{\ell^2 + 1}} \frac{1}{1+\ell}  \left(1+O\left(\frac{\alpha}{1+\ell^2}\right)\right) \left( -\frac{1}{\pi(2 \ell +1 ) \theta_{\alpha,\ell} }  + O(1) \right) \\&
	\lesssim \sqrt{ 	\frac{\alpha}{\ell^2 + 1}} \frac{1}{1+\ell}  \left(  \frac{1}{\sqrt{\alpha} } + 1 \right)\lesssim 1.
\end{align}
We estimate the second term using \eqref{eq:errorestimatee3}, \eqref{eq:estimateJ1}, \eqref{eq:estimateY1}, and \eqref{eq:theta} to obtain
\begin{align}\nonumber
&\left|\sqrt{ \frac{\alpha}{\ell^2 + 1} } \frac{1}{1+\ell}\left(1+O\left(\frac{\alpha}{1+\ell^2}\right)\right) e_{2,\ell} (\theta_{\alpha,\ell} )\right|\\ & \lesssim \sqrt{ \frac{\alpha}{\ell^2 + 1} } \frac{1}{1+\ell^2} \left( \frac{1}{\sqrt \alpha} + 1 \right) \lesssim 1.
\end{align}

We have estimated that $|Q_\ell(x_{\alpha,\ell}) - Q_{\ell + 1}(x_{\alpha,\ell}) |\lesssim \frac{1}{1+\ell}$ which proves the claim in view of \eqref{eq:derivativeQ}. 
\end{proof}
\end{lemma}
Finally, we prove asymptotics for the derivatives of the Legendre of functions of second kind near the singular points.
\begin{lemma}For $0<x<1$ and $x \to 1$ we have
	\begin{align}
(1-x^2) \frac{\d Q_\ell}{\d x}	 =  1 + O_\ell ( (1-x) \log(1-x)).
	\end{align}
	By symmetry this also yields for $-1<x<0$ and $x\to -1$
	\begin{align}
		(1-x^2)  \frac{\d Q_\ell}{\d x}	 =  (-1)^\ell  + O_\ell ( (1+x) \log(1+x)). \label{eq:dqdxlimit}
	\end{align}
	\begin{proof}
		From the recursion relation \cite[\S 14.10]{NIST:DLMF} and \eqref{eq:largeellQ} we obtain
		\begin{align}\nonumber
		& (1-x^2)	\frac{\d Q_\ell }{\d x}  = (\ell +1)(x Q_\ell -  Q_{\ell + 1} ) 
		\\ \nonumber & = (\ell + 1) (x-1)Q_\ell +  (\ell + 1)(Q_\ell - Q_{\ell +1})\\ \nonumber &= 
		 (\ell + 1) ( Q_\ell - Q_{\ell + 1} ) + O_\ell ((1-x) \log(1-x))\\ \nonumber & 
		 = (\ell + 1)  ( \digamma(\ell + 2) - \digamma(\ell + 1) )  + O_\ell( (1-x) \log (1-x) ) \\ & = 1  + O_\ell( ( 1-x) \log(1-x)).
		\end{align}
	\end{proof}
\end{lemma}
Having reviewed the required facts about special functions, we shall now proceed to prove the uniform boundedness of the reflection and transmission coefficients.
\subsubsection{Boundedness of the reflection and transmission coefficients}
\label{sec:boundedrefltransmi}
As mentioned before, we will consider three different regions: a region near the event horizon, an intermediate region, and a region near the Cauchy horizon. In $r_\ast$ coordinates we separate these regions at \begin{align}R_1^\ast(\omega,\ell) := \frac{1}{2 \kappa_+} \log\left(\frac{\omega^2}{1+\ell^2}\right)\label{defn:R1}\end{align} and \begin{align}R_2^\ast(\omega,\ell):= \frac{1}{2\kappa_-}\log\left(\frac{\omega^2}{1+\ell^2 }\right)\label{defn:R2}\end{align}
for $0<|\omega| < \omega_0$ and $\ell \in \mathbb N_0$. 
Note that $-\infty < R_1^\ast(\omega,\ell) < 0 < R_2^\ast(\omega,\ell) < \infty$. 

\paragraph{\textbf{Region near the event horizon}}
\begin{prop}
	\label{prop:firstregion}
	Let $0<|\omega|<\omega_0$ and $\ell\in\mathbb N_0$. Then, we have 
	\begin{align}
&	\| u_1^\prime \|_{L^\infty(-\infty, R_1^\ast)} \lesssim |\omega|,\\
&		\|u_1\|_{L^\infty(-\infty,R_1^\ast)} \lesssim 1.
	\end{align}
\end{prop}
\begin{proof}
Recall the defining Volterra integral equation for $u_1$ from \cref{defn:u1u2}
 \begin{align}
 	u_1(r_\ast) = e^{i\omega r_\ast } + \int_{-\infty}^{r_\ast} \frac{\sin(\omega(r_\ast-y))}{\omega} V(y) u_1(y) \d{y}.
 \end{align}
 with integral kernel 
 \begin{align}
K(r_\ast , y) := \frac{\sin(\omega(r_\ast - y)) }{\omega} V(y).
 \end{align}
From \cref{lem:asymptoticspotential} in the appendix, we obtain for $r_\ast \leq R_1^\ast$  
\begin{align}
	|V(r_\ast)| \lesssim e^{2k_+ r_\ast} (1 + \ell^2) 
\end{align}
and in particular,
\begin{align}
	|V(R_1^\ast) |\lesssim e^{2k_+ R_1^\ast} ( 1+ \ell^2)=  \omega^2 .
\end{align}
This implies for $r_\ast \leq R_1^\ast$ 
 \begin{align}
 	|K(r_\ast,y)| \leq \frac{1}{|\omega|}|  V(y) | \lesssim  
 	\frac{1}{|\omega|}( 1+ \ell^2  ) e^{2k_+ y}
 \end{align}
 and thus,
 \begin{align}
 	\int_{-\infty}^{R_1^\ast} \sup_{y < r_\ast < R_1^\ast} |K(r_\ast,y)| \d{y} \lesssim \frac{\ell^2+1}{|\omega|}e^{2k_+ R_1^\ast} \lesssim 1.
 \end{align}
The claim follows now from \cref{lem:volterra}.
 \end{proof}
 Now, we would like to consider $\omega$ as a small parameter and perturb the explicit solutions for the $\omega =0$ case in order to propagate the behavior of the solution through the intermediate region, where $V_\ell$ is large compared to $\omega$. In particular, $V_\ell$ can be arbitrarily large since $\ell$ is not bounded above in the considered parameter regime.
 \paragraph{\textbf{Intermediate region}}
 First, recall the following fundamental pair of solutions which is based on the Legendre functions of first and second kind
 \begin{align}
	& w_1(r_\ast) :=(-1)^\ell \frac{r(r_\ast)}{r_+} P_\ell(x(r_\ast)), \\
	& w_2(r_\ast) := (-1)^\ell \frac{r(r_\ast)}{k_+ r_+} Q_\ell(x(r_\ast)),
\end{align}
where $P_\ell$ and $Q_\ell$ are the Legendre polynomials and Legendre functions of second kind, respectively. Our first claim is that we have constructed this fundamental pair $(w_1,w_2)$ to have unit Wronskian and moreover $\tilde u_1 = w_1$ and $\tilde u_2 = w_2$ holds true. 
\begin{prop}\label{prop:intermediate}We have $w_1=\tilde u_1$ and $w_2= \tilde u_2$ and 
	the Wronskian of $u_1$ and $u_2$ satisfies\begin{align}
	\mathfrak	W(w_1,w_2 ) = \mathfrak W (\tilde u_1, \tilde u_2) = 1.
	\end{align}
	Similarly, we also have  $\tilde  v_1 = (-1)^\ell \frac{r_+}{r_-} w_1 =  (-1)^\ell \frac{r_+}{r_-} \tilde u_1$.
	\begin{proof}We first prove that $\mathfrak W(w_1,w_2) = 1$.
		Since the Wronskian is independent of $r_\ast$, we will compute its value in the limit $r_\ast \to-\infty$. In this proposition $\ell$ is fixed and we shall allow implicit constants in $\lesssim$ to depend on $\ell$.  Clearly, 
		\begin{align}
			w_1(r_\ast ) \to 1 \text{ as } r_\ast \to-\infty.
		\end{align}
		Moreover, we have that for $r_\ast \leq 0$
		\begin{align}
			\left|\frac{\d}{\d r_\ast} w_1(r_\ast) \right| \lesssim e^{2k_+ r_\ast }| P_\ell(x(r_\ast))  | + \left|\frac{\d P_\ell( x) }{\d x}(r_\ast) \frac{\d x}{\d r_\ast } (r_\ast)\right| \lesssim e^{2k_+ r_\ast } ,
		\end{align}
		where we have used \eqref{eq:derivativelegendre1}. This, in particular, also shows that $w_1$ satisfies the same boundary conditions ($w_1\to 1$, $w_1^\prime \to 0$ as $r_\ast \to-\infty$) as $\tilde u_1$ defined in \cref{defn:u1u2} and thus, $w_1$ and $\tilde u_1$ have to coincide. Similarly, we can deduce $\tilde  v_1 = (-1)^\ell \frac{r_+}{r_-} w_1$.
		
		 For $w_2$, we use  \eqref{eq:largeellQ} to obtain
		\begin{align}
			|w_2(r_\ast) - r_\ast | \lesssim \left( -  \frac{r(r_\ast) }{k_+ r_+} \left(\frac{1}{2}  \log\left( \frac{2}{1+x(r_\ast)}\right) -\gamma - \digamma(\ell+1)   \right) - r_\ast \right) + e^{2 k_+ r_\ast}.
		\end{align} 
		For an intermediate step, we compute $\log( 1+ x(r_\ast))$ from \eqref{eq:x(r)} near $r_\ast = -\infty$. In particular, for the limit $r_\ast \to -\infty$, we can assume that $r_\ast \leq 0$ and thus, $r - r_- \gtrsim r_+ - r_- $. Hence, 
		\begin{align}\nonumber 
			\log( 1+ x(r_\ast) ) & = \log\left( 1 +  \frac{(r_+ - r )+( r_- - r)}{r_+ - r_-}\right)  \\ \nonumber & = \log\left( 1 +  \frac{f(r_\ast) }{r_+ - r_-}e^{2k_+ r_\ast } + \frac{r_- - r}{r_+ - r_-}\right) \\ \nonumber &=  \log\left(   \frac{2 f(r_\ast) }{r_+ - r_-}e^{2k_+ r_\ast } \right) \\ &= 2 k_+ r_\ast + \log( 2 f(r_\ast) (r_+ - r_-)^{-1}),
		\end{align}
		where $f$ is defined in \eqref{eq:fofr}.
		Thus, this directly implies
		\begin{align}
			|w_2(r_\ast) - r_\ast | \lesssim r_\ast e^{2 k_+ r_\ast  } +  1  \lesssim 1.
		\end{align}
		Finally, we claim that $w_2^\prime  \to 1$ as $r_\ast \to-\infty$. We shall use estimate \eqref{eq:dqdxlimit} near $x(r_\ast)=-1$ to obtain
		\begin{align}\nonumber
			|w_2^\prime&(r_\ast) - 1 |  
			\lesssim e^{2k_+ r_\ast}(  |r_\ast | + 1)   + \left| (-1)^\ell \frac{r(r_\ast) }{k_+ r_+} \frac{\d Q_\ell(x)}{\d x}\frac{\d x}{\d r_\ast} - 1 \right|  \\ & \lesssim e^{2k_+ r \ast}   + \left| \frac{  r(r_\ast) }{k_+ r_+} \left[1+ O\left((1+x(r_\ast))\log(1+x(r_\ast)) \right)  \right]\frac{1}{1-x^2(r_\ast)}\frac{\d x}{\d r_\ast} - 1  \right| .
		\end{align} 
		Now, in order to conclude that 
		\begin{align}
				|w_2^\prime(r_\ast) - 1 | \to 0 \text{ as } r_\ast \to -\infty,
		\end{align}
		it suffices to check that 
		\begin{align}
			\frac{1}{1-x^2(r_\ast)} \frac{\d x}{\d r_\ast} \to k_+ \text{ as } r_\ast \to-\infty.
		\end{align}
		But this holds true because
		\begin{align}
				\frac{1}{1-x^2(r_\ast)} \frac{\d x}{\d r_\ast} = 
				\frac{1}{1-x^2(r_\ast)}  \frac{-2}{r_+ - r_- } \frac{\Delta}{r^2 }   = \frac{r_+ -r_-}{2r^2} \to k_+ \text{ as } r_\ast \to -\infty. 
		\end{align}
		Now, this implies that 
		\begin{align}
		\mathfrak	W(w_1,w_2) = \lim_{r_\ast \to -\infty} \left( w_1 w_2^\prime - w_1^\prime w_2\right)= 1,
		\end{align}
		and moreover, that $w_2= \tilde u_2$ as they satisfy the same boundary conditions at $r_\ast = - \infty$. 
	\end{proof}
\end{prop}
Having proved the Wronskian condition we are in the position to define the perturbations of $\tilde u_1$ and $ \tilde u_2$ to non-zero frequencies. 
\begin{definition}\label{def:perurb}Define perturbations $\tilde u_{1,\omega}$ and $\tilde u_{2,\omega}$ of $\tilde u_1$ and $\tilde u_2$ (cf.\ \eqref{eq:uP} and \eqref{eq:uQ}) in the intermediate region by the unique solutions to the Volterra equations
	\begin{align}\label{eq:u1}
	\tilde u_{1,\omega} (r_\ast) =\tilde u_1(r_\ast) + \omega^2 \int_{R_1^\ast}^{r_\ast} \left(\tilde u_1(r_\ast) \tilde{u}_2(y) - \tilde u_1(y) \tilde {u}_2(r_\ast) \right) \tilde u_{1,\omega}(y) \d{y}
	\end{align}
and 
	\begin{align}\label{eq:u2}
	\tilde{u}_{2,\omega }(r_\ast) = \tilde{u_2}(r_\ast )  + \omega^2 \int_{R^\ast_1}^{r_\ast} \left(\tilde u_1(r_\ast) \tilde{u}_2(y) - \tilde u_1(y) \tilde{u}_2(r_\ast) \right)  \tilde{u}_{2,\omega}(y) \d{y}.
	\end{align}
\end{definition}

\begin{prop}\label{prop:boundedAB}
	Let $0<|\omega|<\omega_0$ and $\ell\in\mathbb N_0$, then we have for $r_\ast \in [R_1^\ast,R_2^\ast]$
	\begin{align}
	u_1(\omega, r_\ast) = A(\omega,\ell)\tilde u_{1,\omega}(r_\ast) + B(\omega,\ell) \omega \tilde{u}_{2,\omega}(r_\ast),
	\end{align}
	where \begin{align}|A(\omega,\ell)| + |B(\omega,\ell)| \lesssim 1.\end{align}
	\begin{proof}
	First, note that by construction in \cref{def:perurb} we have
	\begin{align}
	&	\tilde u_{1,\omega}(R_1^\ast ) = \tilde u_1(R_1^\ast),\\
	&	\tilde	u_{1,\omega}^\prime (R_1^\ast ) =\tilde u_1^\prime(R_1^\ast),\\
	&	\tilde u_{2,\omega}(R_1^\ast) = \tilde{u}_2(R_1^\ast),\\
	&		\tilde	 u_{2,\omega}^\prime(R_1^\ast) = \tilde{u}_2^\prime(R_1^\ast).
	\end{align}
Now,	we want to estimate the previous terms. 
	By construction, we directly have that 
	\begin{align}
	|\tilde	u_1(R_1^\ast )| \leq 1.
	\end{align}
Then, note that
	\begin{align}\label{eq:1plusx}	\frac{\omega^2}{\ell^2 + 1}\lesssim 	1+x(R_1^\ast) \lesssim \frac{\omega^2}{\ell^2 + 1}.
	\end{align}
	Hence, from \eqref{eq:largeellQ}, we obtain
	\begin{align}
	| 	\tilde u_2(R_1^\ast) |\lesssim 1 +  \left|-\frac 12 \log(1+x(R_1^\ast))  - \digamma(\ell+1) \right| \lesssim 1 + |\log(|\omega|)|\lesssim \log\left(\frac{1}{|\omega|}\right),
	\end{align}
	where we have used that for $\ell \geq 1$ we have  $\digamma(\ell + 1) = \log(\ell) + \gamma + O(\ell^{-1})$.
	For $\tilde u_2^\prime(R_1^\ast)$ we have the estimate
	\begin{align}
|\tilde	u_2^\prime(R_1^\ast)| & \lesssim | \Delta(R_1^\ast)  Q_\ell(x(R_1^\ast))| +\left| \frac{\d Q_\ell}{\d x}(R_1^\ast) \frac{\d x}{\d r_\ast}(R_1^\ast)\right| 
\lesssim 1 ,\label{eq:analogous}
		\end{align}
		where we have used \eqref{eq:derivativelegendre2} and \eqref{eq:1plusx} as well as the fact that \begin{align}\frac{\d x}{\d r_\ast} (1-x(r_\ast)^2)^{-1}\lesssim 1.\end{align}		
			Now, we can express $A$ via the Wronskian as \begin{align}
			|	A|  =\left|\frac{ \mathfrak W(u_1,\tilde u_{2,\omega}) }{\mathfrak W(\tilde u_{1,\omega}, \tilde u_{2,\omega})} \right|.
			\end{align}
			By construction, we have $\mathfrak W(\tilde u_{1,\omega},\tilde u_{2,\omega})=\mathfrak W(\tilde u_1,\tilde u_2) = 1$. Hence, using \cref{prop:firstregion} we conclude
			\begin{align}
			|A| \leq |  u_1(R_1^\ast ) \tilde u_{2,\omega}^\prime(R_1^\ast) | + | u_1^\prime(R_1^\ast) \tilde u_{2,\omega}(R_1^\ast )| \lesssim |\tilde u_{2}^\prime(R_1^\ast) | + |\omega\tilde u_{2}(R_1^\ast )|.
			\end{align}
		Thus, we conclude 
		\begin{align}
			|A| \lesssim 1.
		\end{align}
		Note that from \eqref{eq:derivativelegendre1}, we have
		\begin{align}\label{eq:u1primer1}
			|\tilde u_1^\prime(R_1^\ast)| \lesssim  \left|\left(1 +\frac{\d P_\ell}{\d x}\right) \frac{\d x}{\d r_\ast}\right| \lesssim (1+\ell^2) \frac{\omega^2}{1+\ell^2}\leq \omega^2.
		\end{align}
		Hence, we can also estimate $B$ by
		\begin{align}\nonumber
			|B| &= \frac{1}{|\omega|}| \mathfrak W(u_1,\tilde u_{1,\omega})|\lesssim \frac{1}{|\omega|} \left( | \tilde u_1^\prime(R_1^\ast) | +| \omega \tilde u_1(R_1^\ast)| \right)
			\\&\lesssim 1 + \frac{1}{|\omega|} |\tilde u_1^\prime(R_1^\ast) |\lesssim 1,
	\end{align}
	where we used \cref{prop:firstregion} again.
	\end{proof}
\end{prop}
For the intermediate region we will need the following result in order to get uniform bounds for the Volterra iteration.
\begin{lemma}\label{lem:u1u2}
		Let $0<|\omega|<\omega_0$ and $\ell\in\mathbb N_0$, then \begin{align}
	 &\int_{R_1^\ast}^{R_2^\ast}|\tilde u_1(r_\ast)| \d{r_\ast} \lesssim \log^2\left(\frac{1}{|\omega|}\right)\label{eq:l1u1},\\
	 &	 \int_{R_1^\ast}^{R_2^\ast} |\tilde u_2(r_\ast)| \d{r_\ast} \lesssim \log^2\left(\frac{1}{|\omega|}\right)\label{eq:l1u2}.
\end{align}
\begin{proof}We first prove \eqref{eq:l1u1}.
	We shall split the integral in two regions. The first region is from $r_\ast = R_1^\ast$ to $r_\ast =0$. In that region we define $\theta \in (0,\frac{\pi}{2}]$ such that $\cos(\theta) = - x(r_\ast)$. Using also \cref{lem:plandql} we obtain
	\begin{align}\nonumber
		|\tilde u_1(r_\ast)| &\lesssim |P_\ell (x(r_\ast)) | = | P_\ell( -x(r_\ast))| = | P_\ell (\cos\theta)| \\ & \lesssim \left| \left( \frac{{\theta}}{\sin\theta}\right)^{\frac 12} J_0( (\ell + \frac 12) \theta)\right| + | e_{1,\ell}(\theta)|.
		\label{eq:estimateu1}
	\end{align}
	The last term shall be treated as an error term. Thus,
	\begin{align}
		\nonumber	\int_{R_1^\ast}^{0}|\tilde u_1(r_\ast)| \d{r_\ast}&
			 \lesssim \int_{x(R_1^\ast)}^{0} |P_\ell(x)| \frac{1}{1+x} \d{x}\leq \int_{-1  + C \frac{\omega^2}{1+\ell^2}}^{0} |P_\ell(-x)| \frac{1}{1+x} \d{x}\\\nonumber& \lesssim \int_{\arccos(1-C\frac{\omega^2}{ 1+\ell^{2} })}^{\frac \pi 2}| P_\ell(\cos\theta)| \frac{1}{1-\cos\theta} \sin\theta\, \d{\theta}\\
			& \leq \int_{C_1 \frac{|\omega|}{1+\ell}}^{\frac \pi 2} |P_\ell(\cos\theta)| \frac{\sin\theta}{1-\cos\theta}  \d\theta.
\end{align}
Here, $C$ and $C_1$ are positive constants only depending on the black hole parameters. 
We further estimate using equation~\eqref{eq:estimateu1}
\begin{align}\nonumber 
&\int_{R_1^\ast}^{0}|\tilde u_1(r_\ast)| \d{r_\ast}\\
	 & \lesssim \int_{C_1 \frac{\omega}{1+\ell}}^{\frac \pi 2 }\left( \frac{\theta}{\sin\theta}\right)^{\frac 12}\left|J_0( (\ell + \frac 12) \theta)\right|  \frac{\sin\theta}{1-\cos\theta}\d\theta + \text{Error} ,
	 \end{align}
	 where we will take care of the term
	 \begin{align}
	 	\text{Error} = \int_{C_1 \frac{\omega}{1+\ell}}^{\frac \pi 2 } |e_{1,\ell}(\theta)|
	 \end{align}
	 later. First, we look at the term 	
	 \begin{align}\nonumber&\int_{C_1 \frac{\omega}{1+\ell}}^{\frac \pi 2 }\left( \frac{\theta}{\sin\theta}\right)^{\frac 12}\left|J_0\left( \left(\ell + \frac 12\right) \theta\right)\right|  \frac{\sin\theta}{1-\cos\theta}\d\theta\\\nonumber 
	 & \lesssim \int_{C_1 \frac{\omega}{1+\ell}}^{\frac \pi 2 } \frac{1}{\theta} \left|J_0\left( \left(\ell + \frac 12\right) \theta\right)\right| \d\theta\\\nonumber 
	&\lesssim \int_{C_1\omega }^{\frac \pi 2 (\ell +1) } \frac 1 \theta \left|J_0\left(\frac{\ell + \frac 12 }{\ell + 1} \theta\right)\right| \d\theta 
	\\ \nonumber & \lesssim 	\int_{C_1\omega}^{1} \frac{\left|J_0\left(\frac{\ell + \frac 12 }{\ell + 1} \theta\right)\right|}{\theta} \d\theta +	\int_{1}^{\infty} \frac{\left|J_0\left(\frac{\ell + \frac 12 }{\ell + 1} \theta\right)\right|}{\theta} \d\theta \label{eq:J0}\\& \lesssim
	\int_{C_1\omega}^1 \frac{1}{\theta} \d\theta  + \int_1^\infty \frac{1}{\theta^\frac{3}{2}}  \d\theta
	\lesssim  |\log(|\omega|)| ,
	\end{align}
	where we have used equation~\eqref{eq:boundsonJY1} and \eqref{eq:boundsonJY2}.
	Now, we are left with the error term 
	\begin{align}\nonumber 
		\text{Error}& \leq \frac{1}{1+\ell}  \int_{C_1 \frac{\omega}{\ell+1}}^{\frac \pi 2} \frac{\sin\theta}{1-\cos\theta} (|J_0((\ell+\frac 12 )\theta)| + |Y_0( ( \ell+\frac 12)\theta)| ) \d \theta \\ \nonumber & \lesssim \frac{1}{1+\ell} \int_{C_1 \frac{\omega}{\ell+1}}^{\frac \pi 2}\frac{\sin\theta}{1-\cos\theta} ( 1+ |\log(|\omega|)|)\d\theta 
		\lesssim \frac{|\log(|\omega|)| }{1+\ell }\int_{C_1 \frac{\omega}{\ell + 1}}^\frac{\pi}{2} \frac{1}{\theta} \d \theta
		 \\ &\lesssim \frac{\log^2(|\omega|) + \log(1+\ell)}{1+\ell}\lesssim \log^2\left(\frac{1}{|\omega|}\right).
	\end{align}
	Thus, 
	\begin{align}
	\int_{R_1^\ast}^0 |\tilde u_1(r_\ast)| \d{r_\ast} \lesssim \log^2\left( \frac{1}{|\omega|}\right) .\end{align}
	Completely analogously, we can compute 
	
	\begin{align}
	\int_{0}^{R_2^\ast} |\tilde u_1(r_\ast)| \d{r_\ast} \lesssim \log^2\left( \frac{1}{|\omega|}\right) .\end{align}
	The proof of equation \eqref{eq:l1u1} is completely similar up to a term which involves
	\begin{align}
		\int_{C_1 \omega}^1 \frac{\left| Y_0\left( \frac{\ell + \frac 12}{\ell + 1}\theta\right) \right|}{\theta} \d\theta \lesssim \log^2\left(\frac{1}{|\omega|}\right)
	\end{align}
	appearing in the estimate analogous to \eqref{eq:J0}. 
\end{proof}
\end{lemma}
With the help of the previous lemma we can now bound our solution $u_1$ at $R_2^\ast$. This results in 
\begin{prop}\label{prop:ur2}	Let $0<|\omega|<\omega_0$ and $\ell\in\mathbb N_0$, then
\begin{align}
	\| u_1 \|_{L^\infty(R_1^\ast, R_2^\ast)}  \lesssim 1 \text{ and } |u_1^\prime|(R_2^\ast) \lesssim | \omega |.
\end{align}
\begin{proof}
	Recall that we have from \cref{prop:boundedAB} for $r_\ast \in [R_1^\ast,R_2^\ast]$
\begin{align}
u_1(\omega, r_\ast ) = A(\omega,\ell) \tilde u_{1,\omega}(r_\ast) + \omega B(\omega,\ell) \tilde u_{2,\omega}(r_\ast)\label{eq:u1expansion}
\end{align}
for some uniformly bounded (in $|\omega | \leq \omega_0$ and $\ell$) constants $A,B$.
In particular, from \cref{lem:volterra} and \cref{rmk:volterra} we obtain the  bound
\begin{align}
	\|\tilde u_{1,\omega}\|_{L^\infty(R_1^\ast,R_2^\ast)}  \leq e^{\alpha} \|\tilde u_1\|_{L^\infty(R_1^\ast, R_2^\ast)}
\end{align}
for 
\begin{align}
	\alpha =\omega^2 \int_{R_1^\ast}^{R_2^\ast} \sup_{ \{ r_\ast\vert y \leq r_\ast \leq R_{2}^\ast \}}|\tilde u_1(r_\ast) \tilde u_2(y) -\tilde u_1(y)\tilde u_2(r_\ast) | \d{y}.
\end{align}
First, we have the bound\begin{align}
	\|\tilde u_1\|_{L^\infty(R_1^\ast, R_2^\ast)}\leq 1.\label{eq:u1tilde}
\end{align}
Secondly, for $r_\ast \in [R_1^\ast, R_2^\ast]$ we have \begin{align}1- x(r_\ast) \gtrsim \frac{\omega^2}{1+\ell^2}\end{align}
and
\begin{align}
	1+x(r_\ast) \gtrsim  \frac{\omega^2}{1+\ell^2}.
\end{align}
Consider the case $x(r_\ast) \geq 0$ first and implicitly define $\theta(r_\ast)$ by $\cos\theta(r_\ast) = x(r_\ast)$. Then, in view of \eqref{eq:estimateql} and $\theta(x(r_\ast)) = \sqrt{2-2x(r_\ast)} + O ((1-x(r_\ast)^{\frac 32} ))$, we estimate
\begin{align}
|	\tilde u_2(r_\ast) |\lesssim |Q_\ell(\cos(\theta(r_\ast)))| \lesssim \left| Y_0\left( \frac{\theta(r_\ast) (2\ell + 1)}{2} \right) \right|\lesssim \left| Y_0 \left( C {|\omega|} \right) \right|
\end{align} for a $C=C(M,Q)>0$. Analogously, this also holds for $x(r_\ast) < 0$ such that \eqref{eq:boundsonJY1} and \eqref{eq:boundsonJY2} imply
\begin{align}
	\|\tilde u_2\|_{L^\infty(R_1^\ast, R_2^\ast)}\lesssim \log\left(\frac{1}{|\omega|}\right). \label{eq:estu2}
\end{align}
Together with \cref{lem:u1u2} we obtain
\begin{align}
	\alpha \lesssim 1.
\end{align}
Hence, 
\begin{align}
	\| \tilde u_{1,\omega}\|_{L^\infty(R_1^\ast, R_2^\ast)} \lesssim 1
\end{align}
and similarly,
\begin{align}
		\|\tilde u_{2,\omega}\|_{L^\infty(R_1^\ast,R_2^\ast)} \lesssim \log\left(\frac{1}{|\omega|}\right). \label{eq:estimu2w}
\end{align}
This shows $\|u_1\|_{L^\infty(R_1^\ast,R_2^\ast)} \lesssim 1$ in view of \eqref{eq:u1expansion}.

Now, we are left with the derivative $u_1^\prime(R_2^\ast)$. To do so, we start by estimating $\tilde u_1^\prime(R_2^\ast)$ and  $\tilde u_2^\prime(R_2^\ast)$. Using the analogous estimate as we did for $R_1^\ast$ in \eqref{eq:analogous} and \eqref{eq:u1primer1}, we obtain
\begin{align}\label{eq:u2r2}
|	\tilde u^\prime_2 (R_2^\ast)| \lesssim 1 \text{ and } |\tilde u^\prime_{1} (R_2^\ast) | \lesssim \omega^2.
\end{align} 
Note that 
\begin{align}
	\tilde u_{2,\omega}^\prime(R_2^\ast) = \tilde u_2^\prime(R_2^\ast) + \omega^2 \int_{R_1^\ast}^{R_2^\ast} \left( \tilde u_1^\prime(R_2^\ast) \tilde u_2(y) - \tilde u_1(y) \tilde u_2^\prime(R_2^\ast) \right) \tilde u_{2,\omega}(y) \d y
\end{align} and thus in view of \cref{lem:u1u2}, \eqref{eq:u2r2}, \eqref{eq:estimu2w}, \eqref{eq:estu2}, and \eqref{eq:u1tilde} we estimate 
\begin{align}\nonumber
|\tilde	u_{2,\omega}^\prime(R_2^\ast ) |&\leq|\tilde u_2^\prime(R_2^\ast)| + \omega^2 \log\left(\frac{1}{|\omega|}\right)\int_{R_1^\ast}^{R_2^\ast} |\tilde u_1^\prime(R_2^\ast)\tilde u_2(y)|  + | \tilde u_1(y)\tilde u_2^\prime(R_2^\ast)| \d{y}\\ &\lesssim 1+ \omega^2 \, | \log(|\omega|)| \,  ( \omega^2 \log^2(|\omega|) + \log^2(|\omega|)) \lesssim 1. 
\end{align}
Similarly, we obtain 
\begin{align}\nonumber
|\tilde	u_{1,\omega}^\prime(R_2^\ast ) |&\leq|\tilde  u_1^\prime(R_2^\ast)| + \omega^2 \int_{R_1^\ast}^{R_2^\ast} |\tilde u_1^\prime(R_2^\ast) \tilde u_2(y)|  + | \tilde u_1(y) \tilde u_2^\prime(R_2^\ast)| \d{y}\\ &\lesssim \omega^2+ \omega^2 ( \omega^2 \log^2(|\omega|) + \log^2(|\omega|)) \lesssim | \omega|
\end{align}
which concludes the proof in the light of \eqref{eq:u1expansion}.
\end{proof}\end{prop}
\paragraph{\textbf{Region near the Cauchy horizon}}
Completely analogously to \cref{prop:firstregion}, we have
\begin{prop}
	\label{prop:thirdregion}
	Let $0<|\omega|<\omega_0$ and $\ell\in\mathbb N_0$. Then, we have 
	\begin{align}
	&	\| v_1^\prime \|_{L^\infty(R_2^\ast,\infty)} \lesssim |\omega|,\;\;
		\|v_1\|_{L^\infty(R_2^\ast,\infty)} \lesssim 1
	\end{align}
	and 
		\begin{align}
		&	\| v_2^\prime \|_{L^\infty(R_2^\ast,\infty)} \lesssim |\omega|,\;\;
			\|v_2\|_{L^\infty(R_2^\ast,\infty)} \lesssim 1.
		\end{align}
\end{prop}
\paragraph{\textbf{Boundedness of the scattering coefficients}}
Finally, we conclude that the reflection and transmission coefficients are uniformly bounded for parameters $0<|\omega|<\omega_0$ and $\ell \in \mathbb{N}_0$. 
\begin{prop}
	\label{prop:easyregime}
	We have
	\begin{align}
		\sup_{0<|\omega|<\omega_0, \ell \in \mathbb N_0} (|\mathfrak R(\omega, \ell)| + |\mathfrak T(\omega,\ell)| )\lesssim 1.
	\end{align}
\begin{proof} Let $0<|\omega|<\omega_0$ and $\ell \in \mathbb N_0$ and recall \cref{defn:TandR}. 
Then, \cref{prop:ur2} and \cref{prop:thirdregion} imply
\begin{align}
&	|\mathfrak T| \lesssim \frac{|\mathfrak W(u_1,v_2)|}{|\omega|} \leq \frac{|u_1(R_2^\ast)v_2^\prime(R_2^\ast)| +|u_1^\prime(R_2^\ast)v_2(R_2^\ast)| }{|\omega|}\lesssim 1\end{align}
and 
\begin{align}
&	|\mathfrak R| \lesssim \frac{|\mathfrak W(u_1,v_1)|}{|\omega|} \leq \frac{|u_1(R_2^\ast)v_1^\prime(R_2^\ast)| +|u_1^\prime(R_2^\ast)v_1(R_2^\ast)| }{|\omega|}\lesssim 1.
\end{align}
\end{proof}
\end{prop}
\subsection{Frequencies bounded from below and bounded angular momenta (\texorpdfstring{$|\omega|\geq \omega_0, \ell \leq \ell_0$}{wgeqw0,lleql0})}
Now, we will consider parameters of the form $|\omega|\geq \omega_0$ and $\ell \leq \ell_0$, where $\omega_0$ is small and determined from \cref{subsec:smallfreq}. Also, the upper bound on the angular momentum $\ell_0$ will be determined from \cref{subsec:boundedbelow}. As before, constants appearing in $\lesssim$ and $\gtrsim$ may depend on $\omega_0$.
\begin{prop} \label{prop9}We have
	\begin{align}
	\sup_{\omega_0\leq|\omega|, \ell\leq \ell_0} (|\mathfrak R(\omega, \ell)| + |\mathfrak T(\omega,\ell)| )\lesssim 1.
	\end{align}
		\begin{proof}
			Recall the definition of $u_1$ as the unique solution to
			\begin{align}\label{eq:integralbound}
				u_1(\omega, r_\ast ) = e^{i\omega r_\ast } + \int_{-\infty}^{r_\ast } \frac{\sin(\omega(r_\ast -y))}{\omega} V(y) u_1(\omega, y) \d y .
			\end{align}
			Note that in the regime $\ell \leq \ell_0$ we have a bound of the form 
			\begin{align}
			|	V(r_\ast)|\lesssim e^{-2 \min(k_+,|k_-|) |r_\ast|}
			\end{align}
			which implies the following bound on the integral kernel of the perturbation in \eqref{eq:integralbound}
			\begin{align}
				|K(r_\ast,y)| =\left| \frac{\sin(\omega(r_\ast - y))}{\omega} V(y)\right| \lesssim |V(y)|
			\end{align}
			in view of $|\omega| \geq \omega_0$. Thus, 
			\begin{align}
				\int_{-\infty}^{\infty} \sup_{r_\ast \in \mathbb R} |K(r_\ast,y)| \d y\lesssim \int_{-\infty}^{\infty} |V(y)| \d y\lesssim 1.
			\end{align}
			Hence, from \cref{lem:volterra} we deduce
			\begin{align}\label{eq:boundonu}
		\|u_1\|_{L^\infty(\mathbb R)} \lesssim 1 
			\end{align}
			and 
			\begin{align}\label{eq:boundsonuprime}
			\|u_1^\prime \|_{L^\infty(\mathbb R)}  \lesssim |\omega|.
			\end{align}
Note that we have obtained similar, indeed even stronger bounds for $u_1$ as in \cref{prop:ur2}. An argument completely similar to \cref{prop:easyregime} allows us to conclude.
\end{proof}
\end{prop}
\subsection{Frequencies and angular momenta bounded from below (\texorpdfstring{$|\omega| \geq \omega_0$, $\ell \geq \ell_0$}{wgeqw0,lgeql0})}
\label{subsec:boundedbelow}
In this regime we assume $\omega\geq \omega_0$ and $\ell \geq \ell_0$, where we choose $\ell_0$ large enough such that $V_\ell < 0$ everywhere. Note that such an $\ell_0$ can be chosen only depending on the black hole parameters. 

We write the o.d.e.\ as \begin{align}
u^{\prime \prime} = - (\omega^2 - V_\ell) u
\end{align}
and will represent the solution of the o.d.e.\ via a WKB approximation. For concreteness we will use the following theorem which is a slight modification of \cite[Theorem 4]{olver1961error}. 
\begin{lemma}[Theorem 4 of \cite{olver1961error}]\label{thm:wkb}
	Let $p\in C^2(\mathbb{R})$ be a positive function such that  
	\begin{align}
	F(x) = \left|\int_{-\infty}^x p^{-\frac{1}{4}} \left|\frac{\d^2}{\d x^2} \left( p^{-\frac 14}\right)\right| \d y\right|
	\end{align}
	satisfies $\sup_{x\in\mathbb{R}}F(x)< \infty$.
	 Then, the differential equation 
	\begin{align}
		\frac{\d^2 u (x) }{\d x^2} = - p(x) u(x)
	\end{align}
has conjugate solutions $u$ and $\bar u$ such that 
\begin{align}
	&u(x) =  p^{-\frac 14}\left( \exp\left(i  \int^x_0 \sqrt p(y) \d y \right) +\epsilon \right), \\
	&u^\prime (x) =  i  p^\frac{1}{4}\left[ \exp\left(i \int_0^x \sqrt{p}(y) \d y \right) -i \eta + \frac{i p^\prime}{4  p^\frac{3}{2}}\left( \exp\left(- i  \int_0^x \sqrt p (y) \d y \right) + \epsilon \right) \right],
\end{align}
where 
\begin{align}
	|\eta(x)|, |\epsilon(x)|\leq  \exp\left(F(x)\right) - 1.\label{eq:boundsonetaolver}
\end{align}
\end{lemma}

\begin{prop}\label{prop:reg3}Let $\omega_0 \leq |\omega|$ and $\ell \geq \ell_0$. Assume without loss of generality that $\omega >0$. Then,
\begin{align}
&u_1(\omega, r_\ast) = A \omega^{\frac 12}  (\omega^2 - V(r_\ast))^{-\frac 14}\left( \exp\left(i  \int_{0}^{r_\ast} (\omega^2 - V_\ell(y))^\frac{1}{2} \d{y}\right) + \epsilon(r_\ast)\right),\\
&u_1^\prime(\omega, r_\ast) =  A {\omega}^\frac{1}{2} i (\omega^2 - V(r_\ast))^{\frac 14} \left[ \exp\left( i \int_0^{r_\ast} (\omega^2 - V_\ell(y))^{\frac{1}{2}}\d y\right) - i \eta(r_\ast) \right. \nonumber \\& \hspace{2cm} \left. - \frac{i V^\prime(r_\ast)}{4 (\omega^2 - V)^{\frac 32}(r_\ast)}\left( \exp\left(i \int_0^{r_\ast} (\omega^2 - V_\ell(y))^\frac{1}{2} \d y\right) + \epsilon(r_\ast)  \right) \right] ,
\end{align}
where \begin{align}|A| =1 , \sup_{r_\ast \in \mathbb{R}}(|\epsilon|(r_\ast)+  |\eta|(r_\ast)) \lesssim 1\label{eq:boundsoneta}\end{align}
and
\begin{align}\label{eq:limits}
\lim_{r_\ast \to -\infty} \eta (r_\ast)= \lim_{r_\ast \to -\infty} 	\epsilon(r_\ast) = 0.
\end{align}
In particular, this proves
\begin{align}
&	\limsup_{r_\ast \to\infty}|u(r_\ast)| \lesssim 1, \\
& \limsup_{r_\ast \to \infty} |u^\prime(r_\ast)|\lesssim |\omega|,
\end{align}
and uniform bounds on the reflection and transmission coefficients
\begin{align}
\sup_{\omega_0\leq|\omega|, \ell\geq \ell_0} (|\mathfrak R(\omega, \ell)| + |\mathfrak T(\omega,\ell)| )\lesssim 1.
\end{align}
\begin{proof}
We will apply \cref{thm:wkb}. First, we set \begin{align}p=(\omega^2 - V_\ell )\end{align} which is positive and smooth. Then, the o.d.e.\ reads
\begin{align}
	u^{\prime \prime} = -  p u.
\end{align} Now we have to show that $F$ is uniformly bounded on the real line.
Note that we have the following bounds on the potential and its derivatives
\begin{align}
 |V_\ell(r_\ast)|, |V_\ell^\prime(r_\ast)|, |V_\ell^{\prime\prime}(r_\ast)|\lesssim \ell^2 e^{2\kappa_+ r_\ast} \text{ and } \ell^2 e^{2\kappa_+ r_\ast} \lesssim |V_\ell(r_\ast)|\text{ for } r_\ast \leq 0,\label{eq:boundsonV01}\\
|	V_\ell(r_\ast)|,|V_\ell^\prime(r_\ast)|, |V_\ell^{\prime\prime}(r_\ast)|\lesssim \ell^2 e^{2\kappa_- r_\ast}  \text{ and } \ell^2 e^{2\kappa_- r_\ast} \lesssim |	V_\ell(r_\ast)| \text{ for } r_\ast \geq 0.\label{eq:boundsonV02}
\end{align}
Here, we might have to choose $\ell_0(M,Q)$ even larger ($r_+^2 (r_+ - 3r_-) + \ell(\ell+1)>0$, cf.\ \eqref{eq:cmV}) in order to assure the lower bounds on the potential.
 Finally, we can estimate $F$ by
\begin{align}\nonumber
	\sup_{r_\ast \in \mathbb{R}} F(r_\ast) &\leq \left| \int_{-\infty}^{\infty} p^{- \frac 14} \left| \frac{\d^2}{\d x^2} \left(p^{-\frac 14 } \right| \right)\d y\right| \\ \nonumber&= \int_{-\infty}^{\infty} p^{- \frac 14}\left( p^{- \frac 94} {p^\prime}^2 + p^{-\frac 54 } |p^{\prime \prime}| \right)\d y\\ \nonumber&
\lesssim \frac{1}{\ell}\int_{0}^{\infty}\left( \frac{e^{4\kappa_- y} }{(\ell^{-2} + e^{2\kappa_- y})^{\frac 52}} + \frac{e^{2\kappa_- y}}{(\ell^{-2} + e^{2\kappa_- y})^{\frac 32}} \right)\d y\\
& + \frac{1}{\ell} \int_{-\infty}^{0} \left( \frac{e^{4\kappa_+ y} }{(\ell^{-2} + e^{2\kappa_+ y})^{\frac 52}} + \frac{e^{2\kappa_+ y}}{(\ell^{-2} + e^{2\kappa_+ y})^{\frac 32}} \right) \d y,
\end{align}
where we have used the bounds from \eqref{eq:boundsonV01} and \eqref{eq:boundsonV02}.
We shall estimate both terms independently. After a change of variables $y \mapsto \frac{1}{2\kappa_-}\log(y)$, we can estimate the first term by
\begin{align} \nonumber &\frac{1}{\ell}\int_{0}^{\infty}\left( \frac{e^{4\kappa_- y} }{(\ell^{-2} + e^{2\kappa_- y})^{\frac 52}} + \frac{e^{2\kappa_- y}}{(\ell^{-2} + e^{2\kappa_- y})^{\frac 32}} \right)\d y\\ \nonumber  & \lesssim \frac{1}{\ell}
	\int_{0}^{1}\left( \frac{y}{(\ell^{-2} +y)^{\frac 52}} + \frac{1}{(\ell^{-2} + y)^{\frac 32}} \right)\d y\\ \nonumber &\lesssim \ell^2 \int_{0}^1 \frac{\ell^2 y}{(1+\ell^2 y)^{\frac 52} } + \frac{1}{(1+ \ell^2 y)^{\frac 32 }}\d y\\ & \lesssim  \int_0^\infty \frac{y}{(1+y)^{\frac 52}} + \frac{1}{(1+y)^{\frac 32 }} \d y \lesssim 1.
\end{align}
Completely analogously, we get the bound for the second integral. In particular, this shows
\begin{align}
	\sup_{\mathbb{R}} F \lesssim 1.
\end{align}
This implies the bounds on $\eta$ and $\epsilon$ in the statement of the theorem (cf.\  \eqref{eq:boundsoneta}) using  \eqref{eq:boundsonetaolver}.

 The limits in equation \eqref{eq:limits} follow from the fact that $F(r_\ast) \to 0$ as $r_\ast \to -\infty$ by construction.  

 The bound on the reflection and transmission coefficients follows now from 
\begin{align}
	|\mathfrak R| \lesssim \left|\frac{\mathfrak W(u_1,v_1)}{\omega}\right|  \leq \frac{1}{|\omega|} \limsup_{r_\ast\to\infty} \left(| u_1^\prime v_1| +| u_1 v_1^\prime|\right) \lesssim 1
\end{align}
and analogously for $\mathfrak T$. 

Finally, $A$ can be determined from the asymptotic behaviour $u \to  e^{i\omega r_\ast}$ as $r_\ast \to - \infty$ and it is given by 
\begin{align}\nonumber
	A & =\lim_{r_\ast \to - \infty} \exp\left( i \omega r_\ast -i  \int_0^{r_\ast} (\omega^2 - V(y))^{\frac 12 } \d y\right) \\ &=
	\lim_{r_\ast \to - \infty} \exp\left( -i  \int_0^{r_\ast} \left((\omega^2 - V(y))^{\frac 12 } - \omega \right) \d y\right) 
	\end{align}
	which converges since $V$ tends to zero exponentially fast. In particular, this also shows that $|A| =1$. 
\end{proof}
\end{prop}
Finally, \cref{thm:boundednesstrans} is a consequence of \cref{prop:easyregime}, \cref{prop9}, and \cref{prop:reg3}.
\section{Proof of \texorpdfstring{\cref{thm:forwardevolution}}{Theorem 1}: Existence and boundedness of the \texorpdfstring{$T$}{T} energy scattering map}
\label{sec:mainthm}
Having performed the analysis of the radial o.d.e.\ and having in particular proven  uniform boundedness of the transmission coefficient $\mathfrak T$ and the reflection coefficients $\mathfrak R$, we shall prove \cref{thm:forwardevolution} in this section.

\subsection{Density of the domains \texorpdfstring{$\mathcal{D}^T_\Hp$}{DTH} and \texorpdfstring{$\mathcal{D}^T_\Ch$}{CTCH}}
We start by proving that the domains $\mathcal{D}^T_\Hp$ and $\mathcal{D}^T_\Ch$ are dense.
\begin{lemma}\label{lem:lemmadense}
	The domains of the forward and backward evolution $\mathcal{D}_{\Ho}^T$ and  $\mathcal{D}^T_{\Ch}$ are dense in $\mathcal{E}^T_{\Ho}$ and $\mathcal{E}^T_{\Ch}$, respectively. 
	\begin{proof}We will only prove that the domain of the forward evolution is dense since the other claim is analogous.
		
		Recall that by definition $C_c^\infty(\Ho)$ is dense in $\mathcal{E}^T_{\Ho}$. Now, let $\Psi \in C_c^\infty(\Ho)$ be arbitrary and denote by $\psi$ its forward evolution. We will show that we can approximate $\Psi$ with functions of $\mathcal{D}^T_{\Ho}$ arbitrarily well. To do so, fix $r_{red}< r_0<r_+$. Then, using the red-shift effect (see \cref{lem:exponentialdecay} in the appendix) the $N$ energy of $\psi\restriction_{r=r_0}$ will have exponential decay towards $i_+$. Hence, it can be approximated with smooth functions $\phi_n$ of compact support on the hypersurface $r=r_0$ w.r.t.\ the norm induced by the non-degenerate $N$ energy (see \cref{rmk:approxwithcompactly} in the appendix). More precisely, on $\Sigma_{r_0} = \{r=r_0\}$ define a sequence $\phi_n \in C_c^\infty(\Sigma_{r_0}) $ by
		\begin{align} \phi_n(t,\theta,\phi) = \psi\restriction_{r=r_0}(t,\theta,\phi)\chi(n^{-1} t ),
		\end{align}
		where $(\theta,\phi) \in \mathbb{S}^2$ and $\chi\colon \mathbb{R}\to [0,1]$ is smooth with $\operatorname{supp} \chi \subseteq [-2,2]$, $\chi\restriction_{[-1,1]} = 1$. Then, we obtain that $ \int_{\Sigma_{r_0}} J_\mu^N[\psi-\phi_n ] n_{\Sigma_{r_0}}^\mu \d\mathrm{vol} \to 0 $ as $n\to\infty$. By construction, the restriction to the event horizon of the backward evolution, $\Phi_n$ of each $\phi_n$ will lie in $\mathcal{D}_{\Ho}^T$. Finally, we can conclude the proof by applying \cref{lem:estimatetorconst} from the appendix, which yields
		\begin{align}
		\| \Psi - \Phi_n\|_{\mathcal{E}^T_\mathcal{H}}^2 = \int_{\mathcal{H}}  J^T_\mu[\Psi-\Phi_n] T^\mu \d\mathrm{vol}\lesssim \int_{r=r_0} J^N_\mu [\psi- \phi_n] n_{\Sigma_{r_0}}^\mu \d\mathrm{vol} \to 0
		\end{align} 
		as $n\to\infty$. 
	\end{proof}
\end{lemma}
\subsection{Boundedness of the scattering and backward map on \texorpdfstring{$\mathcal{D}^T_\Hp$}{DTH} and \texorpdfstring{$\mathcal{D}^T_\Ch$}{DTCH}}
In the following proposition we shall lift the boundedness of the transmission and reflection coefficients (\cref{thm:boundednesstrans}) to the physical space picture on the dense domains $\mathcal{D}^T_\Hp$ and $\mathcal{D}^T_\Ch$.
	\begin{restatable}{prop}{thmbounds}
		\label{thm:thmbounds}
		Let $\psi$ be a smooth solution to \eqref{eq:linearwave} on $\mathcal{M}_{\mathrm{RN}} $ such that $\psi\restriction_{\Ho}\in \mathcal{D}_{\Ho}^T$ (or equivalently, $\psi\restriction_{\Ch}\in \mathcal{D}_{\Ch}^T$). Then,
		\begin{align}
		\|\psi\restriction_{\Ch_A}\|_{\mathcal{E}^T_{\Ch_A}}^2 + \|\psi\restriction_{\Ch_B} \|_{\mathcal{E}^T_{\Ch_B}}^2 
		\leq B \left( \| \psi\restriction_{\Ho_A}\|_{\mathcal{E}^T_{\Ho_A}}^2 + \|\psi\restriction_{\Ho_B}\|_{\mathcal{E}^T_{\Ho_B}}^2\right)
		\end{align}
		and 
		\begin{align}
		\| \psi\restriction_{\Ho_A}\|_{\mathcal{E}^T_{\Ho_A}}^2 + \|\psi\restriction_{\Ho_B}\|_{\mathcal{E}^T_{\Ho_B}}^2
		\leq \tilde B	 \left( \|\psi\restriction_{\Ch_A}\|_{\mathcal{E}^T_{\Ch_A}}^2 + \|\psi\restriction_{\Ch_B} \|_{\mathcal{E}^T_{\Ch_B}}^2 \right)
		\end{align}
		for constants $B$ and $\tilde B$ only depending on the black hole parameters. 
	\end{restatable}
	\begin{proof}		
		Set $\phi:= T\psi$ and note that $\phi\restriction_{\mathcal H} \in \mathcal{D}^T_{\mathcal H}$  and $\phi$ also solves \eqref{eq:linearwave}. Since $\psi \in \mathcal{D}_{\mathcal H}^T \subset \mathcal{E}^T_{\Hp}$, we have that $\phi\restriction_{\Hp_A} = T \psi\restriction_{\Hp_A} \in L^2(\Hp_A)$ with respect to the unique volume form induced by the normal vector field $T$. Analogously, we also have $\phi\restriction_{\Hp_B} = T \psi\restriction_{\Hp_B} \in L^2(\Hp_B)$. Thus, we can define the Fourier transform on the event horizon with the charts \eqref{eq:ha} and \eqref{eq:hb} as \begin{align}
		a_{\Hp_A}(\omega,\theta,\phi) := \frac{1}{\sqrt{2\pi} } \int_{\mathbb{R} } \phi\restriction_{\Hp_A} (v,\theta,\phi) e^{-i\omega v} \d v  
		\end{align}
		and 
		\begin{align}
		a_{\Hp_B}(\omega,\theta,\phi) := \frac{1}{\sqrt{2\pi} } \int_{\mathbb{R} } \phi\restriction_{\Hp_B} (u,\theta,\phi) e^{i\omega u} \d u.
		\end{align}
	 We can further decompose the Fourier coefficients in spherical harmonics to obtain
		\begin{align}
		a^{\ell,m}_{\Hp_A}(\omega) = \langle Y_{\ell m} ,a_{\mathcal{H}_A} \rangle_{L^2(\mathbb S^2)}
	\text{ and }
		a^{\ell,m}_{\Hp_B}(\omega) = \langle Y_{\ell m} ,a_{\mathcal{H}_B} \rangle_{L^2(\mathbb S^2)}.
		\end{align}
			From Plancherel's theorem, we obtain
			\begin{align}
			&	\| \psi\restriction_{\Hp_A} \|_{\mathcal{E}^T_{\Ho_A}}^2 = \sum_{|m|\leq \ell, \ell \geq 0}\int_{\mathbb{R}} |a_{\Hp_A}^{\ell,m} (\omega) |^2 \d \omega,\label{eq:psiT1} \\
			&	\| \psi\restriction_{\Hp_B} \|_{\mathcal{E}^T_{\Hp_b}}^2 = \sum_{|m|\leq \ell, \ell \geq 0}\int_{\mathbb{R}} |a_{\Hp_B}^{\ell,m} (\omega) |^2 \d \omega.\label{eq:psiT2}
			\end{align}
		Similarly, since $\phi\restriction_{\Ch} \in \mathcal{D}^T_{\Ch}$, we define
		\begin{align}
		b_{\Ch_A}(\omega,\theta,\phi) := \frac{1}{\sqrt{2\pi} } \int_{\mathbb{R} } \phi\restriction_{\Ch_A} (v,\theta,\phi) e^{-i\omega v} \d v  
		\end{align}
		and 
		\begin{align}
	b_{\Ch_B}(\omega,\theta,\phi) := \frac{1}{\sqrt{2\pi} } \int_{\mathbb{R} } \phi\restriction_{\Ch_B} (u,\theta,\phi) e^{i\omega u} \d u.
		\end{align}
		We can further decompose the Fourier coefficients in spherical harmonics to obtain
		\begin{align}
		b^{\ell,m}_{\Ch_A}(\omega) = \langle Y_{\ell m} ,b_{\mathcal{CH}_A} \rangle_{L^2(\mathbb S^2)}
		\text{ and }
		b^{\ell,m}_{\Ch_B}(\omega) = \langle Y_{\ell m} ,b_{\mathcal{CH}_B} \rangle_{L^2(\mathbb S^2)}.
		\end{align}
	Again, in view of Plancherel's theorem
		\begin{align}\label{eq:psiT3}
		&\|\psi\restriction_{\Ch_A} \|^2_{\mathcal{E}^T_{\Ch_A}} = \sum_{|m|\leq \ell, \ell \geq 0} \int_{\mathbb{R}} |b^{\ell,m}_{\Ch_A} (\omega)|^2 \d\omega ,\\
		&\|\psi\restriction_{\Ch_B} \|^2_{\mathcal{E}^T_{\Ch_B}} = \sum_{|m|\leq \ell, \ell \geq 0} \int_{\mathbb{R}} |b^{\ell,m}_{\Ch_B} (\omega)|^2 \d\omega .\label{eq:psiT4}
		\end{align}
		and similarly for $\Ch_B$. 
		We shall also decompose $\phi$ on a constant $r$ slice. Fix $r\in(r_-,r_+)$, then set 
		\begin{align} 
		\hat \phi_{m \ell }(\omega,r)  = \frac{1}{\sqrt{2\pi }}\int_{\mathbb{R}} \int_{\mathbb{S}^2} Y_{m\ell }(\theta,\phi) \phi(t,r,\theta,\phi) e^{-i\omega t} \sin\theta \d\theta \d\phi \d t
		\end{align}
		such that 
		\begin{align}
		\phi(t,r,\theta,\phi) = \frac{1}{\sqrt{2\pi}} \sum_{|m|\leq \ell, \ell \geq 0 } \int_{\mathbb{R}} \hat{\phi}_{m\ell} ( \omega, r) Y_{m\ell} (\theta,\phi) e^{i\omega t} \d \omega .
		\end{align}
		This is well-defined since $\phi(t,r,\theta,\phi)$ is compactly supported on each $r=const.$ slice.
		
		Since $\phi$ is smooth, we also know that $\hat{\phi}_{m\ell}$ satisfies the radial o.d.e.\ \eqref{eq:radialode} and  can be expanded as
		\begin{align}
		\hat{\phi}_{m\ell} (\omega,r(r_\ast)) = \alpha^{ \ell, m }_{\Ho_A} (\omega) \frac{r_+}{r} u_1(\omega, r_\ast) +  \alpha^{\ell, m}_{\Ho_B} (\omega) \frac{r_+}{r} u_2(\omega, r_\ast), \label{eq:psihatrplus}
		\end{align}
		where \begin{align} &| u_1 - e^{i\omega r_\ast} | \lesssim_\ell e^{2  \kappa_+ r_\ast} \sim (r_+ - r),
		\\ &| u_2 - e^{ - i\omega r_\ast} | \lesssim_\ell e^{2  \kappa_+ r_\ast} \sim (r_+ - r) 
		\end{align}
		for $r_\ast \leq 0$. Note that this holds uniformly in $\omega$. We shall show in the following that indeed $\alpha_{\Ho_A}^{\ell,m} = a_{\Ho_A}^{\ell,m}$ and $\alpha_{\Ho_B}^{\ell,m} = a_{\Ho_B}^{\ell,m}$. To do so, note that for $r(r_\ast)$ with $r_\ast \leq 0$ we have for fixed $(m,\ell)$ that
		\begin{align}\nonumber
		\phi^{\ell,m}(t,r) = \langle \phi, Y_{m\ell} \rangle_{L^2(\mathbb S^2)} = \int_{\mathbb{R}} &\left(  \alpha^{\ell,m}_{\Ho_A} (\omega) \frac{r_+}{r} u_1(\omega, r_\ast(r))  +  \alpha^{\ell,m}_{\Ho_B} (\omega) \frac{r_+}{r} u_2(\omega, r_\ast(r)) \right) e^{i\omega t}  \frac{ \d \omega }{\sqrt{2\pi}}.
		\end{align}
		We want to interchange the limit $r\to r_+$ with the integral. In order to use Lebesgue's dominated convergence theorem we will estimate $ \alpha_{\Ho_A}^{\ell,m}$ and $ \alpha_{\Ho_B}^{\ell,m}$. Note that
		\begin{align}
			| \alpha_{\Ho_A}^{\ell,m}| = \left| \frac{\mathfrak W(\frac{r}{r_+} \hat \phi_{m\ell}, u_2)}{\mathfrak W(u_1, u_2)} \right| =  \left| \frac{\mathfrak W( \frac{r}{r_+} \hat {T\psi_{m\ell}}, u_2)}{\mathfrak W(u_1, u_2)} \right|\leq   \frac{|\omega \mathfrak W( \frac{r}{r_+} \hat \psi_{m\ell} , u_2)|}{2|\omega|} \leq \left|\mathfrak W\left(\frac{r}{r_+} \hat \psi_{m\ell}, u_2\right)\right|,
		\end{align}
		which is independent of $r(r_\ast)$ and integrable since $\omega\mapsto \hat \psi_{m\ell}(\omega,r_\ast)$ is a Schwartz function. 
		Now, we shall fix $v= r_\ast+ t$ and let $r\to r_+$ such that $r_\ast \to -\infty$. Then, using Lebesgue's dominated convergence theorem, we obtain
		\begin{align}\nonumber
		\phi^{\ell,m} = \int_{\mathbb{R}} \left(  \alpha^{\ell,m}_{\Ho_A} (\omega) e^{i\omega v} + \alpha^{\ell,m}_{\Ho_B} (\omega) e^{-2i \omega r_\ast}e^{i \omega v } \right)  \frac{ \d \omega }{\sqrt{2\pi}}  + O(r_+ - r)
		\end{align}
		as $r\to r_+$. 
		Finally, for $v$ fixed and letting $r\to r_+$ (or $r_\ast \to -\infty$), we obtain
		\begin{align}
		\phi^{\ell,m}\restriction_{\Ho_A} (v) =  \int_{\mathbb{R}}   \alpha^{\ell,m}_{\Ho_A} (\omega) e^{i\omega v}  \frac{ \d \omega }{\sqrt{2\pi}} 
		\end{align}
		in view of the Riemann--Lebesgue lemma.
		Also, by definition of $a^{\ell, m}_{\Ho_A}$, 
		\begin{align}
			\phi\restriction_{\Ho_A} (v,\theta,\phi)= \sum_{|m|\leq \ell, \ell \geq 0} \int_{\mathbb{R}}   a^{\ell,m}_{\Ho_A}(\omega, \theta, \phi) e^{i\omega v } Y_{\ell m}(\theta, \phi) \frac{\d v}{\sqrt{2\pi}}.
		\end{align}
		In view of the Fourier inversion theorem and the fact that the spherical harmonics form a basis we conclude that \begin{align}\label{eq:aequalsalpha} \alpha_{\Ho_A}^{\ell,m} =  a_{\Ho_A}^{\ell,m} \text{  and analogously, }\alpha_{\Ho_B}^{\ell,m} = a_{\Ho_B}^{\ell,m}.\end{align}
		
		Similarly to \eqref{eq:psihatrplus}, we can expand $\hat \psi_{m\ell}$ in a fundamental pair of solutions corresponding to both Cauchy horizons $\Ch_A$ and $\Ch_B$. In particular, we can write
		\begin{align}
		\hat{\phi}_{m\ell} (\omega,r(r_\ast)) = \beta^{\ell,m}_{\Ch_A} (\omega) \frac{r_+}{r} v_1(\omega, r_\ast) +   \beta^{ \ell,m}_{\Ch_A} (\omega)  \frac{r_+}{r} v_2(\omega, r_\ast), 
		\end{align}
		where 
		\begin{align} &| v_1 - e^{-i\omega r_\ast} | \lesssim_\ell e^{2  \kappa_- r_\ast} \sim (r- r_-) ,
		\\ &| v_2 - e^{ i\omega r_\ast} | \lesssim_\ell e^{2  \kappa_- r_\ast} \sim (r- r_-).
		\end{align}
		for $r_\ast \geq 0$. Similarly to \eqref{eq:aequalsalpha}, we can prove
		\begin{align}
\frac{r_+}{r_-}\beta^{\ell,m}_{\Ch_A} (\omega)= b^{\ell,m}_{\Ch_A} (\omega) \text{ and } \frac{r_+}{r_-}\beta^{\ell,m}_{\Ch_B} (\omega) = b^{\ell,m}_{\Ch_B} (\omega).
		\end{align}
		Moreover, from the uniform boundedness of the reflection and transmission coefficients (cf. \cref{thm:boundednesstrans}) we have the estimate
		\begin{align}
| b^{\ell,m}_{\Ch_A} (\omega)| + | b^{\ell,m}_{\Ch_B} (\omega)| &= 	\frac{r_+}{r_-}	| \beta^{\ell,m}_{\Ch_A} (\omega)| + \frac{r_+}{r_-}| \beta^{\ell,m}_{\Ch_B} (\omega) |  =  \frac{r_+}{r_-}\left(\left| \mathfrak R \alpha_{\Hp_A}^{\ell,m} +\bar{ \mathfrak T}\alpha_{\Hp_B}^{\ell,m}\right| + \left| \bar{\mathfrak{R}}\alpha_{\Hp_B}^{\ell,m} + \mathfrak T \alpha_{\Hp_A}^{\ell,m} \right| \right) \nonumber \\&\leq C ( | \alpha^{\ell,m}_{\Ho_A} (\omega)| + |\alpha^{\ell,m}_{\Ho_B} (\omega)| ) = C( | a^{\ell,m}_{\Ho_A} (\omega)| + |a^{\ell,m}_{\Ho_B} (\omega)|) \label{eq:estimatebeta}
		\end{align}
		for a constant $C$ which only depends on the black hole parameters. Here, we have used the fact that
		\begin{align}\label{eq:scatteringcoef1}
			\begin{pmatrix}\beta^{\ell,m}_{\Ch_B} \\ 
			\beta^{\ell,m}_{\Ch_A} 
			\end{pmatrix} = \begin{pmatrix}
			\mathfrak T & \bar{\mathfrak R}\\\mathfrak R & \bar{\mathfrak T}
			\end{pmatrix}
			\begin{pmatrix}
			\alpha^{\ell,m}_{\Hp_A} \\ 	\alpha^{\ell,m}_{\Hp_B}.
			\end{pmatrix}.
		\end{align}
		In view of $1=|\mathfrak T|^2 - |\mathfrak R|^2$, we also have
		\begin{align}\label{eq:scatteringcoef2}
			\begin{pmatrix}
			\alpha^{\ell,m}_{\Hp_A} \\ 	\alpha^{\ell,m}_{\Hp_B}
			\end{pmatrix}	 = \begin{pmatrix}
		\bar{	\mathfrak T} & -\bar{\mathfrak R}\\ -\mathfrak R & {\mathfrak T}
		\end{pmatrix}
		\begin{pmatrix}\beta^{\ell,m}_{\Ch_B} \\ 
		\beta^{\ell,m}_{\Ch_A} 
		\end{pmatrix}
		\end{align}
from which we deduce 
		\begin{align}\label{eq:boundedbackwards}
| a^{\ell,m}_{\Ho_A} (\omega)| + |a^{\ell,m}_{\Ho_B} (\omega)| \lesssim | b^{\ell,m}_{\Ch_A} (\omega)| + | b^{\ell,m}_{\Ch_B} (\omega)|.
		\end{align}
		Estimate \eqref{eq:estimatebeta} and \eqref{eq:boundedbackwards} show the claim in view of \eqref{eq:psiT1}, \eqref{eq:psiT2}, \eqref{eq:psiT3}, and \eqref{eq:psiT4}. Finally, in view of the Fourier inversion theorem, note that the previous also justifies the Fourier representation of scattering map \eqref{eq:scatteringfourier}, and the Fourier representations \eqref{eq:fourierrep1} and \eqref{eq:fourierrep2}. 
	\end{proof}
	\subsection{Completing the proof}
Having proven \cref{lem:lemmadense} and \cref{thm:thmbounds}, we can finally show \cref{thm:forwardevolution} in the following.
\begin{proof}[Proof of \cref{thm:forwardevolution}]  Since $\mathcal{D}^T_{\Hp}\subset \mathcal{E}^T_\Hp$ is dense (\cref{lem:lemmadense}) and $S_0^T\colon \mathcal{D}^T_{\Hp}\subset \mathcal{E}^T_\Hp \to \mathcal{D}^T_{\Ch}\subset \mathcal{E}^T_\Ch$ is a bounded injective map (\cref{rmk:injectiv}, \cref{thm:thmbounds}), we can uniquely extend $S_0^T$ to the bounded injective scattering map 
	\begin{align}
	S^T\colon \mathcal{E}^T_\Hp \to \mathcal{E}^T_\Ch.
	\end{align}
	
	Analogously, in view of \cref{thm:time}, \cref{rmk:domain}, \cref{rmk:injectiv}, and \cref{thm:thmbounds}, we can uniquely extend the bounded injective map $B_0^T\colon \mathcal{D}^T_\Ch \subset \mathcal{E}^T_\Ch \to \mathcal{D}^T_\Ch\subset \mathcal{E}^T_\Hp$ to the bounded injective backward map $B^T\colon \mathcal{E}^T_\Ch \to \mathcal{E}^T_\Hp$ (\cref{lem:lemmadense}).	
	
	Since $B_0^T\circ	S_0^T = \mathrm{Id}_{ \mathcal{D}^T_{\Ho} }$ and $S_0^T \circ B_0^T = \mathrm{Id}_{ \mathcal{D}^T_{\Ch}}$ on dense sets, it also extends to $\mathcal{E}^T_\Hp$ and $\mathcal{E}^T_\Ch$ from which \eqref{eq:inverses} follows.
	Similarly, it suffices to check  \eqref{eq:pseudounitary} for $\psi \in \mathcal{D}_{\Ho}^T$. Indeed, \eqref{eq:pseudounitary} holds true for $\psi \in \mathcal{D}_{\Ho}^T$ in view of the $T$ energy identity. 
\end{proof}
\section{Proof of \texorpdfstring{\cref{thm:cosmological}}{Theorem 6}: Breakdown of \texorpdfstring{$T$}{T} energy scattering for cosmological constants \texorpdfstring{$\Lambda\neq 0$}{Lambdaneq0}}
\label{sec:cosmo}
In the presence of a cosmological constant $\Lambda$, the situation regarding the $T$ energy scattering problem is changed radically. In this section we will consider the subextremal (anti-) de Sitter--Reissner--Nordström black hole interior $(\mathcal{M}_{\text{(a)dSRN}} , g_{Q,M,\Lambda})$ which is completely analogous to $(\mathcal{M}_{\text{RN}} , g_{Q,M}$). We will assume that $(M,Q,\Lambda) \in \mathcal{P}_{\mathrm{se}}$ as defined in \cref{subsec:nonex}. 
 Also, recall that in the presence of a cosmological constant it is natural to look at the Klein--Gordon equation
\begin{align}\label{eq:kg}
\Box_g \psi - \mu \psi = 0
\end{align} with mass $\mu = \frac{3}{2} \Lambda$ for the conformal invariant equation or more general $\mu = \nu \Lambda$ for fixed $\nu \in \mathbb R$. 

 This section is devoted to prove \cref{thm:cosmological} which relies on the fact that solutions of the corresponding radial o.d.e.\ in the vanishing frequency limit $\omega=0$ generically map bounded solutions at $r_\ast = -\infty$ to unbounded solutions at $r_\ast  = +\infty$.  More precisely, for $\Lambda \neq 0$ we obtain---after separation of variables for \eqref{eq:kg} and setting $\d r_\ast = h^{-1} \d r$---the o.d.e.\
\begin{align}
\label{eq:ODElambda}
	-u^{\prime \prime } + V_{\ell,\Lambda} u = \omega^2 u
\end{align}
 for $u(r_\ast) = r(r_\ast) R(r_\ast)$, where
\begin{align}\label{eq:potentiallam}
V_{\ell,\Lambda} = h \left(\frac{ h h^\prime}{r} + \frac{\ell(\ell+1)}{r^2} - \mu \right) =  h \left(\frac{ \frac{\d h}{\d r}}{r} + \frac{\ell(\ell+1)}{r^2} - \mu \right)
\end{align}
and \begin{align}\label{eq:hlambda}h= \frac{\Delta}{r^2} =  1- \frac{2M}{r}- \frac 13 \Lambda r^2 + \frac{Q^2}{r^2}.\end{align}
Here, consider $r(r_\ast)$ as a function $r_\ast$ and recall that $~^\prime$ denotes the derivative with respect to $r_\ast$. The presence of the mass and the cosmological constant leads to a modification of the potential $V_{\ell,\Lambda}$. 

	Nevertheless, the potential $V_{\ell,\Lambda}$ still decays exponentially at $\pm \infty$ and we can define  asymptotic states $u_1^{(\Lambda)},u_2^{(\Lambda)}$, and $ v_1^{(\Lambda)},v_2^{(\Lambda)}$ for $\omega\neq 0$ and $\tilde u_1^{(\Lambda)}$, $\tilde u_2^{(\Lambda)}$, and $\tilde v_1^{(\Lambda)},\tilde v_2^{(\Lambda)}$ for $\omega =0$ just as in the case where $\Lambda=\mu =0$ in \cref{defn:u1u2}.
	In particular, $\tilde u_1^{(\Lambda)}$ and $\tilde v_1^{(\Lambda)}$ remain bounded as $r_\ast \to -\infty$ and $r_\ast \to + \infty$, respectively. In contrast to that, $\tilde u_2^{(\Lambda)}$ and $\tilde v_2^{(\Lambda)}$ grow linearly in their respective limits. The next proposition states that in the presence of a cosmological constant, solutions to \eqref{eq:kg} in the case $\omega=0$ which are bounded at $r_\ast =-\infty$ do not need to be bounded at $r_\ast = +\infty$. 
	\begin{prop}\label{prop:Bneq0} Fix $\nu \in \mathbb{R}$ (e.g.\ $\nu = \frac{3}{2}$ for the conformal invariant mass) and fix subextremal black hole parameters  $(M,Q,\Lambda) \in \mathcal{P}_{\mathrm{se}}$. Assume moreover that $(M,Q,\Lambda) \notin D(\nu)$, where $D(\nu)\subset \mathcal{P}_{\mathrm{se}}$ is defined in the proof and has measure zero. Then, there exists an $\ell_0 = \ell_0(\nu) \in \mathbb{N}_0$ such that we have \begin{align}\label{eq:decompu1inlambda}\tilde u_1^{(\Lambda)} =A(\ell_0 ,\Lambda,M,Q) \tilde v_1^{(\Lambda)} + B(\ell_0,\Lambda,M,Q) \tilde v_2^{(\Lambda)},
		\end{align}
		with $B = B(\ell_0, \Lambda,M,Q) \neq 0$.  
Moreover, $\mathcal{P}_{\mathrm{se}}^{\Lambda =0}\subset D(\nu)$ for all $\nu \in \mathbb R$ and there exists an open subset $U$ with $\mathcal{P}_{\mathrm{se}}^{\Lambda =0} \subset U\subset \mathcal{P}_{\mathrm{se}}$ and  $\mathcal{P}_{\mathrm{se}} \cap U = \mathcal{P}_{\mathrm{se}}^{\Lambda =0}$.
		\begin{proof} Let $\nu \in \mathbb R$ be fixed.
			In the case $\Lambda = 0$ we can represent $\tilde u_1$ with Legendre polynomials and in particular we have that $B(\ell,\Lambda=0,M,Q ) =0$ for all $\ell$ and $0<|Q|<M$. 
			Note that we can write $B$ as 
			\begin{align}\label{eq:defnb}
				B(\Lambda,\ell,M,Q) = \frac{\mathfrak W (\tilde v_2^{(\Lambda)}, \tilde u_1^{(\Lambda)})}{\mathfrak W(\tilde v_1^{(\Lambda)}, \tilde v_2^{(\Lambda)})} = \mathfrak W (\tilde v_2^{(\Lambda)}, \tilde u_1^{(\Lambda)})
			\end{align}
			for all $\Lambda$ such that $(M,Q,\Lambda) \in \mathcal{P}_{\mathrm{se}}$.
				\paragraph{Step 1: \texorpdfstring{$\mathcal{P}_{\mathrm{se}}\subset \mathbb{R}^3$ is open and has two connected components where either $Q>0$ or $Q<0$}{Pse is connected}} For the sake of completeness we will give a proof of \emph{Step 1}, although this seems a quite well-known fact.
			Note that $\mathcal{P}_{\mathrm{se}} = \mathcal{P}_{\mathrm{se}}^{\Lambda >0} \cup\mathcal{P}_{\mathrm{se}}^{\Lambda <0} \cup\mathcal{P}_{\mathrm{se}}^{\Lambda =0} $ is open which can be inferred from its definition.
				
				For the second statement, first note that $\{Q=0\}\cap \mathcal{P}_{\mathrm{se}} =\emptyset$. We will now show that $ \{Q>0\}\cap \mathcal{P}_{\mathrm{se}}$ is connected. In \cref{prop:connectedness} in the appendix we show that $\mathcal{P}_{\mathrm{se}}^{\Lambda >0} \cap \{ Q>0\}$ and $\mathcal{P}_{\mathrm{se}}^{\Lambda <0} \cap \{ Q>0\}$ are path-connected.
			To conclude, note that for every $(M_0,Q_0,\Lambda_0 = 0 )\in \mathcal{P}_{\mathrm{se}}^{\Lambda=0}$, there exist paths from $(M_0,Q_0,\Lambda_0)$ to both $(M_0,Q_0,\epsilon) \in \mathcal{P}_{\mathrm{se}}^{\Lambda >0}$ and  $(M_0,Q_0,-\epsilon) \in \mathcal{P}_{\mathrm{se}}^{\Lambda <0}$ for some $\epsilon(M_0,Q_0)>0$. Together with the fact that $\mathcal{P}_{\mathrm{se}}^{\Lambda=0}\cap \{ Q >0\}$ is path-connected, this shows that  $ \{Q>0\}\cap \mathcal{P}_{\mathrm{se}}$ is path-connected and similarly that $ \{Q<0\}\cap \mathcal{P}_{\mathrm{se}}$ is path-connected which proves the claim.
			\paragraph{Step 2: \texorpdfstring{$\mathcal{P}_{\mathrm{se}}\ni (M,Q,\Lambda) \mapsto B(\ell,\Lambda,M,Q)$ is real analytic}{B is holomorphic}} 
To show \emph{Step 2} we first express \eqref{eq:decompu1inlambda} in $r$ coordinates.  Note that for $(M,Q,\Lambda)\in \mathcal{P}_{\mathrm{se}}$ equation \eqref{eq:decompu1inlambda} is equivalent to 
\begin{align}\label{eq:p=ptilde}
	\frac{r_+}{r_-}  (-1)^\ell P_\ell^{(\Lambda)} (x(r)) =  A(\ell,\Lambda) \tilde P^{(\Lambda)}_\ell (x(r)) +   B(\ell,\Lambda) \tilde Q^{(\Lambda)}_\ell (x(r)),
\end{align} 
where  $r\in (r_-, r_+)$,
\begin{align}
x(r):= - \frac{2r}{r_+ - r_-} + \frac{r_+ + r_-}{r_+ - r_-},
\end{align}
\begin{align}
r(x) =- \frac{r_+ - r_-}{2} x+ \frac{r_+ + r_-}{2}
\end{align}
and $0<r_-<r_+$. Now, note that $\mathcal{P}_{\mathrm{se}} \ni (M,Q,\Lambda)\mapsto r_-$ and $\mathcal{P}_{\mathrm{se}}\ni (M,Q,\Lambda) \mapsto r_+$ are real analytic. Moreover, we can write $\Delta = (r - r_-) (r-r_+) p(r)$ for a second order polynomial $p(r)$, where $\mathcal{P}_{\mathrm{se}} \ni \Lambda \mapsto p(r) $ is also real analytic for fixed $r$. Now, $P_\ell^{(\Lambda)}$, $\tilde P_\ell^{(\Lambda)}$ and $ \tilde Q_\ell^{(\Lambda)}$ appearing in \eqref{eq:p=ptilde} are defined as the unique solutions of 
\begin{align}\label{eq:legendreperturbed}
\frac{\d }{\d x} \left( (1-x^2) p(r(x)) \frac{\d R}{\d x}\right)  + \ell (\ell+1) R - r(x)^2 \nu \Lambda R=0
\end{align}
satisfying
\begin{align}\label{eq:propertiesfirst}
	&P_\ell^{(\Lambda)} = (-1)^\ell +O_\ell(1+x) \text{ as } x\to -1,\\
	& \frac{\d P_\ell^{(\Lambda)}}{\d x} = O_\ell(1) \text{ as } x\to-1,\\
		&\tilde P_\ell^{(\Lambda)} = 1 +O_\ell(1-x) \text{ as } x\to 1,\\
		& \frac{\d \tilde P_\ell^{(\Lambda)}}{\d x} = O_\ell(1) \text{ as } x\to1,\\
			&\tilde Q_\ell^{(\Lambda)} = - \frac 12 \log(1-x) + O_\ell(1) \text{ as } x\to 1,\\
		& \frac{\d \tilde Q_\ell^{(\Lambda)}}{\d x} = \frac{1}{2(1-x)} + O_\ell( (1-x) \log(1-x)) \text{ as } x\to 1.\label{eq:propertieslast}
\end{align}
Note that \eqref{eq:legendreperturbed} depends real analytically on $(M,Q,\Lambda) \in \mathcal{P}_{\mathrm{se}}$ such that $ P_\ell^{(\Lambda)}(x)$, $\tilde P_\ell^{(\Lambda)}(x)$, $\tilde Q_\ell^{(\Lambda)}(x)$ are real analytic functions of $(M,Q,\Lambda) \in \mathcal{P}_{\mathrm{se}}$ for $x\in (-1,1)$. Hence, $\mathcal{P}_{\mathrm{se}} \ni (M,Q,\Lambda) \mapsto B(\ell, \Lambda,M,Q)$ is real analytic.
\paragraph{Step 3: \texorpdfstring{$B(\ell_0(\nu),\Lambda,M,Q)$}{Bdlambda} only vanishes on a set \texorpdfstring{$D(\nu) \subset \mathcal{P}_{\mathrm{se}}$}{D} of measure zero.} The claim follows from 
\begin{align}\label{eq:dbdlam}\left.\frac{\partial B(\ell,\Lambda,M_0,Q_0)}{\partial \Lambda}\right|_{\Lambda = 0}\neq 0\end{align}
for some $0<|Q_0|<M_0$. Throughout \emph{Step 2} we fix $0<|Q_0|<M_0$ and avoid writing their explicit dependence. 
First note that that for $\Lambda =0$ we obtain the Legendre functions of first and second kind, i.e.\ $P_\ell^{(0)} = \tilde P_\ell^{(0)} = P_\ell $ and $\tilde Q_\ell^{(0)} = Q_\ell$ and $B(0,\ell) =0$. 
Now, define coefficients $\tilde A(\ell,\Lambda)$ and $\tilde B(\ell,\Lambda)$ to satisfy
 \begin{align}
 	\label{eq:expansion11}
 	P_\ell^{(\Lambda)} = \tilde A(\ell,\Lambda) \tilde  P_\ell^{(\Lambda)} + \tilde B (\ell,\Lambda)\tilde Q_\ell^{(\Lambda)},
 \end{align}
 and note that \eqref{eq:dbdlam} is equivalent (use that $B(\ell,0) =  \tilde B(\ell,0) = 0 $) to 
 \begin{align}
 	\frac{\partial \tilde B(\ell,\Lambda)}{\partial \Lambda}\Big\vert_{\Lambda =0} \neq 0 .
 \end{align}
 By construction, $P_\ell^{(\Lambda)}$ solves \eqref{eq:legendreperturbed}. Multiplying 
\begin{align}
	\frac{\d }{\d x} \left( (1-x^2) p(r(x)) \frac{\d P_\ell^{(\Lambda)}}{\d x}\right)  + \ell (\ell+1) P_\ell^{(\Lambda)} - r(x)^2 \nu \Lambda P_\ell^{(\Lambda)}=0
\end{align}
by $P_\ell^{(0)}$ and integrating from $x=-1$ to $x=1$ yields
\begin{align}
	0= 	\int_{-1}^{1}P_\ell^{(0)} \left( 	\frac{\d }{\d x} \left( (1-x^2) p(r(x)) \frac{\d P_\ell^{(\Lambda)}}{\d x}\right)  + \ell (\ell+1) P_\ell^{(\Lambda)} - r(x)^2\nu \Lambda P\ell^{(\Lambda)} \right)\d{x}.
\end{align}
Using the expansion \eqref{eq:expansion11} and the properties \eqref{eq:propertiesfirst} -- \eqref{eq:propertieslast} at the end points $x=-1$ and $x=1$ gives after an integration by parts
\begin{align}0= &	\int_{-1}^{1}P_\ell^{(\Lambda)} \left( 	\frac{\d }{\d x} \left( (1-x^2) p(r(x)) \frac{\d P_\ell^{(0)}}{\d x}\right)  + \ell (\ell+1) P_\ell^{(0)} - r(x)^2\nu \Lambda  P_\ell^{(0)}\right)\d{x} + p(r(1)) \tilde B(\ell,\Lambda). \end{align}

Now, taking $\partial_\Lambda \big\vert_{\Lambda =0}$ and integrating by parts once again yields
\begin{align}
p(r(1))	\partial_{\Lambda} \big\vert_{\Lambda =0} \tilde B (\ell, \Lambda) & =  \int_{-1}^1 \left[ \left| \frac{\d P_\ell^{(0)} }{\d x}\right|^2 (1-x^2) \partial_{\Lambda} \big\vert_{\Lambda =0} ( p(r(x)) )+ \left|P_\ell^{(0)}\right|^2\partial_{\Lambda}\big\vert_{\Lambda =0} ( \nu r(x)^2 \Lambda)\right] \d x \nonumber \\ & =  \int_{-1}^1 \left[ \left| \frac{\d P_\ell^{(0)} }{\d x}\right|^2 (1-x^2) \partial_{\Lambda} \big\vert_{\Lambda =0}( p(r(x)) )+\nu \left|P_\ell^{(0)}\right|^2   r(x)^2\vert_{\Lambda =0} \right] \d x . 
\end{align}
Recall that we are in the subextremal range which guarantees that $p(r(1))\neq 0$. We will now distinguish two cases, $\nu =0$ and $\nu \neq 0$. 

\paragraph{Part I: $\nu =0$}In the case $\nu =0$ we have
\begin{align}
p(r(1))	\partial_\Lambda\vert_{\Lambda =0} \tilde B(\ell,\Lambda) = \partial_\Lambda\vert_{\Lambda =0} \int_{-1}^{1} \left|\frac{\d P_\ell}{\d x}\right|^2 (1-x^2)p(r(x)) \d x
\end{align}
In the case $\nu =0$ we will choose $\ell =1$ such that 
\begin{align*}	p(r(1)) \partial_\Lambda\vert_{\Lambda =0} \tilde B(1,\Lambda) =& \partial_\Lambda\vert_{\Lambda =0} \int_{-1}^{1}  (1-x^2)p(r(x)) \d x \\ = & 
	\partial_\Lambda\vert_{\Lambda =0} \int_{-1}^{1}  -\Delta(r(x)) \frac{4}{(r_+ - r_-)^2} \d x
	\\ = & 
	\partial_\Lambda\vert_{\Lambda =0}\left( \frac{-8}{(r_+ - r_-)^3}  \int_{r_-}^{r_+}  \Delta(r) \d r\right)
	\\= & -8\, \partial_\Lambda\vert_{\Lambda =0} \left(
	\frac{\frac{r_+^3-r_-^3}{3} - M_0 (r_+^2 - r_-^2) +Q_0^2(r_+ - r_-) -\frac {1}{15} \Lambda (r_+^5 - r_-^5) }{(r_+ - r_-)^3}\right)\\ =&\,  
	\frac{8(r_+^5-r_-^5)}{15(r_+-r_-)^3} \Big\vert_{\Lambda =0} \\ & + 8
	\frac{\frac{r_+^3-r_-^3}{3} - M_0 (r_+^2 - r_-^2) +Q_0^2(r_+ - r_-) }{(r_+ - r_-)^5}  (r_+^4 +  r_-^4)\Big\vert_{\Lambda =0}\\ & 
	- \frac{8}{3} \frac{ r_+^6 +  r_-^6 - 2M_0 (r_+^5+r_-^5) + Q_0^2(r_+^4 +  r_-^4)}{(r_+ - r_-)^4}\Big\vert_{\Lambda =0} \\ = &
	\frac{-8}{15} \left( 3 r_+^3 + 3 r_-^2 + 4 r_+ r_- \right)\Big\vert_{\Lambda =0} \\ = & \frac{-8}{15}  \left(6 M_0^2-Q_0^2\right) < - 24 M_0^2 .
\end{align*}
The last step is a long but direct computation using that  $\Delta=r^2 - 2M_0r + Q_0^2 -\frac{\Lambda}{3} r^4$ and  $r_{\pm}\vert_{\Lambda =0} = M_0\pm \sqrt{M_0^2-Q_0^2}$, i.e.\ $Q_0^2 = r_+ r_-\vert_{\Lambda =0}$ and $2M_0 = r_+\vert_{\Lambda =0} + r_-\vert_{\Lambda =0}$. Moreover, in view of the inverse function theorem we have
\begin{align}
	\partial_\Lambda \vert_{\Lambda =0} r_+ =  \frac{ r_+^4{}}{3(r_+-r_-) } \Big\vert_{\Lambda =0}
\end{align}
and 
\begin{align}
	\partial_\Lambda \vert_{\Lambda =0} r_- = - \frac{ r_-^4}{3 (r_+-r_-)}\Big\vert_{\Lambda =0}. 
\end{align}

\paragraph{Part II: $\nu \neq 0$} In this case we choose $\ell=0$ such that $P_\ell^{(0)} = 1$ and $\frac{\d P_\ell^{(0)}}{\d x } = 0$. Hence,
\begin{align}\nonumber 
p(r(1)) \partial_\Lambda\vert_{\Lambda =0} \tilde B(\ell,\Lambda) &= \partial_\Lambda\vert_{\Lambda =0} \int_{-1}^{1} r(x)^2 \nu \Lambda \d x = \nu \partial_{\Lambda}\vert_{\Lambda =0} \int_{-1}^{1} \left( -\frac{r_+ - r_-}{2}x + \frac{r_+ + r_-}{2}\right)^2 \Lambda \d x\\ & = \nu   \left(\frac{1}{6} {(r_+  - r_-)^2}  + \frac{1}{2} (r_+ + r_-)^2 \right)\Big\vert_{\Lambda =0}  \neq 0.
\end{align}

This shows that $\mathcal{P}_{\mathrm{se}} \ni (M,Q,\Lambda) \mapsto B(\ell_0(\nu), M,Q,\Lambda)$ is a non-trivial real analytic function which zero set $D(\nu)$ has zero measure. The proof also shows that $\mathcal{P}_{\mathrm{se}}^{\Lambda =0}\subset D(\nu)$ and that there exists an open set $U\subset \mathcal{P}_{\mathrm{se}}$ with $\mathcal{P}_{\mathrm{se}}^{\Lambda =0} \subset U$  and $D(\nu) \cap U = \mathcal{P}_{\mathrm{se}}^{\Lambda =0}$.
	\end{proof}
\end{prop}
\begin{prop}\label{cor:RTunbounded}Let $\nu \in \mathbb{R}$ be fixed.
Let $\omega\neq 0$, $(M,Q,\Lambda) \in \mathcal{P}_{\mathrm{se}}$, and $\ell \in \mathbb N_0$. Then, define completely analogously to \cref{defn:TandR} transmission and reflection coefficients $\mathfrak T(\omega,\ell,\Lambda)$ and $\mathfrak R(\omega,\ell,\Lambda)$ as the unique coefficients such that
\begin{align}
	 u_1^{(\Lambda)} = \mathfrak T (\omega,\ell,\Lambda)  v_1^{(\Lambda)} + \mathfrak R(\omega,\ell,\Lambda)  v_2^{(\Lambda)}
\end{align}
holds. 

Now, assume further that $(M,Q,\Lambda) \in \mathcal{P}_{\mathrm{se}}\setminus D(\nu)$, where $D(\nu)$ is defined in \cref{prop:Bneq0}. Then, there exists an $ \ell_0 =  \ell_0(\nu)$ such that 
\begin{align}
	\lim_{\omega \to 0} |\mathfrak R(\omega,\ell_0)| = \lim_{\omega \to 0} |\mathfrak T(\omega,\ell_0)| = +\infty.
\end{align}  
This shows that $\mathfrak T$ and $\mathfrak R$ have a simple pole at $\omega=0$.
\begin{proof}Fix $\ell_0 = \ell_0(\nu)$ from \cref{prop:Bneq0} and $(M,Q,\Lambda) \in \mathcal{P}_{\mathrm{se}}$ such that $B(\ell_0,\Lambda,M,Q) \neq 0$. Now, note that the o.d.e.\ implies that $\frac{\d}{\d r_\ast}\text{Im}( \bar u u^\prime)  = 0 $ which shows that $1 = |\mathfrak T|^2  -  |\mathfrak R|^2$. In particular, either $|\mathfrak T|$ and $|\mathfrak R|$ are both bounded or both unbounded as $\omega \to 0$. Also note that as $\omega \to 0$, we have that $u_1^{(\Lambda)}\to \tilde u_1^{(\Lambda )}$ pointwise.

Now, assume for a contradiction that there exists a sequence $\omega_n \to 0$ such that $|\mathfrak T(\omega_n)|$ and $|\mathfrak R(\omega_n)|$ remain bounded. Thus, \begin{align}\nonumber\limsup_{\omega_n \to 0} \|u_1^{(\Lambda)} \|_{L^\infty(\mathbb{R})}\leq& \limsup_{\omega_n \to 0} \|u_1^{(\Lambda)} \|_{L^\infty((-\infty,0))}\\& + \limsup_{\omega_n \to 0} \|\mathfrak R v_1^{(\Lambda)} + \mathfrak T v_2^{(\Lambda)} \|_{L^\infty((0,\infty))} \leq C \end{align} for some constant $C>0$.
 Now, using that $B(\ell_0,\Lambda,M,Q)\neq 0$ in  \cref{prop:Bneq0}, we can choose a $r^\ast_{0} \in \mathbb{R}$ such that $|\tilde u_1^{(\Lambda)}(r^\ast_{0})| > C$ which contradicts the fact that $u_1^{(\Lambda)} \to \tilde u_1^{(\Lambda)}$ pointwise as $\omega_n \to 0$.
\end{proof}
\end{prop}
Finally, this allows us to prove \cref{thm:cosmological} which we restate in the following for the convenience of the reader.
\cosmological*
\begin{proof}Fix $\ell_0 = \ell_0(\nu)$ from \cref{cor:RTunbounded} such that the reflection and transmission coefficients blow up as $\omega \to 0$.
Define a sequence of compactly supported functions $\Psi_n $ on  $\mathcal{H}_A$ by $\Psi_n(v,\theta,\varphi) =  f_n(v) Y_{0\ell}(\theta,\varphi)$, such that $f_n \in C_c^\infty(\mathbb R)$,  \begin{align}\int_{\mathbb R} \omega^2 |\hat f_n (\omega)|^2 \d \omega = 1 \text{ and } \int_{-\frac{1}{n}}^{\frac{1}{n}} \omega^2 |\hat f_n (\omega)|^2 \d \omega \geq \epsilon \int_{\mathbb R} \omega^2 |\hat f_n (\omega)|^2 \d \omega = \epsilon\end{align} 
for some $\epsilon >0$.\footnote{ Such a function can be constructed by setting $f_n(v) := \frac{c}{\sqrt n} f(\frac{v}{n})$ for smooth $f\colon \mathbb R \to [0,1]$ with $\operatorname{supp}(f) \subset [-2,2]$, $f\restriction_{[-1, 1]} =1 $ and some normalization constant $c >0$. Indeed, 
\begin{align}
	\int_{-\frac 1n}^{\frac 1n} \omega^2 |\hat f_n(\omega)|^2 \d \omega = \int_{-\frac 1n}^{\frac 1n} \omega^2 |\sqrt n \hat f(n\omega)|^2\d \omega = \int_{-1}^{1} \omega^2  |\hat f (\omega)|^2 =: \epsilon >0
\end{align} in view of $\hat f (0) = \int_\mathbb R f(v) \d v  >0$.}  Imposing vanishing data on $\Ho_B$, this gives rise to a unique smooth solutions $\psi_n$ up to but excluding the Cauchy horizon.
	Arguments completely analogous to those given in the proof of \cref{thm:thmbounds} show that  
		\begin{align}
		\| \psi_n \restriction_{\mathcal{CH}}\|_{ \mathcal{E}^T_{\Ch}}^2 = \frac{r_+^2}{r_-^2} \int_{\mathbb{R}} \omega^2 (|\mathfrak R(\omega,\ell)|^2 + |\mathfrak T(\omega,\ell)|^2) |\hat f_n(\omega)|^2\d \omega.
		\end{align}
		Thus, \begin{align}
\| \psi_n \restriction_{\mathcal{CH}}\|_{ \mathcal{E}^T_{\Ch}}^2  \geq \frac{r_+^2}{r_-^2} \int_{-\frac 1n}^{\frac 1n} \omega^2 (|\mathfrak R(\omega,\ell)|^2 + |\mathfrak T(\omega,\ell)|^2) |\hat f_n(\omega)|^2\d \omega \geq \epsilon  \,\frac{r_+^2}{r_-^2} \inf_{\omega \in [- \frac 1n, \frac 1n]}\left(  |\mathfrak R|^2 +  |\mathfrak T|^2 \right) .
		\end{align}
Since $|\mathfrak R|,|\mathfrak T|\to\infty$ as $\omega \to 0$, also $\inf_{\omega \in [\frac {1}{2n}, \frac 1n]} |\mathfrak R| \to \infty$ and $\inf_{\omega \in [\frac {1}{2n}, \frac 1n]} |\mathfrak T| \to \infty$ as $n\to\infty$. Thus, as $n \to \infty$, we have 
	\begin{align}
			\| \psi_n \restriction_{\mathcal{CH}}\|_{ \mathcal{E}^T_{\Ch}}^2 \to \infty.
			\end{align}
		\end{proof}
		\section{Proof of \texorpdfstring{\cref{thm:kleingordon}}{Theorem 7}: Breakdown of \texorpdfstring{$T$}{T} energy scattering for the Klein--Gordon equation}
			\label{sec:kleingordonequation}
			In this last section we will prove that for a generic set of Klein--Gordon masses, there does not exist a $T$ scattering theory on the interior of Reissner--Nordstr\"om for the Klein--Gordon equation. For the convenience of the reader, we have restated \cref{thm:kleingordon}. 
			\kleingordon*
			 \begin{proof}
			 	The proof of this statement is easier than and similar to the proof of \cref{thm:cosmological} and the proofs of the propositions leading up to it.  More precisely, similar to \cref{sec:cosmo} we define asymptotic states $\tilde u_1^{(\mu)}$, $\tilde v_1^{(\mu)}$ and $\tilde v_2^{(\mu)}$ and define $A(\ell,\mu)$ and $B(\ell,\mu)$ by
			 	$\tilde u_1^{(\mu)} = A(\ell,\mu) \tilde v_1^{(\mu)} + B(\ell,\mu) \tilde v_2^{(\mu)}$. As in \cref{sec:cosmo}, $\mathbb R \ni \mu \mapsto B(\ell,\mu)$ is real analytic and from the o.d.e.\ $-u^{\prime\prime} + V_{\ell,\mu} u =0$ we obtain 
			 	\begin{align}
\left.	\frac{\partial B(\ell,\mu)}{\partial \mu}\right|_{\mu = 0} = \int_{-\infty}^\infty\left. \frac{\partial V_{\ell,\mu}}{\partial \mu}\right|_{\mu =0} \tilde u_1^2 \d r_\ast,
			 	\end{align}
			 	where 
			 	\begin{align}
V_{\ell,\mu} = h \left(\frac{ h h^\prime}{r} + \frac{\ell(\ell+1)}{r^2} - \mu \right) =  h \left(\frac{ \frac{\d h}{\d r}}{r} + \frac{\ell(\ell+1)}{r^2} - \mu \right)
			 	\end{align}
			 	and \begin{align}h= 1- \frac{2M}{r}+ \frac{Q^2}{r^2}\end{align}
			as in \eqref{eq:h}.	Now, note that 
			 	\begin{align}
			 		\left. \frac{\partial V_{\ell,\mu}}{\partial \mu}\right|_{\mu =0}  = - h >0
			 	\end{align}
			 	which is manifestly positive from which we can infer, by analyticity, that $B(\ell,\mu) \neq 0$ for all $\mu \in \mathbb R \setminus \tilde D$, where $\tilde D = \tilde D(M,Q)\subset \mathbb R$ is a discrete set. This proves the analogous statements to \cref{prop:Bneq0} and \cref{cor:RTunbounded}. The claim of \cref{thm:kleingordon} follows now as in the proof of \cref{thm:cosmological}.
			 \end{proof}
\appendix
\section{Additional lemmata}
\paragraph{\textbf{Energy estimates in the interior}}
\begin{lemma}\label{lem:exponentialdecay}
Let $\Psi \in C_c^\infty(\mathcal H)$ and denote by $\psi$ its evolution in the interior. Then, the non-degenerate $N$ energy of $\Psi$ decays exponentially towards $i_+$ on every $\{r=r_0 \}$ hypersurface for $r_{red}< r_0 < r_+$. Here, $r_{red}$ only depends on the black hole parameters. 
\begin{proof}
This argument is very similar to \cite[Proposition 4.2]{franzen2016boundedness}. We  only prove it for the right component of $i^+$ and clearly only have to look at a neighborhood of $i^+$.   First, recall the existence of the celebrated redshift vector field $N$ satisfying $K^N[\psi] \geq b J^N_\mu[\psi]n_v^\mu$ for $r_+ \geq r\geq r_{red}$, where $n_v$ is the normal to a $v=const.$ hypersurface.\footnote{The normal is fixed by making a choice of a volume form on the null hypersurface}

We set 
\begin{align}
	E(v_0) = \int_{v=v_0, r_{red}\leq r \leq r_+} J^N_\mu n_v^\mu \d\mathrm{vol},
\end{align}
and apply the energy identity with the redshift vector field $N$ in the region $\mathcal{R} = \{ r\in[r_{red}, r_+], v \in [v_0,  v_1] \}$, where $v_0$ is large enough such that $v_0 > \sup\operatorname{supp}(\Psi)$.
This gives in view of the coarea formula that
\begin{align}
	E(v_1) - E(v_0)  + \tilde b \int_{v_0}^{v_1} E(v) \d v  \leq 0 \label{eq:expbound}
\end{align}
for every $v_1 \geq v_0> \sup \text{supp}(\Psi)$. Inequality~\eqref{eq:expbound}, smoothness of $v\mapsto E(v)$ and a further application of the energy identity in the region $\{ v\geq v_0, r_+ \geq r \geq r_{red} \}$ finally shows
\begin{align}
	\int_{v\geq v_0, r=r_{red}}J^N_\mu n_{r}^\mu \d{\text{vol}}\leq C \exp(-\tilde b v_0),
\end{align}where $C$ is a constant depending on $\Psi$. This concludes the proof.
\end{proof}
\end{lemma}
\begin{rmk}\label{rmk:approxwithcompactly}
	By cutting off smoothly we can clearly approximate $\Psi$ on a $\{r=const.\}$ hypersurface with compactly supported functions for any fixed $r\in(r_{red}, r_+)$. 
\end{rmk}
\begin{lemma}\label{lem:estimatetorconst}
Let $\psi$ be a smooth solution of the wave equation on $\mathcal{M}_\mathrm{RN}$ such that its restriction to the event horizon has compact support and let $r_0 \in (r_{red}, r_+)$. Then,
\begin{align}
	\int_{\mathcal{H}} J^T_{\mu} n^\mu \d\mathrm{vol} \lesssim \int_{ \{r=r_0\} } J^N_\mu n^\mu \d \mathrm{vol}.
\end{align}
\begin{proof}
We shall use the vector field $S = r^{-2} \partial_{r_\ast}$. By potentially making $r_{red}$ larger, we can assure that the bulk term $K^S := \nabla^\mu J^{S}_\mu$ of the vector field $S$ has a fixed negative sign in $r_0\in(r_{red},r_+)$. This current is analogous to the current introduced in \cite[par.\ 4.1.3.2]{franzen2016boundedness}. Moreover, applying the energy identity in the region $\mathcal{R}=\{ r_0 \leq r \leq r_+ \}$ and noting that $J^N[\psi]_\mu n^\mu\vert_{r=r_0} \sim J^S[\psi]_\mu n^\mu\vert_{r=r_0}$ as well as $J^T[\psi]_\mu n^\mu \vert_\mathcal{H} \sim J^S[\psi]_\mu n^\mu\vert_\mathcal{H}$ yields 
\begin{align}
\int_{ \{r=r_0\} } J^N[\psi]_\mu n^\mu \d\mathrm{vol} + \int_{\mathcal{R}} K^S \d\mathrm{vol} \gtrsim \int_{\mathcal{H}} J^T_\mu n^\mu \d\mathrm{vol}.
\end{align}
This concludes the proof.
\end{proof}
\end{lemma}
\paragraph{\textbf{Analytic properties of the potential and the scattering coefficients}}
In the following we would like to summarize analytic properties of the potential $V_\ell(r)$ and $u_1$,$u_2$, $v_1$ and $v_2$  as functions of $\omega$. This is similar to parts of \cite{hartle1982crossing}. 

First, however we will show the the exponential decay of the potential $V_\ell$ as $r_\ast \to \pm \infty$.  
\begin{lemma}\label{lem:asymptoticspotential}
	We have \begin{align}
	|\Delta(r_\ast )| \lesssim e^{2k_+ r_\ast} \text{ for } r_\ast \leq 0
	\end{align}
	and 
	\begin{align}\label{eq:decayindelta}
	|\Delta(r_\ast)|\lesssim e^{2k_- r_\ast} \text{ for } r_\ast \geq 0.
	\end{align}
	Moreover, we have 
	\begin{align}\label{eq:delta1}
	|	V_\ell(r_\ast)|, |V_\ell^\prime(r_\ast)|, |V_\ell^{\prime\prime}(r_\ast)| \lesssim  (1+ \ell(\ell+1) ) e^{2k_+ r_\ast} \text{ for } r_\ast \leq 0
	\end{align}
	and
	\begin{align}\label{eq:delta2}
	|V_\ell(r_\ast)|, |V_\ell^\prime(r_\ast)|, |V_\ell^{\prime\prime}(r_\ast)| \lesssim  (1+ \ell(\ell+1) ) e^{2k_- r_\ast} \text{ for } r_\ast \geq 0.
	\end{align}
	\begin{proof}
		Note that 
		\begin{align}
		r_+ - r = \tilde C \left( r - r_-\right)^{\frac{k_-}{k_+}} e^{-2k_+ r } e^{2k_+ r_\ast}
		\end{align}
		for a constant $\tilde C$ only depending on the black hole parameters.
		Thus, for $r_\ast \leq 0 $, we have \begin{align}r_+ - r(r_\ast)  = f(r_\ast) e^{2k_+ r_\ast}\label{eq:fofr}\end{align}
		for a smooth function $f(r_\ast)$, which is uniformly bounded below and above for $r_\ast \leq 0$. Moreover, we have $f'(r_\ast),f''(r_\ast) \to 0$ exponentially fast as $r_\ast \to -\infty$. The estimates \eqref{eq:delta1} and \eqref{eq:delta2} are now straightforward applications of the chain rule and the fact that $\frac{\d r}{\d r_\ast} = \frac{\Delta}{r^2}$ and $\Delta = (r-r_-) (r - r_+)$. 
	\end{proof}
\end{lemma}
\begin{prop}
	The potential $V_\ell$ can be expanded as\begin{align}
		V_\ell(r_\ast) = \sum_{m\in \mathbb N} C_m e^{2\kappa_+ m r_\ast},
	\end{align}
	where $|C_m |\lesssim_\ell e^{-\sigma m}$ for a $\sigma >0$. 
\end{prop} 
\begin{proof}
Define the variable 
\begin{align}z(r):= e^{2\kappa_+ r_\ast(r)} = C e^{2\kappa_+ r} (r_+ - r) (r - r_-)^{\frac{\kappa_+}{\kappa_-}},\end{align}
where $C >0$ is such that $z(\frac{r_+ + r_-}{2}) = 1$. From the inverse function theorem it follows that $V_\ell(z) = V_\ell(r(z))$ can be analytically continued in a neighborhood of $z=0$ and thus, there exists a Taylor expansion around $z=0$ such that
\begin{align}
	V_\ell(z) = \sum_{n=1}^\infty C_m z^m.
\end{align}
Hence,
\begin{align}\label{eq:expansion}
	V_\ell(r_\ast) = \sum_{n=1}^\infty C_m e^{2\kappa_+ m r_\ast },
\end{align}
where 
\begin{align}\nonumber
	C_1 = \left.\frac{\d V_\ell}{\d z} \right\vert_{z=0} &=  \left. \frac{\d V_\ell}{\d r}\right\vert_{r=r_+} \left. \frac{\d r}{\d z}\right\vert_{z=0} \\ & = \frac{ r_+ -  r_- }{r_+^4} \left( r_+^2 (r_+ -3 r_-) + \ell(\ell+1) \right).
	\label{eq:cmV}
\end{align}
Note that the coefficients $C_m$ decay exponentially fast in $m$. To see this, remark that we can re-define $\tilde r_\ast := r_\ast -  \rho$ for some constant $\rho>0$. Similarly to \eqref{eq:expansion}, we expand $V_\ell$ as 
\begin{align}
	V_\ell = \sum_{m=1}^\infty D_m e^{2\kappa_+ m \tilde r_\ast}
\end{align}
which shows $C_m = D_m e^{-2\kappa_+m \rho}$. By analyticity we have $|D_m| \leq |\tilde C|^{m+1}$ for some $\tilde C >0$ and thus,
\begin{align}\label{eq:boundonV}
	|C_m| \lesssim_\ell e^{-\sigma m}
\end{align}
for a fixed $\sigma >0$. 
\end{proof}
\begin{prop}\label{prop:rtboundedcomplex}
Let $\ell\in \mathbb{N}$ be fixed. Then, 
\begin{align}
	\sup_{\{  |\operatorname{Re}(\omega) |> 1  \} } |\mathfrak{R} (\omega,\ell)| + |\mathfrak{T} (\omega,\ell) |\lesssim_\ell 1.\label{eq:boundonrt}
\end{align}
Moreover, $\mathfrak T(\omega,\ell)$ has a pole of order one at $\omega = i\kappa_+$ given that $\ell(\ell+1) \neq r_+^2 ( r_+ - 3 r_-)$.
\begin{proof}
	Recall, that $u_1$ is the unique solution to \begin{align}
		u_1(r_\ast) = e^{i\omega r_\ast} + \int_{-\infty}^{r_\ast} \frac{\sin(\omega ( r_\ast - y))}{\omega} V(y) u_1(y) \d y.
	\end{align}
	In \cite{hartle1982crossing} it is shown that the Volterra iteration has the form
	\begin{align}
		u_1 (r_\ast) = e^{i\omega r_\ast}\left( 1 + \sum_{n=1}^\infty u_1^{(n)}(r_\ast) \right),
	\end{align}
	where
	\begin{align}\label{eq:volterracomplex}
		u_1^{(n)} (r_\ast) = \sum_{\stackrel{m_n\dots m_1 \in \mathbb N}{m_n> \dots > m_1}} C_{m_n - m_{n-1}} C_{m_{n-1} - m_{n-2}} \dots C_{m_1} d_{m_n} \dots d_{m_1} e^{2\kappa_+ m_n r_\ast }
	\end{align}
	with $d_m = -(4 m \kappa_+ ( m\kappa_+ + i \omega))^{-1}$. 
	Note that in view of the bound in \eqref{eq:boundonV} one can check that the Volterra iteration for $u_1$ converges on $\omega\in\mathbb C \setminus \{ im\kappa_+: m \in \mathbb N\}$ and moreover, 
	\begin{align}&\sup_{\{ |\operatorname{Re}(\omega) |> 1  \} } | u_1 (r_\ast= 0) |  \lesssim_\ell 1,\\
	&\sup_{\{ |\operatorname{Re}(\omega) |> 1  \} } | u_1^\prime (r_\ast= 0) |  \lesssim_\ell |\omega| .
	\end{align}
	Analogously, we have that $v_1$ is analytic on $\omega \in \mathbb C \setminus \{ im\kappa_-:m\in \mathbb N\}$ and $v_2$ is analytic on $\omega\in \mathbb C \setminus \{ -im\kappa_-:m\in \mathbb N\}$. Moreover,
	\begin{align}
&\sup_{\{ |\operatorname{Re}(\omega) |> 1  \} } | v_1 (r_\ast= 0) |  \lesssim_\ell 1,\\
&\sup_{\{ |\operatorname{Re}(\omega) |> 1  \} } | v_1^\prime (r_\ast= 0) |  \lesssim_\ell |\omega| .
	\end{align}
	and 
	\begin{align}
	&\sup_{\{ |\operatorname{Re}(\omega) |> 1  \} } | v_2 (r_\ast= 0) |  \lesssim_\ell 1,\\
	&\sup_{\{ |\operatorname{Re}(\omega) |> 1  \} } | v_2^\prime (r_\ast= 0) |  \lesssim_\ell |\omega| .
	\end{align}
	This finally shows \eqref{eq:boundonrt} in view of the definition of the transmission and reflection coefficients $\mathfrak T $ and $\mathfrak R$ using Wronskians, cf. \cref{defn:TandR}.
	
	Now, we prove that $\mathfrak T(\omega,\ell)$ has a pole of order one at $\omega = i \kappa_+$ assuming that $\ell(\ell+1) \neq r_+^2 ( r_+ - 3 r_-)$. First note that
	\begin{align}u_1^{(1)} (r_\ast)= \sum_{m_1 \in \mathbb N} C_{m_1} d_{m_1} e^{2\kappa_+ m_1 r_\ast}\end{align}
	has a pole of order one at $\omega = i \kappa_+$ since $C_1\neq 0$, see \eqref{eq:cmV}.  Since for $n\neq 1$ there is no term of the form $e^{2\kappa r_\ast}$ in \eqref{eq:volterracomplex} as $m_n \geq n$, the pole at $\omega = i\kappa_+$ cannot be canceled by the other terms and must occur in $u_1$. Moreover, this pole of $u_1$ at $\omega = i \kappa_+$ is not of higher order that one since $d_1$ does not occur at higher powers than one in the Volterra iteration. This implies that $\mathfrak T(\omega,\ell)$ has a pole of order one at $\omega= i\kappa_+$. 
\end{proof}
\end{prop}
\paragraph{\textbf{Connectedness of the subextremal parameter range}}
\begin{prop}\label{prop:connectedness}
	Let the subextremal parameter space $\mathcal{P}_{\mathrm{se}}^{\Lambda >0}$  and $\mathcal{P}_{\mathrm{se}}^{\Lambda <0}$ be defined as in \eqref{defn:subextremal>0} and \eqref{defn:subextremal<0}, respectively. 
	Then, $\mathcal{P}_{\mathrm{se}}^{\Lambda >0}  \cap \{ Q>0\}$,  $\mathcal{P}_{\mathrm{se}}^{\Lambda <0}  \cap \{ Q>0\}$, $\mathcal{P}_{\mathrm{se}}^{\Lambda >0}  \cap \{ Q<0\}$ and $\mathcal{P}_{\mathrm{se}}^{\Lambda <0}  \cap \{ Q<0\}$ are path-connected. 
	\begin{proof}
	 The claim follows for $\mathcal{P}_{\mathrm{se}}^{\Lambda >0}  \cap \{ Q>0\} $ and $\mathcal{P}_{\mathrm{se}}^{\Lambda >0}  \cap \{ Q>0\} $ from the following continuous parametrizations
	\begin{align}\nonumber \mathcal{P}_{\mathrm{se}}^{\Lambda >0}  \cap \{ Q>0\} =& \Big\{ (M,Q,\Lambda)\in \mathbb{R} \times \mathbb{R} \times \mathbb{R} \colon \Lambda = 3 (r_+^2 + r_-^2 + r_c^2 + r_+ r_c + r_c r_- +r_+ r_-)^{-1}, \\ &6 M =  \Lambda (r_+ + r_-)(r_+ + r_c ) (r_- + r_c), Q = \left( \frac{\Lambda}{3} (r_+ + r_- + r_c)(r_- r_+ r_c) \right)^{\frac 12} \nonumber   \\ &\text{ for } 0 < r_- < r_+ < r_c \Big\} \label{eq:parametrization1}\end{align}
	and 
	\begin{align}\nonumber \mathcal{P}_{\mathrm{se}}^{\Lambda <0}&  \cap \{ Q>0\} = \Big\{ (M,Q,\Lambda)\in \mathbb{R} \times \mathbb{R} \times \mathbb{R} \colon  \Lambda = 3 \left( \frac{3}{4} \left({r_+ + r_-}\right)^2 - r_+ r_- - \xi_i \right)^{-1} , \\ & 6M = -  \Lambda \left(\frac{1}{4} \left({r_+ + r_-}\right)^2 + \xi_i -r_+ r_-\right) (r_+ + r_-),
	\nonumber  
	Q = \left( - \frac{ \Lambda}{3} r_+ r_- \left( \frac 34 (r_+ + r_-)^2 + \xi_i \right) \right)^{\frac 12},
	\\ & \text{ for } 0 < r_- < r_+ \text{ and } \xi_i > \left(\frac{3}{4}(r_+ + r_-)^2 - r_+ r_-\right)^{\frac 12}   \Big\} \label{eq:parametrization2}\end{align} in view of the fact  that  $\{ 0 < r_- < r_+ < r_c \}$ and $ \{0 < r_- < r_+ , \xi_i >( \frac{3}{4}(r_+ + r_-)^2 - r_+ r_-)^{\frac 12}   \} $  are path-connected as subsets of $\mathbb R^3$. In the following we will show \eqref{eq:parametrization1} and \eqref{eq:parametrization2}.
	
	First, in the case $\Lambda>0$, note that \eqref{eq:parametrization1} follows from comparing coefficients of \begin{align*}\frac{-3}{\Lambda} (r^2 - 2 Mr +Q^2 - \frac 13 \Lambda r^4) = (r-r_-) (r - r_+)(r-r_c) (r - r_0)\end{align*} for $r_0 < 0 < r_- < r_+ < r_c$. Indeed, we obtain $r_0 = - (r_- + r_+ + r_c)$ and \eqref{eq:parametrization1} can be deduced.
	
	In the case $\Lambda <0$, note that $\frac{-3}{\Lambda} (r^2 - 2 Mr +Q^2 - \frac 1 3 \Lambda r^4)$ only has two real roots $0<r_- < r_+$ such that we compare coefficients of \begin{align*}\frac{-3}{\Lambda} (r^2 - 2 Mr +Q^2 - \frac 1 3 \Lambda r^4) = (r - r_-) (r - r_+) (r - \xi) (r-\bar \xi)\end{align*} with $\xi = \xi_r +  i \xi_i$. We obtain $2 \xi_r = -(r_+ + r_-)$ and $\xi_i >\left(\frac{3}{4} (r_+ + r_-)^2 - r_+ r_-\right)^{\frac 12}$ to guarantee $\Lambda < 0$. Now, a direct computation shows \eqref{eq:parametrization2}.
	
	Completely analogously we can show path-connectedness for $\mathcal{P}_{\mathrm{se}}^{\Lambda >0}  \cap \{ Q<0\}$ and $\mathcal{P}_{\mathrm{se}}^{\Lambda <0}  \cap \{ Q<0\}$.
	\end{proof}
\end{prop}
\printbibliography[heading=bibintoc]
\end{document}